\RequirePackage{varioref}
\documentclass[sigconf]{acmart}

\copyrightyear{2019}
\acmYear{2019}
\setcopyright{acmlicensed}
\acmConference[SPAA '19]{31st ACM Symposium on Parallelism in Algorithms and Architectures}{June 22--24, 2019}{Phoenix, AZ, USA}
\acmBooktitle{31st ACM Symposium on Parallelism in Algorithms and Architectures (SPAA '19), June 22--24, 2019, Phoenix, AZ, USA}
\acmPrice{15.00}
\acmDOI{10.1145/3323165.3323200}
\acmISBN{978-1-4503-6184-2/19/06}

\usepackage{todonotes}
\usepackage{graphicx}
\usepackage{letltxmacro}
\usepackage{tikz}
\usepackage{xcolor}
\usepackage[toc,page]{appendix}
\usepackage{setspace}
\usepackage{emptypage}
\usepackage[labelfont=bf,textfont=md]{caption}
\usepackage{subcaption}
\usepackage{float}
\usepackage{amsfonts}
\usepackage{amssymb}
\usepackage{amsthm}
\usepackage{mathtools}
\usepackage{braket}
\usepackage[shortlabels]{enumitem}
\usepackage{algorithm}
\usepackage[noend]{algpseudocode}
\usepackage[utf8]{inputenc}
\usepackage{csquotes}
\usepackage{thmtools}
\usepackage{thm-restate}
\usepackage{nameref}
\usepackage{varioref}
\usepackage{hyperref}
\usepackage{cleveref}

\usetikzlibrary{decorations.pathreplacing,patterns,arrows.meta}
\LetLtxMacro{\oldtodo}{\todo}
\renewcommand{\todo}[2][]{\tikzset{external/export next=false}\oldtodo[{#1}]{#2}}
\LetLtxMacro{\oldmifi}{\missingfigure}
\renewcommand{\missingfigure}[2][]{\tikzset{external/export next=false}\oldmifi[{#1}]{#2}}

\newcommand{\para}[1]{\noindent\textbf{#1}\,\,}

\crefname{appsec}{Appendix}{Appendices}

\makeatletter\newcommand*{\currentname}{\@currentlabelname}\makeatother

\providecommand{\keywords}[1]{\textbf{Keywords:} #1}

\renewcommand{\emptyset}{\text{\O}}

\DeclarePairedDelimiter\ceil{\lceil}{\rceil}
\DeclarePairedDelimiter\floor{\lfloor}{\rfloor}

\newcommand{\problem}[3]{${\normalfont \mathrm{#1\hspace{0.02cm}\textbf{\textbar}\hspace{0.02cm}#2\hspace{0.02cm}\textbf{\textbar}\hspace{0.02cm}#3}}$}

\newcommand{\op}[1]{\operatorname{#1}}
\newcommand{\oh}{\mathcal{O}}
\newcommand{\eps}{\varepsilon}
\newcommand{\naturals}{\mathbb{N}}
\newcommand{\rationals}{\mathbb{Q}}
\newcommand{\cupdot}{\mathbin{\dot{\cup}}}

\newcommand{\Sref}[1]{\para{Step \ref{#1}}}

\DeclareMathOperator{\OPT}{OPT}
\DeclareMathOperator{\chp}{chp}
\DeclareMathOperator{\load}{L}

\DeclareMathOperator{\ILP}{ILP}
\DeclareMathOperator{\ks}{ks}
\DeclareMathOperator{\cks}{cks}

\DeclareMathOperator{\parent}{parent}

\newcommand{\DrawMachines}[1]{
	\edef\DrawMachineWidth{0.0}
	\foreach [count=\i from 0] \machine/\w/\globalopts in {#1} {
		\edef\DrawMachineNewWidth{\DrawMachineWidth+\w}
		\edef\DrawMachineTop{0.0}
		\ifdim \w pt > 0 pt
			\foreach \t/\name/\[\opts\] in \machine {
				\edef\DrawMachineNewTop{\DrawMachineTop+\t*3.2}
				\draw \globalopts \opts (\DrawMachineWidth,\DrawMachineTop) rectangle 		(\DrawMachineNewWidth,\DrawMachineNewTop) node[pos=.5] {\name};
				\xdef\DrawMachineTop{\DrawMachineNewTop}
			}
			\xdef\DrawMachineWidth{\DrawMachineNewWidth}
		\else
			\xdef\DrawMachineWidth{\DrawMachineWidth-\w}
		\fi
	}
}

\newcommand{\problemSplit}{\problem{P}{split, setup\!=\!s_i}{C_{\max}}}
\newcommand{\problemInt}{\problem{P}{setup\!=\!s_i}{C_{\max}}}
\newcommand{\problemPmtn}{\problem{P}{pmtn, setup\!=\!s_i}{C_{\max}}}

\newtheorem{definition}{Definition}
\newtheorem{theorem}{Theorem}
\newtheorem{lemma}{Lemma}

\newtheorem{note}{Note}

\begin{document}
\title{Near-Linear Approximation Algorithms for
	Scheduling Problems with Batch Setup Times}
\titlenote{Research was supported by German Research Foundation (DFG) project JA 612/20-1}
\renewcommand{\shorttitle}{Near-Linear Approximation of
	Scheduling Problems with Batch Setup Times}

\author{Max A. Deppert}
\affiliation{%
  \institution{Kiel University}
  \streetaddress{Christian-Albrechts-Platz 4}
  \city{Kiel}
  \state{Germany}
  \postcode{24118}
}
\email{made@informatik.uni-kiel.de}

\author{Klaus Jansen}
\affiliation{%
  \institution{Kiel University}
  \streetaddress{Christian-Albrechts-Platz 4}
  \city{Kiel}
  \state{Germany}
  \postcode{24118}
}
\email{kj@informatik.uni-kiel.de}

\renewcommand{\shortauthors}{M. A. Deppert and K. Jansen}

\begin{abstract}
	We investigate the scheduling of $n$ jobs divided into $c$ classes on $m$ identical parallel machines. For every class there is a setup time which is required whenever a machine switches from the processing of one class to another class. The objective is to find a schedule that minimizes the makespan. We give near-linear approximation algorithms for the following problem variants: the non-preemptive context where jobs may not be preempted, the preemptive context where jobs may be preempted but not parallelized, as well as the splittable context where jobs may be preempted and parallelized.
	
	We present the first algorithm improving the previously best approximation ratio of $2$ to a better ratio of $3/2$ in the preemptive case. In more detail, for all three flavors we present an approximation ratio $2$ with running time $\mathcal{O}(n)$, ratio $3/2+\varepsilon$ in time $\mathcal{O}(n\log 1/\varepsilon)$ as well as a ratio of $3/2$. The $(3/2)$-approximate algorithms have different running times. In the non-preemptive case we get time $\mathcal{O}(n\log (n+\Delta))$ where $\Delta$ is the largest value of the input. The splittable approximation runs in time $\mathcal{O}(n+c\log(c+m))$ whereas the preemptive algorithm has a running time $\mathcal{O}(n \log (c+m)) \leq \mathcal{O}(n \log n)$. So far, no PTAS is known for the preemptive problem  without restrictions, so we make progress towards that question. Recently Jansen et al. found an EPTAS for the splittable and non-preemptive case but with impractical running times exponential in $1/\varepsilon$.
\end{abstract}

%
%



\maketitle

\section{Introduction}

Scheduling problems with setup times have been intensively studied for over 30 years now; in fact, they allow very natural formulations of scheduling problems.

In the general scheduling problem with setup times, there are $m$ identical and parallel machines, a set $J$ of $n \in \naturals$ jobs $j \in J$, $c \in \naturals$ different classes, a partition $\dot{\bigcup}_{i=1}^c C_i = J$ of $c$ nonempty and disjoint subsets $C_i \subseteq J$, a \emph{processing time} of $t_j \in \naturals$ time units for each job $j \in J$ and a \emph{setup} (or \emph{setup time}) of $s_i \in \naturals$ time units for each class $i \in [c]$. The objective is to find a schedule which minimizes the makespan while holding all of the following.

All jobs (or its complete sets of job pieces) are scheduled.
A setup $s_i$ is scheduled whenever a machine starts processing load of class $i$ and when switching processing from one class to another \emph{different} class on a machine. A setup is \emph{not} required between jobs (or job pieces) of the \emph{same} class.
There are various types of setups discussed; here we focus on \emph{sequence-independent} batch setups, i.e. they do not depend on the previous job/class. All machines are \emph{single-threaded} (jobs (or job pieces) and setups do not intersect in time on each machine) and no setup is preempted.


There are three variants of scheduling problems with setup times which have been gaining the most attention in the past.
There is the \emph{non-preemptive} case where no job may be preempted, formally known\footnote{This is the notation introduced by Graham et al. \cite{GRAHAM1979287}} as problem \problemInt{}. Another variant is the \emph{preemptive} context, namely \problemPmtn{}, where a job may be preempted at any time but be processed on at most one machine at a time, so a job may not be parallelized. In the generous case of \emph{splittable} scheduling, known as \problemSplit{}, a job is allowed to be split into any number of job pieces which may be processed on any machine at any time.
\begin{table*}[h]
	\centering
	\scalebox{1.0}{
		\begin{tabular}{l|l|l|l|}
			\cline{2-4}
			& \multicolumn{2}{c|}{m variable}                                                                                                                                                                                               & \multicolumn{1}{c|}{m fixed} \\ \cline{2-4} 
			& unrestricted                                                                                                                             & \emph{small batches} or $|C_i|=1$ or $P(C_i) \leq \gamma\OPT$                               &                              \\ \hline
			\multicolumn{1}{|l|}{Splittable}     & \begin{tabular}[c]{@{}l@{}}$5/3$ in poly \cite{DBLP:journals/dam/XingZ00}\\ $3/2$ in $\oh(n+c\log(c+m))$ *\\ EPTAS \cite{DBLP:conf/innovations/JansenKMR19}\end{tabular}                          & \begin{tabular}[c]{@{}l@{}}$\approx\tfrac32$ in $\oh(n+(m+c)\log(m+c))$ \cite{DBLP:journals/siamcomp/Chen93}\end{tabular}                                       & \begin{tabular}[c]{@{}l@{}}FPTAS \cite{DBLP:journals/dam/XingZ00}\end{tabular}                        \\ \hline
			\multicolumn{1}{|l|}{Non-Preemptive} & \begin{tabular}[c]{@{}l@{}}$2+\eps$ in $\oh(n \log 1/\eps)$, PTAS \cite{DBLP:conf/europar/JansenL16}\\ $3/2$ in $\oh(n\log(n+\Delta))$ *\\ EPTAS \cite{DBLP:conf/innovations/JansenKMR19}\end{tabular} & \begin{tabular}[c]{@{}l@{}}$(1+\eps)\min\set{\tfrac32\OPT, \OPT+t_{\max}-1}$  in poly \cite{DBLP:conf/wads/MackerMHR15}\end{tabular}                             & \begin{tabular}[c]{@{}l@{}}FPTAS \cite{DBLP:conf/wads/MackerMHR15}\end{tabular}                 \\ \hline
			\multicolumn{1}{|l|}{Preemptive}     & \begin{tabular}[c]{@{}l@{}}$(2-(\floor{m/2}+1)^{-1})$ in $\oh(n)$ \cite{DBLP:journals/ior/MonmaP93}\\ $3/2$ in $\oh(n\log n)$ *\end{tabular}               & \begin{tabular}[c]{@{}l@{}}$4/3+\eps$ in poly \cite{DBLP:conf/soda/SchuurmanW99}\\ EPTAS \cite{DBLP:conf/innovations/JansenKMR19}\end{tabular} & \qquad /                       \\ \hline
		\end{tabular}
	}
	\caption{An overview of known results \hspace{22.2em} * Result is in this paper}
	\label{table:known-results}
\end{table*}

\para{Related results.}
Monma and Potts began their investigation of these problems considering the preemptive case. They found first dynamic programming approaches for various single machine problems \cite{DBLP:journals/ior/MonmaP89} polynomial in $n$ but exponential in $c$. Furthermore, they showed NP-hardness for \problemPmtn{} even if $m = 2$. In a later work \cite{DBLP:journals/ior/MonmaP93} they found a heuristic which resembles McNaughton's preemptive wrap-around rule \cite{McNaughton:1959:SDL:2780402.2780403}. It requires $\oh(n)$ time for being $(2-(\floor{\tfrac{m}2}+1)^{-1})$-approximate. Notice that this ratio is truly greater than $\tfrac32$ if $m \ge 4$ and the asymptotic bound is $2$ for $m \rightarrow \infty$. Monma and Potts also discussed the problem class of \emph{small batches} where for any batch $i$ the sum of one setup time and the total processing time of all jobs in $i$ is smaller than the optimal makespan, i.e. $s_i + \sum_{j \in C_i} t_j \leq \OPT$. Most suitable for this kind of problems, they found a heuristic that first uses list scheduling for complete batches followed by an attempt of splitting some batches so that they are scheduled on two different machines. This second approach needs a running time of $\oh(n+(m+c)\log(m+c))$ and considering only small batches it is $(\tfrac32-\tfrac1{4m-4})$-approximate if $m \leq 4$ whereas it is $(\tfrac53-\tfrac1{m})$-approximate for small batches if $m$ is a multiple of $3$ and $m \ge 6$.

Then Chen \cite{DBLP:journals/siamcomp/Chen93} modified the second approach of Monma and Potts. For small batches Chen improved the heuristic to a worst case guarantee of $\max\set{\tfrac{3m}{2m+1},\tfrac{3m-4}{2m-2}}$ if $m \geq 5$ while the same time of $\oh(n+(m+c)\log(m+c))$ is required.

Schuurman and Woeginger \cite{DBLP:conf/soda/SchuurmanW99} studied the preemptive problem for \emph{single-job-batches}, i.e. $|C_i|=1$. They found a PTAS for the uniform setups problem \problem{P}{pmtn, setup=s}{C_{\max}}. Furthermore, they presented a $(\tfrac43+\eps)$-approximation in case of arbitrary setup times. Both algorithms have a running time linear in $n$ but exponential in $1/\eps$.
Then Chen, Ye, and Zhang \cite{DBLP:journals/dam/XingZ00} turned to the splittable case. Without other restrictions they presented an FPTAS if $m$ is fixed and a $\tfrac53$-approximation in polynomial time if $m$ is variable. They give some simple arguments that the problem is weakly NP-hard if $m$ is fixed and NP-hard in the strong sense otherwise.

More recently Mäcker et al. \cite{DBLP:conf/wads/MackerMHR15} made progress to the case of non-preemptive scheduling. They used the restrictions that all setup times are equal ($s_i = s$) and the total processing time of each class is bounded by $\gamma\OPT$ for some constant $\gamma$, i.e. $\sum_{j \in C_i} t_j \leq \gamma\OPT$. Mäcker et al. found a simple $2$-approximation, an FPTAS for fixed $m$, and a $(1+\eps)\min\set{\tfrac32\OPT,\OPT + t_{\max}-1}$-approximation (where $t_{\max} = \max_{j \in J} t_j$) in polynomial time if $m$ is variable. Therefore, this especially yields a PTAS for unit processing times $t_j = 1$.

Jansen and Land \cite{DBLP:conf/europar/JansenL16} found three different algorithms for the non-preemptive context without restrictions. They presented an approximation ratio $3$ using a next-fit strategy running in time $\oh(n)$, a $2$-dual approximation running in time $\oh(n)$ which leads to a $(2+\eps)$-approximation running in time $\oh(n \log (\tfrac1{\eps}))$, as well as a PTAS.
Recently Jansen et al. \cite{DBLP:conf/innovations/JansenKMR19} found an EPTAS for all three problem variants. For the preemptive case they assume $|C_i|=1$. They make use of n-fold integer programs, which can be solved using the algorithm by Hemmecke, Onn, and Romanchuk. However, even after some runtime improvement the runtime for the splittable model is $2^{\oh(1/\eps^4\log^6(1/\eps))}n^4\log(m)$, for example. These algorithms are interesting answers to the question of complexity but they are useless for solving actual problems in practice. Therefore the design of \emph{fast} (and especially polynomial) approximation algorithms with small approximation ratio remains interesting.

\para{Our Contribution.} For all three problem variants we give a $2$-approximate algorithm running in time $\oh(n)$ as well as a $(\tfrac32+\eps)$-approximation with running time $\oh(n\log(\tfrac1{\eps}))$. With some runtime improvements we present some very efficient near-linear approximation algorithms with a constant approximation ratio equal to $\tfrac32$. In detail, we find a $\tfrac32$-approximation for the splittable case with running time $\oh(n + c\log(c+m)) \leq \oh(n\log(c+m))$. Also we will see a $\tfrac32$-approximate algorithm for the non-preemptive case that runs in time $\oh(n \log(T_{\min}))$ where $T_{\min} = \max\set{\tfrac1m N, \max_{i\in[c]}(s_i+t_{\max}^{(i)})}$, $t_{\max}^{(i)} = \max_{j \in C_i} t_j$ and $N = \sum_{i=1}^c s_i + \sum_{j \in J} t_j$. For the most complicated case of these three problem contexts, the preemptive case, we study a $\tfrac32$-approximation running in time $\oh(n\log(c+m)) \leq \oh(n \log n)$. Especially this last result is interesting; we make progress to the general case where classes may consist of an \emph{arbitrary} number of jobs. The best approximation ratio was the one by Monma and Potts \cite{DBLP:journals/ior/MonmaP93} mentioned above. All other previously known results for preemptive scheduling used restrictions like \emph{small batches} or even \emph{single-job-batches}, i.e. $|C_i|=1$  (cf. \Cref{table:known-results}). As a byproduct we give some new \emph{dual} lower bounds.

\para{Algorithmic Ideas.} The $\tfrac32$-approximate algorithm for the preemptive case is our main result. It is highly related to the right partitioning of classes and jobs into different sizes; in fact, the right partition allows us to reduce the problem to a fine-grained knapsack instance. To achieve the truly constant bounds in the splittable and preemptive case \emph{while} speeding up the algorithm we use a technique that we call \emph{Class Jumping} (see \Cref{class-jumping-splittable,improvedsearch}). However, we also make extensive use of a simple idea that we name \emph{Batch Wrapping} (see \Cref{principle:mcnaughton}).

\section{Preliminaries}

\para{Notation.}
Natural numbers are truly greater than zero, i.e. $\naturals = \set{1,2,3,\dots}$. The set of all natural numbers from $1$ to $k \in \naturals$ is $[k] := \set{l \in \naturals| 1 \leq l \leq k}$. The \emph{load} of a machine $u \in [m]$ in a schedule $\sigma$ is $\load_{\sigma}(u)$ (or simply $\load(u)$). This is the sum of all setup times and the processing times of all jobs (or job pieces) scheduled on machine $u$. The processing time of a set of jobs $K$ is $P(K) := \sum_{j \in K} t_j$. The jobs of a set of classes $X \subseteq [c]$ are $J(X) := \bigcup_{i \in X} C_i$. A \emph{job piece} of a job $j \in J$ is a (new) job $j'$ with a processing time $t_{j'} \leq t_j$.  Whenever we split a job $j \in C_i$ of a class $i \in [c]$ into two new job pieces $j_1,j_2$, we understand these jobs to be jobs of class $i$ as well - although $j_1,j_2 \in C_i$ does not hold formally.

\para{Properties.}
For later purposes we need to split the classes into expensive classes and cheap classes as follows.
Let $T>0$ be a makespan. We say that a class $i \in [c]$ is \emph{expensive} if $s_i > \tfrac12T$ and we call it \emph{cheap} if $s_i \leq \tfrac12T$. We define $I_{\exp} \subseteq [c]$ as the set of all expensive classes and $I_{\chp} \subseteq [c]$ as the set of all cheap classes such that $I_{\exp} \cupdot I_{\chp} = [c]$.
We denote the total load of a feasible schedule $\sigma$ by $\load(\sigma) = \sum_{u=1}^m \load_{\sigma}(u) = \sum_{i = 1}^c (\lambda_i^{\sigma} s_i + P(C_i))$ for some setup multiplicities $\lambda_i^{\sigma} \in \naturals$ with $\lambda^{\sigma}_i \leq |C_i|$.  For any instance $I$ it is true that $\OPT(I) \leq N := \sum_{i \in [c]} s_i + \sum_{j \in J} t_j$ (all jobs on one machine) as well as $\OPT(I) > s_{\max}$ and $\OPT(I) \geq \tfrac1mN$ and therefore $\OPT(I) \geq \max\set{\tfrac1mN,s_{\max}}$.

An important value to our observations will be the minimal number of machines to schedule all jobs of an expensive class. 	In the following we give two simple lemmas to find this minimal machines numbers. Therefore, for all classes $i \in [c]$ let
$\alpha_i := \ceil*{P(C_i)/(T - s_i)}$ and $\beta_i := \ceil*{2P(C_i)/T}$.

\begin{lemma}\label{lemma:a_i-lower-bound}
	Given a feasible schedule $\sigma$ with makespan $T$ and load $\load(\sigma) = \sum_{i=1}^c (\lambda_i^{\sigma}s_i + P(C_i))$, it is true that $\lambda_i^{\sigma} \geq \alpha_i \geq 1$. Furthermore, $i \in I_{\exp}$ implies that $\lambda_i^{\sigma} \geq \alpha_i \geq \beta_i \geq 1$ and $\sigma$ needs at least $\lambda^{\sigma}_i$ different machines to place all jobs in $C_i$.
\end{lemma}
\begin{proof}
	Apparently $\alpha_i \geq 1$ and $\beta_i \geq 1$ are direct results for all $i \in [c]$. There must be at least one initial setup time $s_i$ to schedule any jobs of class $i$ on a machine. Since setups may not be split, there is a processing time of at most $T - s_i$ per machine to schedule jobs of $C_i$ and therefore, $\sigma$ needs at least $\alpha_i = \ceil{P(C_i)/(T - s_i)} \leq \lambda_i^{\sigma}$ setups to schedule all jobs of $C_i$. If $i \in I_{\exp}$ we have $s_i > \tfrac12T$ such that there cannot be two expensive setups on one machine and
	$\alpha_i = \ceil{P(C_i)/(T - s_i)} \geq \ceil{P(C_i)/(T - \tfrac12T)} = \ceil{2P(C_i)/T} = \beta_i.$
\end{proof}

\begin{lemma}\label{lemma:machine_number_expensive_classes_multiplicities}
	Let $\sigma$ be a feasible schedule with makespan $T$ for an instance $I$. Then $\sigma$ schedules jobs of different expensive classes on different machines and $m \geq \sum_{i \in I_{\exp}} \lambda_i^{\sigma}$.
\end{lemma}
\begin{proof}
	Assume that $m < \sum_{i \in I_{\exp}} \lambda_i^{\sigma}$. Then setups of two different classes $i_1, i_2 \in I_{\exp}$ must have been scheduled on one machine $u \in [m]$. We obtain $\load_{\sigma}(u) \geq s_{i_1} + s_{i_2} > \tfrac12T + \tfrac12T = T$ since $i_1$ and $i_2$ are expensive. That is a contradiction to the makespan $T$.
\end{proof}

\section{Overview}

Here we give a briefly overview to our results. We start with our general results.

\begin{theorem}
	For all three problems there is a $2$-approximation running in time $\oh(n)$.\qed
\end{theorem}
For the details see \Vref{simple_upper_bounds}. Especially if the reader is not familiar to these problems, the simple $2$-approximations in \Cref{simple_upper_bounds} might be a good point to start.

We use the well-known approach of \emph{dual approximation algorithms}\footnote{A $\rho$-dual approximation algorithm gets the input and a value $T$ and either computes a feasible schedule with makespan at most $\rho T$ or rejects $T$ which then implicates that $T < \OPT$.} introduced by Hochbaum and Shmoys \cite{DBLP:journals/jacm/HochbaumS87} to get the following result.

\begin{theorem}
	For all three problems there is a $(\tfrac32+\eps)$-approximation running in time $\oh(n\log1/\eps)$.
\end{theorem}
Already this result is much stronger for the preemptive case than the previous ratio of $2$ by Monma and Potts. In more detail, we find $\tfrac32$-dual approximations for all three problem variants, all running in time $\oh(n)$. Also in all problem cases there is a value $T_{\min}$ depending only on the input such that $\OPT \in [T_{\min},2T_{\min}]$ due to the $2$-approximations. So a binary search suffices.
In the following we briefly describe these dual approximations.

\subsection{Preemptive Scheduling}

Also in the setup context preemptive scheduling means that each job may be preempted at any time, \emph{but} it is allowed to be processed on at most \emph{one} machine at a time. In other words, a job may not run in parallel time. So, this is a job-constraint only; in fact, the load of a class may be processed in parallel but not the jobs themselves.

\begin{note}\label{preemptive_simple_lower_bound}
	$\OPT_{\op{pmtn}} \geq \max_{i \in [c]} (s_i + t_{\max}^{(i)})$ where $t_{\max}^{(i)} = \max_{j \in C_i} t_j$.
\end{note}
\begin{proof}
	Let $\sigma$ be a feasible schedule for an instance $I$ with a makespan $T$ and consider a job $j \in C_i$ of a class $i \in [c]$. There may be $k$ job pieces $j_1,\dots,j_k$ of job $j$ with a total processing time of $\sum_{l=1}^k t_{j_l}$. Let $p_l$ be the point in time $\sigma$ starts to schedule job piece $j_l$. Without loss of generality we assume $p_1 \leq \dots \leq p_k$. So the execution of job $j$ ends at time $e_j = p_k + t_{j_k}$. Now remark that $p_l \leq p_{l+1}$ means that $p_l + t_{j_l} \leq p_{l+1}$ since otherwise $j_l$ and $j_{l+1}$ run in parallel time. It follows that $p_1 + \sum_{l=1}^{k-1} t_{j_l} \leq p_k$, which means $p_1 + t_j \leq p_k + t_{j_k} = e_j$. There is at least one setup $s_i$ before time $p_1$, i.e. $p_1 \geq s_i$, and we obtain $T \geq e_j \geq p_1 + t_j \geq s_i + t_j$.
\end{proof}

\noindent Due to this, we assume that $m < n$ in the preemptive case, because $m \geq n$ leads to a trivial optimal solution by simply scheduling one job (and a setup) per machine.

The preemptive case appears to be very natural on the one hand but hard to approximate (for arbitrary large batches) on the other hand. Aiming for the ratio of $\tfrac32$, we managed to reduce the problem to a knapsack problem efficiently solvable as a \emph{continuous} knapsack problem. Therefore, we need to take a closer look on $I_{\exp}$ and $I_{\chp}$ so we split them again.
We divide the expensive classes into three disjoint subsets $I_{\exp}^+$, $I_{\exp}^0$ and $I_{\exp}^-$ such that $i \in I_{\exp}$ holds $i \in I_{\exp}^+$ iff. $T \leq s_i + P(C_i)$, $i \in I_{\exp}^0$ iff. $\tfrac34T < s_i + P(C_i) < T$ and $i \in I_{\exp}^-$ iff. $s_i + P(C_i) \leq \tfrac34T$.
Also we divide the cheap classes into $I_{\chp}^+$, $I_{\chp}^-$ s.t. $i \in I_{\chp}$ holds $i \in I_{\chp}^+$ iff. $\tfrac14T \leq s_i \leq \tfrac12T$ and $i \in I_{\chp}^-$ iff. $s_i < \tfrac14T$.

\begin{restatable*}[Nice Instances]{definition}{restateniceinstances}
	For a makespan $T$ we call an instance \emph{nice} if $I_{\exp}^0$ is empty.
\end{restatable*}
\noindent The next theorem yields a $\tfrac32$-dual approximation for nice instances and will be important to find a $\tfrac32$-ratio for general instances too.

\begin{restatable*}{theorem}{restatepreemptivesimple}\label{preemptive:simple}
	Let $I$ be a nice instance for a makespan $T$. Moreover, let
	\[L_{\op{nice}} = P(J) + \sum_{i \in I_{\exp}^+} \alpha_i' s_i + \sum_{i \in I_{\exp}^- \cup I_{\chp}} s_i\] and $m_{\op{nice}} = \ceil{\tfrac12|I_{\exp}^-|} + \sum_{i \in I_{\exp}^+} \alpha_i'$ where $\alpha_i' = \floor*{\tfrac{P(C_i)}{T-s_i}}$. Then the following properties hold.
	\begin{enumerate}[label=(\roman*)]
		\item If $mT < L_{\op{nice}}$ or $m < m_{\op{nice}}$, it is true that $T < \OPT_{\op{pmtn}}(I)$.
		\label{preemptive:simple:decision:T_check_false}
		\item Otherwise a feasible schedule with makespan at most $\tfrac32T$ can be computed in time $\oh(n)$.
		\label{preemptive:simple:decision:T_check_true}
	\end{enumerate}
\end{restatable*}
\noindent It turns out that nice instances are some sort of well-behaving instances which can be handled very easily \emph{and} actually their definition is helpful for general instances too.

The motivation behind a general algorithm is the following. Obviously jobs of different expensive classes can not be placed on a common machine in a \emph{$T$-feasible} schedule (a feasible schedule with a makespan of at most $T$). Especially the jobs of $J(I_{\exp}^0)$ and $J(I_{\exp}^+ \cup I_{\exp}^-)$ cannot. So we first place the classes of $I_{\exp}^0$ on one machine per class, which is reasonable as we will see later (cf. \Vref{preemptive:large_machines}). They obviously fit on a single machine, since $\tfrac34T < s_i + P(C_i) < T$ for all $i \in I_{\exp}^0$. Each of these \emph{large} machines got free processing time less than $\tfrac14T$ in a $T$-feasible schedule. After that we decide which jobs of cheap classes will get processing time on the large machines or get processed as part of a \emph{nice} instance with the residual load that is scheduled on the residual $m-|I_{\exp}^0|$ machines. Apparently only jobs of $I_{\chp}^- \subseteq I_{\chp}$ can actually be processed on large machines in a $T$-feasible schedule, because the setups of other cheap classes have a size of at least $\tfrac14T$ so we only need to decide about this set. We will find a fine-grained knapsack instance on an appropriate subset for this decision. See \Cref{preemptive-algorithm} for the details.

\subsection{Splittable Scheduling}

In case of the splittable problem, jobs are allowed to be preempted at any time and all jobs (or job pieces) can be placed on any machine at any time. Especially jobs are allowed to be processed in parallel time (on different machines). It is important to notice that one should \emph{not} assume $n \geq m$ in the splittable case, since increasing the number of machines may result in a lower (optimal) makespan; in fact, every optimal schedule makes use of \emph{all} $m$ machines. Due to this, it is remarkable that we allow a weaker definition of schedules in the following manner. A schedule may consist of machine configurations with associated multiplicities instead of (for example) explicitly mapping each job (piece) $j$ to a pair $(u_j,x_j) \in [m]\times \rationals$ where $u_j$ is the machine on which $j$ starts processing at time $x_j$.

\begin{restatable*}{theorem}{restatelemmasplittablealgorithm}
	\label{lemma:splittable:algorithm}
	Let $I$ be an instance and let $T$ be a makespan. Let
	\[L_{\op{split}} = P(J) + \sum_{i \in I_{\chp}} s_i + \sum_{i \in I_{\exp}} \beta_i s_i\]
	and $m_{\exp} = \sum_{i \in I_{\exp}} \beta_i$. Then the following properties hold.
	\begin{enumerate}[label=(\roman*)]
		\item If $mT < L_{\op{split}}$ or $m < m_{\exp}$, then it is true that $T < \OPT_{\op{split}}(I)$.
		\label{lemma:splittable:algorithm:T_check_false}
		\item Otherwise a feasible schedule with makespan at most $\tfrac32T$ can be computed in time $\oh(n)$.
		\label{lemma:splittable:algorithm:T_check_true}
	\end{enumerate}
\end{restatable*}
The idea of the algorithm is rather simple. We schedule the expensive classes by using as few setups as possible (imagining an optimal makespan, i.e. $T = \OPT(I)$). An optimal schedule needs at least $\alpha_i$ setups/machines to schedule a class $i \in I_{\exp}$, but we will only use $\beta_i \leq \alpha_i$ setups/machines (cf. \Cref{lemma:a_i-lower-bound}). For each expensive class $i$ we may get at most one machine $\bar{u}_i$ with a load $\load(\bar{u}_i) < T$. So we can reserve the time interval of $[\load(\bar{u}_i),\load(\bar{u}_i)+\tfrac12T]$ for a cheap setup on these machines before filling the residual time of $T-\load(\bar{u}_i)$ with load of cheap classes, since $\load(\bar{u}_i) + \tfrac12T + (T-\load(\bar{u}_i)) = \tfrac32T$. Once all machines are filled up, we turn to unused machines and wrap between time $\tfrac12T$ and $\tfrac32T$ such that cheap setups can be placed below line $\tfrac12T$.
\Cref{fig:alg_splittable_exp,fig:alg_splittable_exp_chp} illustrate an example situation after step (1) and (2) with green colored wrap templates. See \Vref{section:splittable} for the details.

\begin{figure}[h]
	\begin{subfigure}{0.49\textwidth}
		\centering
		\begin{tikzpicture}
  \usetikzlibrary{patterns}

\def\width{0.48}

\DrawMachines{
{
    {.60}/$s_1$/,
{0.06}//,
{0.03}//,
{0.07}//,
{0.09}//,
{0.1}//,
{0.08}//,
{0.06}//,
{0.01}//
}/\width/[fill=green!30!white],
{
    {.60}/$s_1$/[fill=lightgray],
    {0.01}//,
{0.13}//,
{0.06}//,
{0.02}//,
{0.07}//,
{0.09}//,
{0.12}//
}/\width/[fill=green!30!white],
{
    {.60}/$s_1$/[fill=lightgray],
    {0.1}//,
{0.13}//,
{0.07}//,
{0.03}//,
{0.1}//,
{0.03}//,
{0.04}//
}/\width/[fill=green!30!white],
{
    {.60}/$s_1$/[fill=lightgray],
{0.09}//,
{0.08}//,
{0.02}//,
    {.31}//[{color=transparent,opacity=0,fill=green,fill opacity=.3}]
}/\width/[fill=green!30!white],
{
    {.90}/$s_2$/,
{0.12}//,
{0.03}//,
{0.04}//,
{0.03}//,
{0.08}//,
{0.09}//,
{0.11}//
}/\width/[fill=green!30!white],
{
    {.90}/$s_2$/[fill=lightgray],
{0.09}//,
{0.1}//,
{0.08}//,
{0.02}//,
{0.13}//,
{0.06}//,
{0.02}//
}/\width/[fill=green!30!white],
{
    {.90}/$s_2$/[fill=lightgray],
{0.06}//,
{0.1}//,
{0.06}//,
{0.28}//[{color=transparent,opacity=0,fill=green,fill opacity=.3}]
}/\width/[fill=green!30!white],
{
    {.55}/$s_3$/,
{0.08}//,
{0.05}//,
{0.37}//[{color=transparent,opacity=0,fill=green,fill opacity=.3}]
}/\width/[fill=green!30!white],
{
    {.70}/$s_4$/,
{0.08}//,
{0.05}//,
{0.15}//,
{0.02}//,
{0.04}//,
{0.07}//,
{0.09}//
}/\width/[fill=green!30!white],
{
    {.70}/$s_4$/[fill=lightgray],
{0.08}//,
{0.07}//,
{0.12}//,
{0.01}//,
{0.07}//,
{0.05}//,
{0.1}//[{color=transparent,opacity=0,fill=green,fill opacity=.3}]
}/\width/[fill=green!30!white]
}

\draw (-0.2,0) -- (15*\width+0.2, 0);
\draw (-0.2,1.6) -- (15*\width+0.2, 1.6) [dotted];
\draw (-0.2,3.2) -- (15*\width+0.2, 3.2) [dotted];
\draw (-0.2,4.8) -- (15*\width+0.2, 4.8) node [right,color=white] {$\tfrac32T$};

\draw [<->] (0,-0.4) -- (15*\width,-0.4) node[fill=white, pos=.5] {$m$};

\end{tikzpicture}

		\caption{Situation after step (1)}
		\label{fig:alg_splittable_exp}
	\end{subfigure}
	\begin{subfigure}{0.5\textwidth}
		\centering
		\begin{tikzpicture}
  \usetikzlibrary{patterns}

\def\width{0.48}

\DrawMachines{
{
    {.60}/$s_1$/,
{0.06}//,
{0.03}//,
{0.07}//,
{0.09}//,
{0.1}//,
{0.08}//,
{0.06}//,
{0.01}//
}/\width/[fill=lightgray],
{
    {.60}/$s_1$/[fill=lightgray],
    {0.01}//,
{0.13}//,
{0.06}//,
{0.02}//,
{0.07}//,
{0.09}//,
{0.12}//
}/\width/[fill=lightgray],
{
    {.60}/$s_1$/[fill=lightgray],
    {0.1}//,
{0.13}//,
{0.07}//,
{0.03}//,
{0.1}//,
{0.03}//,
{0.04}//
}/\width/[fill=lightgray],
{
    {.60}/$s_1$/[fill=lightgray],
{0.09}//,
{0.08}//,
{0.02}//,
{0.5}//[opacity=0],
{0.12}/$s_5$/[fill=green!30!white],
    {.05}//[fill=green!30!white],
    {.04}//[fill=green!30!white]
}/\width/[fill=lightgray],
{
    {.90}/$s_2$/,
{0.12}//,
{0.03}//,
{0.04}//,
{0.03}//,
{0.08}//,
{0.09}//,
{0.11}//
}/\width/[fill=lightgray],
{
    {.90}/$s_2$/[fill=lightgray],
{0.09}//,
{0.1}//,
{0.08}//,
{0.02}//,
{0.13}//,
{0.06}//,
{0.02}//
}/\width/[fill=lightgray],
{
    {.90}/$s_2$/[fill=lightgray],
{0.06}//,
{0.1}//,
{0.06}//
}/\width/[fill=lightgray],
{
    {.55}/$s_3$/,
{0.08}//,
{0.05}//,
{0.38}//[opacity=0],
{0.12}/$s_5$/,
{0.1}//[fill=green!30!white],
{0.03}//[fill=green!30!white],
{0.14}/$s_6$/[fill=green!30!white],
{0.05}//[fill=green!30!white]
}/\width/[fill=lightgray],
{
    {.70}/$s_4$/,
{0.08}//,
{0.05}//,
{0.15}//,
{0.02}//,
{0.04}//,
{0.07}//,
{0.09}//
}/\width/[fill=lightgray],
{
    {.70}/$s_4$/[fill=lightgray],
{0.08}//,
{0.07}//,
{0.12}//,
{0.01}//,
{0.07}//,
{0.05}//
}/\width/[fill=lightgray],
{
    {.36}//[opacity=0],
{0.14}/$s_6$/[fill=lightgray],
{0.17}//,
{0.11}//,
{0.08}//,
{0.21}//,
{0.12}//,
{0.01}//,
{0.14}//,
{0.16}//
}/\width/[fill=green!30!white],
{
    {.36}//[opacity=0],
{0.14}/$s_6$/[fill=lightgray],
{0.04}//,
{0.05}//,
{0.13}//,
{0.12}//,
{0.37}/$s_7$/,
{0.07}//,
{0.02}//,
{0.01}//,
{0.19}//
}/\width/[fill=green!30!white],
{
    {.13}//[opacity=0],
{0.37}/$s_7$/[fill=lightgray],
{0.11}//,
{0.19}//,
{0.07}//,
{0.02}//,
{0.1}//,
{0.09}//,
{0.04}//,
{0.12}//,
{0.23}//,
{0.03}//
}/\width/[fill=green!30!white],
{
    {.13}//[opacity=0],
{0.37}/$s_7$/[fill=lightgray],
{0.09}//,
{0.11}//,
{0.25}/$s_8$/,
{0.04}//,
{0.13}//,
{0.04}//,
{0.07}//,
{0.27}//[{color=transparent,opacity=0,fill=green,fill opacity=.3}]
}/\width/[fill=green!30!white],
{
    {.5}//[opacity=0],
{1}//[{color=transparent,opacity=0,fill=green,fill opacity=.3}]
}/\width/[fill=lightgray]
}

\draw (-0.2,0) -- (15*\width+0.2, 0);
\draw (-0.2,1.6) -- (15*\width+0.2, 1.6) [dotted] node [right] {$\tfrac12T$};
\draw (-0.2,3.2) -- (15*\width+0.2, 3.2) [dotted] node [right] {$T$};
\draw (-0.2,4.8) -- (15*\width+0.2, 4.8) node [right] {$\tfrac32T$};

\draw [<->] (0,-0.4) -- (15*\width,-0.4) node[fill=white, pos=.5] {$m$};

\end{tikzpicture}

		\caption{Situation after step (2)}
		\label{fig:alg_splittable_exp_chp}
	\end{subfigure}
	\caption{An example for the algorithm for the splittable case with $I_{\exp} = \set{1,2,3,4}$ and $I_{\chp} = \set{5,6,7,8}$}
\end{figure}
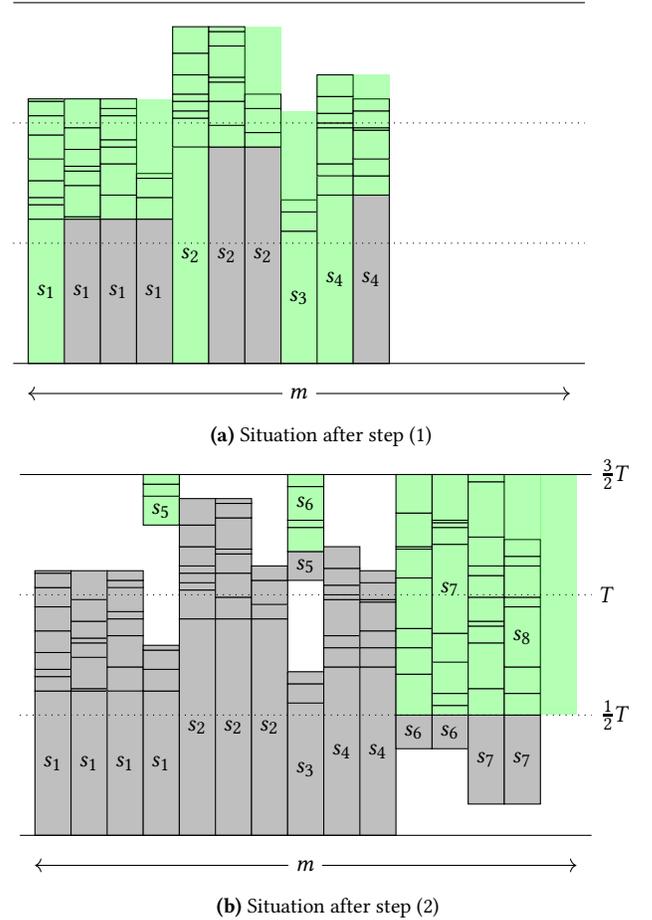
\subsection{Non-Preemptive Scheduling}

Doing \emph{non-preemptive} scheduling we do \emph{not} allow jobs to be preempted. Even an optimal schedule needs to place at least one setup to schedule a job on a machine, so we find another lower bound.

\begin{note}\label{non-preemptive_simple_lower_bound}
	$\OPT_{\op{nonp}} \geq \max_{i \in [c]} (s_i + t_{\max}^{(i)})$ where $t_{\max}^{(i)} = \max_{j \in C_i} t_j$.
\end{note}
Therefore, analogous to preemptive scheduling assume $m < n$. Let $J_+ = \set{j \in J|t_j > \tfrac12T}$ be the big jobs whereas the small jobs be denoted by $J_- = \set{j \in J| t_j \leq \tfrac12T}$. Our algorithm is based on the fact, that there are three subsets of jobs such that pairwise they cannot be scheduled on a single machine. These subsets are $J_+$, $J(I_{\exp})$, and the set $K = \bigcup_{i \in I_{\chp}}\set{j \in C_i \cap J_-| s_i + t_j > \tfrac12T}$.

In \Cref{nonpreemptivescheduling} we find the following minimum number of machines for each class. Let
\[m_i = \begin{cases}
\ceil*{\frac{P(C_i)}{T - s_i}} = \alpha_i & : i \in I_{\exp}\\
|C_i \cap J_+| + \ceil*{\frac{P(C_i \cap K)}{T - s_i}} & : i \in I_{\chp}
\end{cases}\]
for all $i \in [c]$. The following result yields our $\tfrac32$-dual approximation.

\begin{restatable*}{theorem}{restatenonpreemptivealgorithmdecision}
	\label{non-preemptive:algorithm-decision}
	Let $I$ be an instance and let $T$ be a makespan. Let
	\[L_{\op{nonp}} = P(J) + \sum_{i=1}^c m_i s_i + \sum_{i:x_i > 0} s_i\]
	and $m' = \sum_{i=1}^c m_i$ where $x_i = P(C_i) - m_i(T - s_i)$. Then the following properties hold.
	\begin{enumerate}[label=(\roman*)]
		\item If $mT < L_{\op{nonp}}$ or $m < m'$, then it is true that $T < \OPT_{\op{nonp}}(I)$.
		\label{lemma:non-preemptive:algorithm:T_check_false}
		\item Otherwise a feasible schedule with makespan at most $\tfrac32T$ can be computed in time $\oh(n)$.
		\label{lemma:non-preemptive:algorithm:T_check_true}
	\end{enumerate}
\end{restatable*}
\noindent Now binary search leads to a constant approximation:
\begin{restatable*}{theorem}{restatenonpreemptiverunningtime}\label{non-preemptive:32-approx}
	There is a $\tfrac32$-approximation for the non-preemptive case running in time $\oh(n \log (n+\Delta))$ where $\Delta = \max\{s_{\max},t_{\max}\}$ is the largest number of the input.
\end{restatable*}
\begin{proof}
	Unlike the other cases, here the optimal value is an integral number, i.e. $\OPT_{\op{nonp}} \in \naturals$, since all values in the input are integral numbers and neither jobs nor setups are allowed to be preempted. Therefore, a binary search in $[T_{\min},2T_{\min}]$ can find an appropriate makespan in time $\ceil*{\log T_{\min}} \cdot \oh(n) = \oh(n\log T_{\min})$ and it is easy to show that $\log T_{\min} \leq \oh(\log(n+\Delta))$. Given the $\tfrac32$-dual approximation of \Cref{non-preemptive:algorithm-decision} this completes the proof of \Cref{non-preemptive:32-approx}.
\end{proof}
\noindent See \Vref{nonpreemptivescheduling} for all the details.

\subsection{Class Jumping}
\label{class-jumping-splittable}

With a different idea for a binary search routine for an appropriate makespan we are able to improve both the running time and the approximation ratio for the splittable and preemptive case.
As with a single binary search we test makespan guesses with our dual algorithms. The general idea is to look at some points in time which we call \emph{jumps}. A \emph{jump} of an expensive class $i$ is some makespan guess $T$ such that any lower guess $T' < T$ will cause at least one more setup/machine to schedule the jobs of class $i$. The goal is to find two jumps $T_{\op{fail}},T_{\op{ok}}$ of two classes such that there is no jump of any other class between them, while $T_{\op{fail}}$ is rejected and $T_{\op{ok}}$ is accepted. In fact, this means that any makespan $T$ between both jumps causes the same load $L$ (with our dual algorithm). Therefore, either $T_{\op{ok}}$ or $\tfrac1{m}L$ will be an appropriate makespan.

\begin{theorem}\label{splittable:running_time_improvement}
	There is a $\tfrac32$-approximation for the splittable case running in time $\oh(n+c\log(c+m))$.
\end{theorem}

With a small modification the idea can be reused to be applied to the preemptive case as well and this yields our strongest result for the preemptive case.
\begin{restatable*}{theorem}{restatepreemptiverunningtimeimprovement}\label{preemptive:running_time_improvement}
	There is a $\tfrac32$-approximation for the preemptive case running in time $\oh(n \log n)$.
\end{restatable*}
Here we only present the improvement for the splittable case. The improved search for the preemptive case is slightly more complicated and based even more on the details of the $\tfrac32$-dual approximation so we refer to \Cref{improvedsearch} for the details.

The following ideas are crucial. Once the total processing times $P_i = P(C_i)$ are computed, the values $\beta_i$ can be computed in constant time and one can test \Cref{lemma:splittable:algorithm}\ref{lemma:splittable:algorithm:T_check_true} in time $\oh(c)$ for a given makespan $T$. Whenever we test a makespan $T$ we save the computed values $\beta_i$ as $\beta_i(T)$. We will call $(A,B]$ a \emph{right interval} if makespan $B$ satisfies \Cref{lemma:splittable:algorithm}\ref{lemma:splittable:algorithm:T_check_true} ($B$ is \emph{accepted}) while $A$ does not ($A$ is \emph{rejected}). For example $(s_{\max},N]$ is a right interval.

\def\fastClass{f}
\def\slowClass{i}
\def\goodClassA{i_a}
\def\goodClassB{i_b}
\def\mValue{\beta}

\begin{algorithm}[h]
	\caption{Class Jumping for Splittable Scheduling}
	\label{alg:runtime_improvement}
	\begin{enumerate}[1.]
		\item Set $\tilde{s}_0 = 0$ and $\tilde{s}_{c+1} = N$
		\item Sort the setup time values $s_i$ ascending and name them $\tilde{s}_1,\dots,\tilde{s}_c$ in time $\oh(c\log c)$
		\label{alg:runtime_improvement:sort_setup_values}
		\item Compute $P_i = P(C_i)$ for all $i \in [c]$ in time $\oh(n)$
		\label{alg:runtime_improvement:compute_processing_times}
		\item Guess the right makespan interval $V = (2\tilde{s}_{i-1},2\tilde{s}_i]$ in time $\oh(c\log c)$ and set $T_1 := 2\tilde{s}_i$
		\label{alg:runtime_improvement:first_guess}
		\item Find a \emph{fastest jumping class} $\fastClass \in I_{\exp}$, i.e. $P_{\fastClass} \geq P_{\slowClass}$ for all $\slowClass \in I_{\exp}$ in time $\oh(c)$
		\label{alg:runtime_improvement:find_fastest_jumping_class}
		\item Guess the right interval $W = \left(\tfrac{2P_{\fastClass}}{\mValue_{\fastClass}(T_1)+k+1},\tfrac{2P_{\fastClass}}{\mValue_{\fastClass}(T_1)+k}\right]$ for some integer $k < m$ such that $X := V \cap W \neq \emptyset$, in time $\oh(c\log m)$ and set $T_2 := \tfrac{2P_{\fastClass}}{\mValue_{\fastClass}(T_1)+k}$
		\label{alg:runtime_improvement:second_guess}
		\item Find and sort the $\oh(c)$ jumps in $X$ in time $\oh(c\log c)$
		\label{alg:runtime_improvement:find_and_sort_jumps}
		\item Guess the right interval $Y = (T_{\op{fail}},T_{\op{ok}}] \subseteq W$ for two jumps $T_{\op{fail}},T_{\op{ok}}$ of two classes $i_a,i_b \in I_{\exp}^+$ which jump in $X$ such that no \emph{other} class jumps in $Y$, in time $\oh(c\log c)$
		\label{alg:runtime_improvement:third_guess}
		\item Choose a suitable makespan in this interval in constant time
		\label{alg:runtime_improvement:compute_opt}
	\end{enumerate}
\end{algorithm}

We go through the interesting details of \Cref{alg:runtime_improvement} to get the idea of the search routine.

\Sref{alg:runtime_improvement:first_guess}
Note that $T \in [2\tilde{s}_{i-1},2\tilde{s}_i)$ means $\tilde{s}_0 \leq \dots \leq \tilde{s}_{i-2} \leq \tilde{s}_{i-1} \leq \tfrac12T < \tilde{s}_i \leq \tilde{s}_{i+1} \leq \dots \leq \tilde{s}_{c+1}$ and so any makespan in an interval $[2\tilde{s}_{i-1},2\tilde{s}_i)$ causes the \emph{same} partition $I_{\exp} \cupdot I_{\chp}$. The running time can be obtained with binary search.

\Sref{alg:runtime_improvement:find_fastest_jumping_class}
We call $T$ a \emph{jump} of a class $i \in I_{\exp}$ if $2P_i/T$ is integer. That means all machines containing jobs of class $i$ are filled up to line $s_i + \tfrac12T$. So $T$ represents a point in time such that any $T' < T$ will cause at least one more (obligatory) setup to schedule class $i$. It takes a time of $\oh(c)$ to find some class $\fastClass \in I_{\exp}$ with $P_{\fastClass} = \max_{i \in I_{\exp}} P_i$.

\Sref{alg:runtime_improvement:second_guess}
Just remark that $2P_{\fastClass}/(\mValue_{\fastClass}(T_1)+k)$ and $2P_{\fastClass}/(\mValue_{\fastClass}(T_1)+k+1)$ are two consecutive jumps of class $\fastClass$.

\Sref{alg:runtime_improvement:find_and_sort_jumps}
In the analysis we will see that $X$ contains at most $c$ jumps in total (of all classes). To find a jump of a class $\iota \in I_{\exp}$ in $X$ just look at $\mValue_\iota(T_2)$. If $2P_\iota/\mValue_\iota(T_2) < X$ then there is no jump of class $\iota$ in $X$. Otherwise $T_\iota := 2P_\iota/\mValue_\iota(T_2)$ is the only jump of class $\iota$ in $X$.

\Sref{alg:runtime_improvement:compute_opt}
So $T_{\op{ok}}$ was accepted while $T_{\op{fail}}$ got rejected and there are no jumps of any other classes between them. Let $L_{\op{split}}(T_{\op{fail}})$ be the load which is required to place $T_{\op{fail}}$ and set $T_{\op{new}} := \tfrac1m L_{\op{split}}(T_{\op{fail}})$.

\noindent\textbf{Case} $m < m_{\exp}(T_{\op{fail}})$. So the jump causes too many required machines and hence $T < T_{\op{ok}}$ means $T < \OPT$. We return $T_{\op{ok}}$.

\noindent\textbf{Case} $m \geq m_{\exp}(T_{\op{fail}})$. We do another case distinction as follows.

If $T_{\op{new}} \geq T_{\op{ok}}$ we find that $T_{\op{ok}}$ is smaller than the smallest makespan that may be suitable to place $L_{\op{split}}(T_{\op{fail}})$. Therefore, we return $T_{\op{ok}}$.

If $T_{\op{new}} < T_{\op{ok}}$ it follows $T_{\op{fail}} = \tfrac{mT_{\op{fail}}}m < \tfrac1m L_{\op{split}}(T_{\op{fail}}) = T_{\op{new}} < T_{\op{ok}}$ since $m \geq m_{\exp}(T_{\op{fail}})$ and the rejection of $T_{\op{fail}}$ implies that $mT_{\op{fail}} < L_{\op{split}}(T_{\op{fail}})$. So we have $T_{\op{new}} \in (T_{\op{fail}},T_{\op{ok}})$ and we get $m \geq m_{\exp}(T_{\op{fail}}) = m_{\exp}(T_{\op{new}})$ and $mT_{\op{new}} = L_{\op{split}}(T_{\op{fail}}) = L_{\op{split}}(T_{\op{new}})$. Because of \Cref{lemma:splittable:algorithm} we return $T_{\op{new}}$.

The interesting part of the analysis is the fact, that there are no more than $\oh(c)$ jumps in $X$ and we want to show the following lemma.

\begin{lemma}\label{splittable:fastClass}
	If $T'$ is a jump of $\fastClass$, i.e. $T' = 2P_{\fastClass}/\mValue_{\fastClass}(T')$, and $T''$ is a jump of a different class $\slowClass$, i.e. $T'' = 2P_{\slowClass}/\mValue_{\slowClass}(T'')$, such that $T'' \leq T'$, then the next jump of class $i$ is smaller than the next jump of $\fastClass$, which can be written as
	\[
	\frac{2P_{\slowClass}}{\mValue_{\slowClass}(T'')+1} \leq \frac{2P_{\fastClass}}{\mValue_{\fastClass}(T')+1}.
	\]
\end{lemma}
\begin{proof}
	Since $T'' \leq T'$ we have
	\begin{equation}\label{equ:fastVSslow}
	\mValue_{\slowClass}(T'') P_{\fastClass} = \frac{2P_{\slowClass}}{T''} P_{\fastClass} \geq \frac{2P_{\slowClass}}{T'} P_{\fastClass} = \mValue_{\fastClass}(T') P_{\slowClass}
	\end{equation}
	and thus it follows that
	\begin{align*}
	\frac{2P_{\slowClass}}{\mValue_{\slowClass}(T'')+1}
	= \frac{2P_{\fastClass}P_{\slowClass}}{\mValue_{\slowClass}(T'')P_{\fastClass}+P_{\fastClass}}
	&\stackrel{\eqref{equ:fastVSslow}}{\leq} \frac{2P_{\fastClass}P_{\slowClass}}{\mValue_{\fastClass}(T')P_{\slowClass}+P_{\fastClass}}\\
	&\!\!\!\!\stackrel{P_{\fastClass} \geq P_{\slowClass}}{\leq}\!\!\!\! \frac{2P_{\fastClass}P_{\slowClass}}{\mValue_{\fastClass}(T')P_{\slowClass} + P_{\slowClass}}
	= \frac{2P_{\fastClass}}{\mValue_{\fastClass}(T')+1}.
	\end{align*}
\end{proof}
\begin{proof}[Proof of \Cref{splittable:running_time_improvement}]
	Due to \Cref{splittable:fastClass} any class that jumps in $X$ jumps \emph{outside} of $X$ the next time. So for every class $i \in I_{\exp}$ there is at most one jump in $X$ and hence $X$ contains at most $|I_{\exp}| \leq c$ jumps. Apparently the sum of the running times of each step is $\oh(n + c\log(c+m))$ and the returned value $T \in \set{T_{\op{ok}},T_{\op{new}}}$ holds $T \leq \OPT_{\op{split}}$ while being accepted by \Cref{lemma:splittable:algorithm}\ref{lemma:splittable:algorithm:T_check_true} such that the algorithm for the splittable case computes a feasible schedule with ratio $\tfrac32$ in time $\oh(n)$. The total running time is $\oh(n + c\log(c+m))$.
\end{proof}

\section{Preemptive Scheduling}
\label{preemptive-algorithm}

	One basic tool will be \emph{Batch Wrapping}, i.e. the wrapping of \emph{wrap sequences} into \emph{wrap templates}. See \Cref{principle:mcnaughton} for the short details.
	For this section let $T_{\min} := \max\set{\tfrac1m N, \max_{i\in[c]}(s_i+t_{\max}^{(i)})}$ where $N = \sum_{i=1}^c s_i + \sum_{j \in J} t_j$ and $t_{\max}^{(i)} = \max_{j \in C_i} t_j$.

\subsection{Nice Instances}
\label{nice_instances}

In the following we take a closer look on $I_{\exp}$ and $I_{\chp}$ so we split them again.
As stated before we divide the expensive classes into three disjoint subsets $I_{\exp}^+$, $I_{\exp}^0$ and $I_{\exp}^-$ such that $i \in I_{\exp}$ holds $i \in I_{\exp}^+$ iff. $T \leq s_i + P(C_i)$, $i \in I_{\exp}^0$ iff. $\tfrac34T < s_i + P(C_i) < T$ and $i \in I_{\exp}^-$ iff. $s_i + P(C_i) \leq \tfrac34T$.
Also we divide the cheap classes into $I_{\chp}^+$, $I_{\chp}^-$ s.t. $i \in I_{\chp}$ holds $i \in I_{\chp}^+$ iff. $\tfrac14T \leq s_i \leq \tfrac12T$ and $i \in I_{\chp}^-$ iff. $s_i < \tfrac14T$.
Now denote the big jobs of a class $i \in I_{\chp}^-$ as $C_i^* = \set{j \in C_i | s_i + t_j > \tfrac12T}$ and let $I_{\chp}^* = \set{i \in I_{\chp}^- | 1 \leq |C_i^*|} \subseteq I_{\chp}^-$ be the set of classes that contain at least one of these jobs.

\restateniceinstances
\noindent Let $\alpha_i' := \floor*{P(C_i)/(T-s_i)} \leq \ceil*{P(C_i)/(T-s_i)} = \alpha_i$ and remark that $\alpha_i' \geq 1$ for all $i \in I_{\exp}^+$. The following theorem will be of great use to find a $(3/2)$-ratio also for general instances.

\restatepreemptivesimple

\begin{algorithm}
	\caption{A $\tfrac32$-dual Approximation for Nice Instances}
	\label{preemptive:algorithm:simple}
	\begin{enumerate}[1.]
		\item Schedule $J(I_{\exp}^+)$ on $\sum_{i \in I_{\exp}}\alpha_i'$ new machines with $\alpha_i'$ machines for each class $i$
		\label{preemptive:algorithm:simple:exp+}
		\item Schedule $J(I_{\exp}^-)$ in pairs of two classes on $\ceil{\tfrac12|I_{\exp}^-|}$ new machines
		\label{preemptive:algorithm:simple:exp-}
		\item Wrap $J(I_{\chp})$ onto the residual machines starting on machine $\mu$ (last machine of step 2.)
		\label{preemptive:algorithm:simple:chp}
	\end{enumerate}
\end{algorithm}

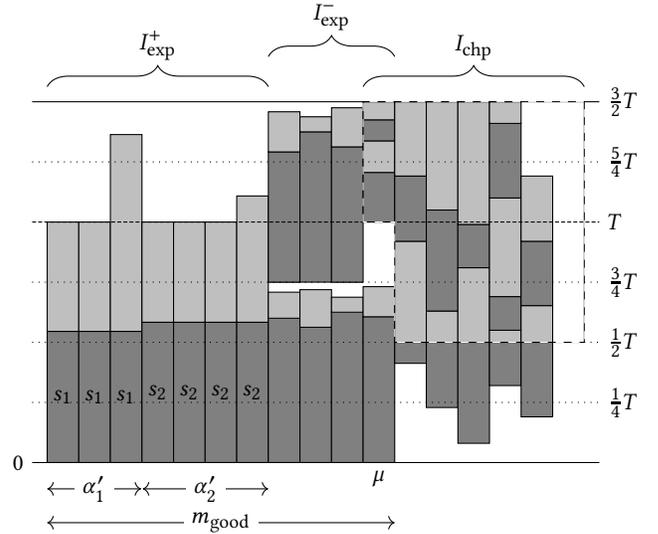
\begin{figure}
	\centering
	\begin{tikzpicture}
  \usetikzlibrary{patterns}

\def\width{0.42}

\draw [-, decorate, decoration={brace,amplitude=8pt}] (0*\width,5) -- node[above=7pt] {$I_{\exp}^+$} (7*\width,5);
\draw [-, decorate, decoration={brace,amplitude=8pt}] (7*\width,5.4) -- node[above=7pt] {$I_{\exp}^-$} (11*\width,5.4);
\draw [-, decorate, decoration={brace,amplitude=8pt}] (10*\width,5) -- node[above=7pt] {$I_{\chp}$} (17*\width,5);

\DrawMachines{
{
	{6/11}/$s_1$/[fill=gray],
	{5/11}//
}/\width/[fill=lightgray],
{
	{6/11}/$s_1$/[fill=gray],
	{5/11}//
}/\width/[fill=lightgray],
{
	{6/11}/$s_1$/[fill=gray],
	{9/11}//
}/\width/[fill=lightgray],
{
	{7/12}/$s_2$/[fill=gray],
	{5/12}//
}/\width/[fill=lightgray],
{
	{7/12}/$s_2$/[fill=gray],
	{5/12}//
}/\width/[fill=lightgray],
{
	{7/12}/$s_2$/[fill=gray],
	{5/12}//
}/\width/[fill=lightgray],
{
	{7/12}/$s_2$/[fill=gray],
	{6.3/12}//
}/\width/[fill=lightgray],
{
	{7.2/12}//[fill=gray],
	{1.3/12}//,
    {0.5/12}//[opacity=0],
	{6.5/12}//[fill=gray],
	{2./12}//
}/\width/[fill=lightgray],
{
	{9/16}//[fill=gray],
	{2.5/16}//,
    {0.5/16}//[opacity=0],
	{10/16}//[fill=gray],
	{1/16}//
}/\width/[fill=lightgray],
{
	{10/16}//[fill=gray],
	{1/16}//,
    {1/16}//[opacity=0],
	{9/16}//[fill=gray],
	{2.6/16}//
}/\width/[fill=lightgray],
{
	{9.7/16}//[fill=gray],
	{2/16}//,
	{4.3/16}//[opacity=0],
    {3.3/16}//[fill=gray],
	{2.1/16}//,
	{1.4/16}//[fill=gray],
	{1.2/16}//
}/\width/[fill=lightgray],
{
	{6.6/16}//[opacity=0],
	{1.4/16}//[fill=gray],
	{3.36/8}//,
	{2.17/8}//[fill=gray],
	{2.47/8}//
}/\width/[fill=lightgray],
{
	{3.66/16}//[opacity=0],
	{2.17/8}//[fill=gray],
{0.13}//,
{0.42}//[fill=gray],
{0.45}//
}/\width/[fill=lightgray],
{
	{0.08}//[opacity=0],
{0.42}//[fill=gray],
	{0.31}//,
{0.18}//[fill=gray],
{0.51}//
}/\width/[fill=lightgray],
{
	{0.32}//[opacity=0],
	{0.18}//[fill=gray],
	{0.05}//,
{0.14}//[fill=gray],
{0.41}//,
{0.31}//[fill=gray],
{0.09}//
}/\width/[fill=lightgray],
{
	{0.19}//[opacity=0],
{0.31}//[fill=gray],
	{1.22/8}//,
	{2.14/8}//[fill=gray],
	{2.17/8}//
}/\width/[fill=lightgray]
}

\draw (-0.2,0) node [left] {$0$} -- (17*\width+0.2, 0);
\draw (-0.2,0.8) -- (17*\width+0.2, 0.8) [dotted] node [right] {$\tfrac14T$};
\draw (-0.2,1.6) -- (17*\width+0.2, 1.6) [dotted] node [right] {$\tfrac12T$};
\draw (-0.2,2.4) -- (17*\width+0.2, 2.4) [dotted] node [right] {$\tfrac34T$};
\draw (-0.2,3.2) -- (17*\width+0.2, 3.2) [dash pattern={on 1.5pt off 0.8pt}] node [right] {$T$};
\draw (-0.2,4) -- (17*\width+0.2, 4) [dotted] node [right] {$\tfrac54T$};
\draw (-0.2,4.8) -- (17*\width+0.2, 4.8) node [right] {$\tfrac32T$};

\draw (16*\width,1.6) -- (17*\width,1.6) -- (17*\width,4.8);
\draw [dashed,white] (10*\width,3.2) -- (10*\width,4.8) -- (17*\width,4.8) -- (17*\width,1.6) -- (11*\width,1.6) -- (11*\width,3.2) -- (10*\width,3.2);

\draw [<->] (0,-0.35) -- (3*\width-0.01,-0.35) node[fill=white, pos=.5] {$\alpha_1'$};
\draw [<->] (0.01+3*\width,-0.35) -- (7*\width,-0.35) node[fill=white, pos=.5] {$\alpha_2'$};
\draw [<->] (0,-0.8) -- (11*\width,-0.8) node[fill=white, pos=.5] {$m_{\op{good}}$};

\node (mu) at (10.5*\width,-0.5*\width) {$\mu$};

\end{tikzpicture}

	\caption{An example solution after using \Cref{preemptive:algorithm:simple} with $I_{\exp}^+ = \set{1,2}$}
	\label{fig:preemptions:simple}
\end{figure}

\Sref{preemptive:algorithm:simple:exp+}
First, we look at the classes $i \in I_{\exp}$ with $s_i + P(C_i) \geq T$, i.e. $i \in I_{\exp}^+$. We define a wrap template $\omega^{(i)}$ of length $|\omega^{(i)}| = \floor{P(C_i)/(T-s_i)} = \alpha_i'$ for each class $i \in I_{\exp}^+$ as follows. Let $\omega^{(i)}_1 = (u_i,0,T)$ and $\omega^{(i)}_{1+r} = (u_i+r,s_i,T)$ for all $1 \leq r < \alpha_i'$. The first machines $u_i$ have to be chosen distinct to all machines of the other wrap templates. We construct simple wrap sequences $Q^{(i)} = [s_i,C_i]$ for each class $i \in I_{\exp}^+$ consisting of an initial setup $s_i$ followed by an arbitrary order of all jobs in $C_i$. For all $i \in I_{\exp}^+$ we use $\textsc{Wrap}(Q^{(i)},\omega^{(i)})$ to wrap $Q^{(i)}$ into $\omega^{(i)} = (\omega_1^{(i)},\dots,\omega_{\alpha_i'}^{(i)})$. The last machine $\bar{u}_i := u_i + \alpha_i' - 1$ will have a load of at most $T$ but its job load will be less than $\tfrac12T$ since $s_i > \tfrac12T$. We move these jobs to the second last machine and place them on top. So the new load will be greater than $T$ but at most $\tfrac32T$. Finally we remove the setup time $s_i$ on the last machine.

\Sref{preemptive:algorithm:simple:exp-}
Second, we turn to the classes $i \in I_{\exp}$ with $s_i + P(C_i) \leq \tfrac34T$, i.e. $i \in I_{\exp}^-$. We place them paired on one new machine $u$. So $u$ will have a load $\load(u) = s_{i_1} + P(C_{i_1}) + s_{i_2} + P(C_{i_2})$ for different classes $i_1, i_2$ that holds $T = \tfrac12T + \tfrac12T < s_{i_1} + s_{i_2} < \load(u) \leq \tfrac34T + \tfrac34T = \tfrac32T$. Note that the number of such classes can be odd. In this case we schedule one class separate on a new machine $\mu$. Otherwise we choose an unused machine and name it $\mu$. Be aware that this is for the ease of notation; in fact, an unused machine may not exist. So in both cases $\mu$ will hold $L(\mu) \leq \tfrac34T$. Apparently this step uses $\ceil{\tfrac12|I_{\exp}^-|}$ new machines.

\Sref{preemptive:algorithm:simple:chp}
The third and last step is to place the jobs of cheap classes. We build a simple wrap template $\omega$ as with a case distinction as follows. If $|I_{\exp}^-|$ is odd, we set $\omega_1 = (\mu,T,\tfrac32T)$ and $\omega_{1+r} = (\mu+r,\tfrac12T,\tfrac32T)$ for $1 \leq r \leq m-m_{\op{nice}}$. Otherwise we set $\omega_r = (\mu+r,\tfrac12T,\tfrac32T)$ for all $0 \leq r < m - m_{\op{nice}}$. So in any case we have $|\omega| \leq m - m_{\op{nice}} + 1$. We define $Q$ to be the simple wrap sequence $Q = [s_i, C_i]_{i \in I_{\chp}}$ that contains all jobs of cheap classes. So we wrap $Q$ into $\omega = (\omega_1,\dots,\omega_{|\omega|})$ using $\textsc{Wrap}(Q,\omega)$.

\Cref{fig:preemptions:simple} shows an example schedule after the last step. Be aware that setups are dark gray and that jobs of a class are not explicitly drawn.

\begin{proof}[Proof of \Cref{preemptive:simple}]
	\ref{preemptive:decision:T_check_false}. We show that $T \geq \OPT_{\op{pmtn}}(I)$ implies $mT \geq L_{\op{nice}}$ and $m \geq m_{\op{nice}}$. Let $T \geq \OPT_{\op{pmtn}}(I)$. Then there is a feasible schedule $\sigma$ with makespan $T$. Let $\load(\sigma) = \sum_{i=1}^c (\lambda_i^{\sigma}s_i + P(C_i))$. Apparently \Cref{lemma:a_i-lower-bound} implies that \[mT \geq \load(\sigma) \geq P(J) + \sum_{i=1}^c \alpha_i s_i \geq P(J) + \!\!\sum_{i \in I_{\exp}^+}\!\! \alpha_i' s_i + \!\!\!\!\!\!\!\sum_{i \in I_{\exp}^- \cup I_{\chp}} \!\!\!\!\!\!\! s_i = L_{\op{nice}}.\]	
	Due to \Cref{lemma:a_i-lower-bound,lemma:machine_number_expensive_classes_multiplicities} we know that $m \geq \sum_{i \in I_{\exp}} \lambda_i^{\sigma} \geq \sum_{i \in I_{\exp}} \alpha_i \ge \sum_{i \in I_{\exp}} \alpha_i'$ and hence
	\[
	m \geq \sum_{i \in I_{\exp}^+} \alpha_i' + \sum_{i \in I_{\exp}^-} \alpha_i'
	\geq \sum_{i \in I_{\exp}^+} \alpha_i' + \ceil{\tfrac12|I_{\exp}^-|}
	= m_{\op{nice}}.
	\]
	
	\ref{preemptive:decision:T_check_true}.
	One may easily confirm that the wrap templates $\omega^{(i)}$ suffice to wrap the sequences $Q^{(i)}$ into them, i.e. $S(\omega^{(i)}) \geq \load(Q^{(i)})$ for all $i \in I_{\exp}^+$. Apparently the total number of machines used by step \ref{preemptive:algorithm:simple:exp+} and \ref{preemptive:algorithm:simple:exp-} is $m_{\op{nice}} \leq m$ so there are enough machines for the first two steps.
	It remains to show that wrap template $\omega$ is sufficient to wrap $Q$ into it, i.e. $S(\omega) \geq \load(Q)$. We find that
	\begin{equation}\label{eq:alpha_i'}
	\alpha_i's_i + P(C_i) \geq \alpha_i's_i + \floor*{\frac{P(C_i)}{T-s_i}}(T-s_i) = \alpha_i's_i + \alpha_i'(T-s_i) = \alpha_i'T
	\end{equation}
	for all $i \in I_{\exp}^+$. Considering that $|I_{\exp}^-|$ is odd we show the inequality.
	\begin{align*}
		S(\omega) &= \tfrac12T + (m-m_{\op{nice}})T\\
		&\ge L_{\op{nice}} + \tfrac12T - m_{\op{nice}}T \qquad\qquad\qquad\qquad\qquad \text{// } mT \geq L_{\op{nice}}\\
		&= L_{\op{nice}} + \tfrac12T - (\ceil{\tfrac12|I_{\exp}^-|} + \!\!\sum_{i \in I_{\exp}^+} \!\! \alpha_i')T\\
		&\ge \!\!\! \sum_{i \in I_{\chp}} \!\! (s_i + P(C_i)) + \!\!\!\!\sum_{i \in I_{\exp}^-} \!\! (s_i + P(C_i)) + \tfrac12T - \ceil{\tfrac12|I_{\exp}^-|}T \quad\,\, \text{// } \eqref{eq:alpha_i'}\\
		&\ge \load(Q) + \sum_{i \in I_{\exp}^-} \tfrac12T + \tfrac12T - \ceil{\tfrac12|I_{\exp}^-|}T \,\,\, = \, \load(Q)
	\end{align*}
	This gets even easier if $|I_{\exp}^-|$ is even.
\end{proof}

\subsection{General Instances}

\noindent Consider a makespan $T \geq T_{\min}$ and remember the partitions $I_{\exp} = I_{\exp}^+ \cupdot I_{\exp}^0 \cupdot I_{\exp}^-$ and $I_{\chp} = I_{\chp}^+ \cupdot I_{\chp}^-$ as well as the machine numbers $\alpha_i' \leq \alpha_i$ as mentioned in \vref{nice_instances}. We state the algorithm and then we go through the details.

\begin{algorithm}[h]
	\caption{A $\tfrac32$-dual Approximation for Preemptive Scheduling}
	\label{preemptive:algorithm}
	\begin{enumerate}[1.]
		\item Schedule $J(I_{\exp}^0)$ on $l=|I_{\exp}^0|$ machines using one machine per class (the \emph{large machines})
		\label{preemptive:algorithm:exp0}
		\item Find the free time $F$ for $J(I_{\chp}^-)$ on the residual machines in order to apply \Cref{preemptive:algorithm:simple} on a nice instance and split each big job of $J(I_{\chp}^-)$ in two pieces (due to \Cref{note:too_much_load_for_large_machines} below)
		\label{preemptive:algorithm:split}
		\item Find a feasible placement for $J(I_{\chp}^-)$ on the residual $m - l$ empty machines and the free time at the bottom of the large machines of step \ref{preemptive:algorithm:exp0} In more detail:
		
		\textbf{If} $F < \sum_{i \in I_{\chp}^*} (s_i + P(C_i))$ \textbf{then}
		\begin{enumerate}
			\item Solve an appropriate knapsack instance for the decision, place a nice instance (containing the solution, $J(I^+_{\exp})$, $J(I^-_{\exp})$ and $J(I^+_{\chp})$) with \Cref{preemptive:algorithm:simple}, and place the unselected items at the bottom of the large machines
			\label{preemptive:algorithm:chp-:a}
		\end{enumerate}
		\textbf{else}
		\begin{enumerate}
			\setcounter{enumii}{1}
			\item Use a greedy approach and the last placement idea of \ref{preemptive:algorithm:chp-:a}
			\label{preemptive:algorithm:chp-:b}
		\end{enumerate}
		\label{preemptive:algorithm:chp-}
	\end{enumerate}
\end{algorithm}

\Sref{preemptive:algorithm:exp0}
First we consider all classes $i \in I_{\exp}$ with $\tfrac34T < s_i + P(C_i) < T$, i.e. $i \in I_{exp}^0$. We place every class on its own machine $u$, i.e. $\load(u) = s_i + P(C_i)$, starting at time $\tfrac12T$. Note that the used machines have less than $\tfrac14T$ free time to schedule any other jobs in a $T$-feasible schedule. Let $l = |I_{\exp}^0|$ be the number of these machines. We refer to them as the \emph{large machines}. \Cref{fig:preemptions:exp0} illustrates the situation. The lighter drawn items indicate the future use of \Cref{preemptive:algorithm:simple}, whereas the question marks indicate the areas where we need to decide the placement of $J(I_{\chp}^-)$.

\begin{figure}[h]
	\centering
	\scalebox{0.93}{\begin{tikzpicture}
  \usetikzlibrary{patterns}

\def\width{0.3}

\draw [-, decorate, decoration={brace,amplitude=8pt}] (0,5) -- node[above=7pt] {$I_{\exp}^0$} (10*\width,5);
\draw [-, decorate, decoration={brace,amplitude=8pt},opacity=0.5] (10*\width,5) -- node[above=7pt] {$I_{\exp}^+$} (17*\width,5);
\draw [-, decorate, decoration={brace,amplitude=8pt},opacity=0.5] (17*\width,5) -- node[above=7pt] {$I_{\exp}^-$} (21*\width,5);

\DrawMachines{
{
	{0.5}//[opacity=0],
	{17/32}//[fill=gray],
	{8/32}//
}/\width/[fill=lightgray],
{
	{0.5}//[opacity=0],
	{23/32}//[fill=gray],
	{5/32}//
}/\width/[fill=lightgray],
{
	{0.5}//[opacity=0],
	{25/32}//[fill=gray],
	{1/32}//
}/\width/[fill=lightgray],
{
	{0.5}//[opacity=0],
	{20/32}//[fill=gray],
	{10/32}//
}/\width/[fill=lightgray],
{
	{0.5}//[opacity=0],
	{18/32}//[fill=gray],
	{8/32}//
}/\width/[fill=lightgray],
{
	{0.5}//[opacity=0],
	{17/32}//[fill=gray],
	{10/32}//
}/\width/[fill=lightgray],
{
	{0.5}//[opacity=0],
	{19/32}//[fill=gray],
	{7/32}//
}/\width/[fill=lightgray],
{
	{0.5}//[opacity=0],
	{25/32}//[fill=gray],
	{5/32}//
}/\width/[fill=lightgray],
{
	{0.5}//[opacity=0],
	{18/32}//[fill=gray],
	{10/32}//
}/\width/[fill=lightgray],
{
	{0.5}//[opacity=0],
	{21/32}//[fill=gray],
	{4/32}//
}/\width/[fill=lightgray],
{
	{6/11}/$s_1$/[fill=gray],
	{5/11}//
}/\width/[{fill=lightgray,opacity=0.5}],
{
	{6/11}/$s_1$/[fill=gray],
	{5/11}//
}/\width/[{fill=lightgray,opacity=0.5}],
{
	{6/11}/$s_1$/[fill=gray],
	{9/11}//
}/\width/[{fill=lightgray,opacity=0.5}],
{
	{7/12}/$s_2$/[fill=gray],
	{5/12}//
}/\width/[{fill=lightgray,opacity=0.5}],
{
	{7/12}/$s_2$/[fill=gray],
	{5/12}//
}/\width/[{fill=lightgray,opacity=0.5}],
{
	{7/12}/$s_2$/[fill=gray],
	{5/12}//
}/\width/[{fill=lightgray,opacity=0.5}],
{
	{7/12}/$s_2$/[fill=gray],
	{6.3/12}//
}/\width/[{fill=lightgray,opacity=0.5}],
{
	{7.2/12}//[fill=gray],
	{1.3/12}//,
	{0.5/12}//[opacity=0],
	{6.5/12}//[fill=gray],
	{2./12}//
}/\width/[{fill=lightgray,opacity=0.5}],
{
	{9/16}//[fill=gray],
	{2.5/16}//,
	{0.5/16}//[opacity=0],
	{10/16}//[fill=gray],
	{1/16}//
}/\width/[{fill=lightgray,opacity=0.5}],
{
	{10/16}//[fill=gray],
	{1/16}//,
	{1/16}//[opacity=0],
	{9/16}//[fill=gray],
	{2.6/16}//
}/\width/[{fill=lightgray,opacity=0.5}],
{
	{9.7/16}//[fill=gray],
	{2/16}//,
	{4.3/16}//[opacity=0],
	{3.3/16}//[fill=gray],
	{2.1/16}//,
	{1.4/16}//[fill=gray],
	{1.2/16}//
}/\width/[{fill=lightgray,opacity=0.5}],
{
	{6.6/16}//[opacity=0],
	{1.4/16}//[fill=gray],
	{3.36/8}//,
	{2.17/8}//[fill=gray],
	{2.47/8}//
}/\width/[{fill=lightgray,opacity=0.5}],
{
	{3.66/16}//[opacity=0],
	{2.17/8}//[fill=gray],
	{0.13}//,
	{0.42}//[fill=gray],
	{0.45}//
}/\width/[{fill=lightgray,opacity=0.5}],
{
	{0.08}//[opacity=0],
	{0.42}//[fill=gray],
	{0.31}//,
	{0.18}//[fill=gray],
	{0.51}//
}/\width/[{fill=lightgray,opacity=0.5}],
{
	{0.32}//[opacity=0],
	{0.18}//[fill=gray],
	{0.05}//,
	{0.14}//[fill=gray],
	{0.41}//,
	{0.31}//[fill=gray],
	{0.09}//
}/\width/[{fill=lightgray,opacity=0.5}],
{
	{0.19}//[opacity=0],
	{0.31}//[fill=gray],
	{1.22/8}//,
	{2.14/8}//[fill=gray],
	{2.17/8}//
}/\width/[{fill=lightgray,opacity=0.5}]
}

\draw (-0.2,0) node [left] {$0$} -- (26*\width+0.2, 0);
\draw (-0.2,0.8) -- (26*\width+0.2, 0.8) [dotted] node [right] {$\tfrac14T$};
\draw (-0.2,1.6) -- (26*\width+0.2, 1.6) [dotted] node [right] {$\tfrac12T$};
\draw (-0.2,2.4) -- (26*\width+0.2, 2.4) [dotted] node [right] {$\tfrac34T$};
\draw (-0.2,3.2) -- (26*\width+0.2, 3.2) [dash pattern={on 1.5pt off 0.8pt}] node [right] {$T$};
\draw (-0.2,4) -- (26*\width+0.2, 4) [dotted] node [right] {$\tfrac54T$};
\draw (-0.2,4.8) -- (26*\width+0.2, 4.8) node [right] {$\tfrac32T$};

\draw [<->] (0,-0.4) -- (10*\width,-0.4) node[fill=white, pos=.5] {$l = |I_{\exp}^0|$};
\draw [<->] (0,-0.8) -- (26*\width,-0.8) node[fill=white, pos=.5] {$m$};

\node [opacity=0.5] (mu) at (20.5*\width,-0.4*\width) {$\mu$};

\node at (5*\width,0.8) {? \quad ? \quad ? \quad ? \quad ? \quad ?};

\node at (22.5*\width,4) {? \quad ? \quad ?};
\node at (23.5*\width,2.4) {? \quad ? \quad ?};

\end{tikzpicture}
}
	\caption{An example situation after step \ref{preemptive:algorithm:exp0} with $I_{\exp}^+ = \set{1,2}$}
	\label{fig:preemptions:exp0}
\end{figure}
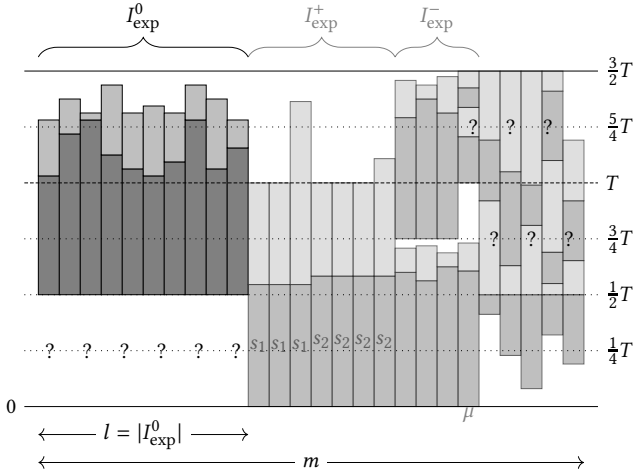

\begin{lemma}\label{note:too_much_load_for_large_machines}
	In a $T$-feasible schedule a job $j \in C_i^*$ of a class $i \in I_{\chp}^*$ can not be scheduled on large machines only. Furthermore, $j$ can be processed on large machines with a total processing time of at most $\tfrac12T - s_i$.
\end{lemma}
\begin{proof}
	Remark that the load of the large machines is at least $\tfrac34T$. The placed setup times are greater than $\tfrac12T$ so we can not pause their jobs execution to schedule other jobs because that would need at least one more setup time. So the load of the large machines has to be scheduled consecutively and hence there is at most $\tfrac14T$ time to place other load at the bottom or on top of them. Due to that and since at least one setup $s_i$ is required, it is easy to see that there can be scheduled at most $2(\tfrac14T - s_i) < \tfrac12T - s_i < t_j$ time of job $j$ on large machines (on one at the top and on another one at the bottom) since $j$ is not allowed to be processed on different machines at the same time.
\end{proof}

\Sref{preemptive:algorithm:split}
Since the setups of the classes $I_{\chp}^+$ are too big to place any of their jobs on a large machine, we definitely will place the jobs $J(I_{\exp}^+ \cup I_{\exp}^- \cup I_{\chp}^+)$ entirely on the residual $m - l$ machines. We need to obtain the free time $F$ for $J(I_{\chp}^-)$ on the residual machines; in fact, we want to place as much load as possible, so, looking at \Cref{preemptive:algorithm:simple} we find that the time
\begin{equation}\label{eq:preemptive:free_time}
F = (m-l)T - \sum_{i \in I_{\exp}^+} (\alpha_i's_i + P(C_i))\,\, - \!\!\!\!\!\!\sum_{i\in I_{\exp}^- \cup I_{\chp}^+}\!\!\!\! (s_i + P(C_i))
\end{equation}
can be used to place the jobs of $J(I_{\chp}^-)$ in step \ref{preemptive:algorithm:simple:chp} of \Cref{preemptive:algorithm:simple}.

Apparently the free processing time on the large machines is $\tilde{F} := \sum_{u = 1}^{l} (T - \load(u))$, so for a suitable value of $T$ the residual available processing time $\tilde{F} + F$ suffices to schedule the residual jobs $J(I_{\chp}^-)$.	
Remember that $C_i^* = \set{j \in C_i | s_i + t_j > \tfrac12T}$ are the big jobs of a class $i \in I_{\chp}^-$, and $I_{\chp}^* \subseteq I_{\chp}^-$ denotes the classes that contain at least one of these jobs. As stated out in \Cref{note:too_much_load_for_large_machines}, we can not place them on large machines only. So we split them in the following way. For all jobs $j \in C_i^*$ with $i \in I_{\chp}^*$ we create new job pieces $j^{(1)}$ and $j^{(2)}$ with processing times $t_j^{(1)}$ and $t_j^{(2)}$, satisfying $t_j^{(1)} = \tfrac12T - s_i$ as well as $t_j^{(2)} = s_i + t_j - \tfrac12T$ and hence, $t_j^{(1)} + t_j^{(2)} = t_j$. Note that $s_i + t_j^{(1)} \leq \tfrac12T$ and $t_j^{(2)} \leq T - \tfrac12T = \tfrac12T$.
Due to \Cref{note:too_much_load_for_large_machines} we have to schedule a processing time of at least $t_j^{(2)}$ of job $j \in C_i^*$ \emph{outside} the large machines. So we also have an obligatoy setup time $s_i$ outside the large machines.
The now following case distinction of step \ref{preemptive:algorithm:chp-} is a bit more complicated.

\noindent\textbf{Case \ref{preemptive:algorithm:chp-:a}:} $F < \sum_{i \in I_{\chp}^*} (s_i + P(C_i))$. Now we have to use large machines to schedule all jobs of $J(I_{\chp}^*)$. The task is to optimize the use of setup times. We do this by minimizing the total load of necessary new setup times to be placed on the large machines $1,\dots,l$. Each class that can be scheduled entirely outside the large machines will not cause a setup time on large machines. So the setup optimization can be done by maximizing the total sum of setup times of classes we schedule \emph{entirely} outside large machines.
The obligatory job load \emph{outside} large machines for a class $i \in I_{\chp}^*$ is 
\begin{equation}\label{obligatory_load_of_class_outside_large_machines}
L^*_i = \sum_{j \in C_i^*} t_j^{(2)} = \sum_{j \in C_i^*} (s_i + t_j - \tfrac12T) = P(C_i^*) - |C_i^*|(\tfrac12T - s_i).
\end{equation}
Therefore the total obligatory load outside large machines of all classes in $I_{\chp}^*$ is
\begin{equation}\label{obligatory_load_outside_large_machines}
L^* = \sum_{i \in I_{\chp}^*} (s_i + L^*_i) = \sum_{i \in I_{\chp}^*} (s_i + P(C_i^*) - |C_i^*|(\tfrac12T - s_i)).
\end{equation}
Now we can interpret the maximization problem as a knapsack problem by setting $\mathcal{I} := I_{\chp}^*$, capacity $Y := F - L^*$, profit $p_i := s_i$ and weight $w_i := P(C_i) - L^*_i$ for all $i \in I_{\chp}^*$. We compute an optimal solution $x_{\cks}$ for the \emph{continuous knapsack problem} with split item $e \in I_{\chp}^*$ that leads to a nearly optimal solution $x_{\ks}$ for $\ILP_{\ks}$ (general knapsack problem) computable in time $\oh(|I_{\chp}^*|) \leq \oh(c)$.	
So $0 < (x_{\cks})_e < 1$ may be critical because this means we need to schedule an extra setup time $s_e$ although it might not be necessary in an optimal schedule. We overcome this issue later.
Remark that $\sum_{i \in I_{\chp}^*} (x_{\ks})_i w_i = Y - (x_{\cks})_e w_e$ so $(x_{\cks})_e w_e$ is the time to fill with job load of class $e$. Therefore, we create new job pieces as follows. For all $j \in C_e$ let $j^{[1]}$ and $j^{[2]}$ be jobs with processing times $t_j^{[1]}$ and $t_j^{[2]}$ holding
\begin{equation}
t_j^{[2]} = \begin{cases}
(x_{\cks})_e t_j & : \quad j \in C_e \setminus C_e^* \\
(x_{\cks})_e t_j^{(1)} + t_j^{(2)} & : \quad j \in C_e^*
\end{cases}
\end{equation}
as well as $t_j^{[1]} = t_j - t_j^{[2]}$. Now we define a new instance $I^{(\op{new})}$ containing all classes $i \in I_{\exp}^+ \cup I_{\exp}^- \cup I_{\chp}^+$ with all of their jobs $C_i^{(\op{new})} := C_i$, all selected classes $i \in I_{\chp}^*$ holding $(x_{\cks})_i = 1$ with all of their jobs $C_i^{(\op{new})} := C_i$, all unselected classes $i \in I_{\chp}^*\setminus\set{e}$ holding $(x_{\cks})_i = 0$ with just their obligatory load $C_i^{(\op{new})} := \set{j^{(2)}|j \in C_i^*}$, and the split item class $e \in I_{\chp}^*$ with just the load $C_e^{(\op{new})} := \set{j^{[2]}|j \in C_e}$. Last we set $m^{(\op{new})} := m-l$. Apparently $I^{(\op{new})}$ is a nice instance and we schedule it on the residual $m^{(\op{new})}$ machines using \Cref{preemptive:algorithm:simple}.
Later we will see that this load fills the gap of $(x_{\cks})_e w_e$ to $Y$ since the obligatory job load of $\sum_{j \in C_e^*} t_j^{(2)}$ is enlarged by exactly $(x_{\cks})_e w_e$.
So with $x_{\cks}$ we found a (sub-)schedule that fills up the free time $Y$ outside the large machines in an optimal way; in fact, we maximized the setup times of the selected classes such that the sum of the setup times of unselected classes got minimized. Hence, the residual load can be scheduled feasibly in the free time $\tilde{F}$ on the large machines, if $T$ is suitable.
Let $K$ be the set of the residual jobs and job pieces, i.e.
\begin{equation}
K = \set{j^{[1]}\!| j \in C_e\!}\, \cup \!\!\!\!\!\!\!\bigcup_{\substack{i \in I_{\chp}^*\setminus\set{e}\\(x_{\cks})_i = 0}} \!\!\!\!\!\!\!\!\! \left(\set{j^{(1)}\!| j \in C_i^*\!} \cup (C_i \! \setminus \! C_i^*)\right) \,\cup\, J(I_{\chp}^- \!\setminus I_{\chp}^*).
\end{equation}
In \Cref{note:pmtn:residual_jobs_are_small} we will see that all jobs (or job pieces) $\iota \in K$ of a class $i$ hold $s_i + t_\iota \le \tfrac12T$. In the following we schedule the jobs of $K$ at the bottom of the large machines.

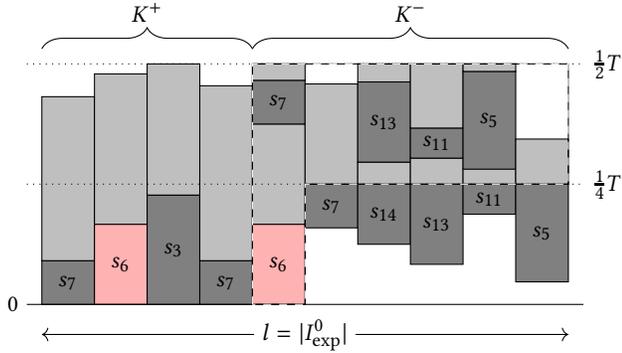
\begin{figure}[h]
	\centering
	\begin{tikzpicture}
  \usetikzlibrary{patterns}

\def\width{0.7}

\draw [-, decorate, decoration={brace,amplitude=8pt}] (0,3.4) -- node[above=7pt] {$K^+$} (4*\width,3.4);
\draw [-, decorate, decoration={brace,amplitude=8pt}] (4*\width,3.4) -- node[above=7pt] {$K^-$} (10*\width,3.4);

\DrawMachines{
{
	{2/11}/$s_7$/[fill=gray],
	{7.5/11}//
}/\width/[fill=lightgray],
{
	{4/12}/$s_6$/[fill=red!30!white],
	{7.5/12}//
}/\width/[fill=lightgray],
{
	{5/11}/$s_3$/[fill=gray],
	{6/11}//
}/\width/[fill=lightgray],
{
	{2/11}/$s_7$/[fill=gray],
	{8/11}//
}/\width/[fill=lightgray],
{
	{4/12}/$s_6$/[fill=red!30!white],
	{5/12}//,
    {2/11}/$s_7$/[fill=gray],
	{3/44}//
}/\width/[fill=lightgray],
{
	{7/22}//[opacity=0],
	{2/11}/$s_7$/[fill=gray],
	{5/12}//
}/\width/[fill=lightgray],
{
	{3/12}//[opacity=0],
	{3/12}/$s_{14}$/[fill=gray],
	{1.1/12}//,
	{4/12}/$s_{13}$/[fill=gray],
	{.9/12}//
}/\width/[fill=lightgray],
{
	{2/12}//[opacity=0],
	{4/12}/$s_{13}$/[fill=gray],
	{1.3/12}//,
	{1.5/12}/$s_{11}$/[fill=gray],
	{3.2/12}//
}/\width/[fill=lightgray],
{
	{4.5/12}//[opacity=0],
	{1.5/12}/$s_{11}$/[fill=gray],
	{1/16}//,
	{13/32}/$s_5$/[fill=gray],
	{0.03125}//
}/\width/[fill=lightgray],
{
	{3/32}//[opacity=0],
	{13/32}/$s_5$/[fill=gray],
	{3/16}//
}/\width/[fill=lightgray]
}

\draw (-0.2,0) node [left] {$0$} -- (10*\width+0.2, 0);

\draw (-0.2,1.6) -- (10*\width+0.2, 1.6) [dotted] node [right] {$\tfrac14T$};
\draw (-0.2,3.2) -- (10*\width+0.2, 3.2) [dotted] node [right] {$\tfrac12T$};

\draw (4*\width,0) -- (4*\width,3.2) -- (10*\width,3.2) -- (10*\width,1.6) -- (5*\width,1.6) -- (5*\width,0) -- (4*\width,0);
\draw [dashed,white] (4*\width,0) -- (4*\width,3.2) -- (10*\width,3.2) -- (10*\width,1.6) -- (5*\width,1.6) -- (5*\width,0) -- (4*\width,0);

\draw [<->] (0,-0.4) -- (10*\width,-0.4) node[fill=white, pos=.5] {$l = |I_{\exp}^0|$};

\end{tikzpicture}

	\caption{An example solution at the bottom of the large machines after using \Cref{preemptive:algorithm} with $\set{3,5,6,7,11,13,14} \subseteq I_{\chp}^-$ and $e=6$}
	\label{fig:preemptions:bottom_of_large_machines}
\end{figure}

\begin{note}\label{note:pmtn:residual_jobs_are_small}
	All jobs (or job pieces) $\iota \in K$ of class $i$ hold $s_i + t_\iota \leq \tfrac12T$.	
\end{note}
\begin{proof}
	First we show that this is true for all $\iota \in \set{j^{[1]}| j \in C_e}$.
	For $j \in C_e \setminus C_e^*$ we have $s_e + t_j \leq \tfrac12T$ and hence
	\[s_e + t_j^{[1]} = s_e + t_j - (x_{\cks})_e t_j \leq \tfrac12T - (x_{\cks})_e t_j \leq \tfrac12T.\]
	For $j \in C_e^*$ we use $t_j^{(2)} = s_e + t_j - \tfrac12T$ to see that
	\[
	s_e + t_j^{[1]} = s_e + t_j - (x_{\cks})_e t_j^{(1)} - (s_e + t_j - \tfrac12T)
	= \tfrac12T - (x_{\cks})_e t_j^{(1)}
	\leq \tfrac12T.
	\]
	Let $i \in I_{\chp}^*$ be a class with $(x_{\cks})_i = 0$. We have $s_i + t_j^{(1)} = s_i + \tfrac12T - s_i = \tfrac12T$ for all $j \in C_i^*$ and per definition of $C_i^*$ the bound holds for all jobs in $C_i \setminus C_i^*$.
	Also as a direct result of the definition of $I_{\chp}^*$ we get the bound for all jobs of $J(I_{\chp}^- \setminus I_{\chp}^*)$.
\end{proof}

We split $K = K^+ \cupdot K^-$ into big jobs $K^+ = \set{\iota \in K| t_\iota > \tfrac14T}$ and small jobs $K^- = \set{\iota \in K| t_\iota \leq \tfrac14T}$. Due to \Cref{preemptive:large_machines} it suffices to fill large machines with an obligatory load of at least $T$. Since the large machines already have a load of at least $\tfrac34T$, it is enough to add an obligatory load of at least $\tfrac14T$.
We start with the jobs of $K^+$. On the one hand all $\iota \in K^+$ of a class $i$ hold $t_\iota > \tfrac14T$ and on the other hand they can be placed entirely at the bottom of a large machine since $s_i + t_\iota \leq \tfrac12T$. So this is what we do. We place the jobs of $K^+$ on the first $l' \leq l$ large machines $1,\dotsc,l'$ with an initial associated setup time at time $0$ directly followed by the job (or job piece).
The very last step is to schedule the jobs of $K^-$. We remember that we need to schedule one setup time $s_e$ extra iff $(x_{\cks})_e > 0$. To avoid a case distinction we define a wrap template which is slightly larger than required for the obligatory load. Since all jobs (or job pieces) $\iota \in K^-$ need at most $\tfrac14T$ time, they can be wrapped without parallelization using a wrap template $\omega$ with $|\omega| = l - l'$ defined by $\omega_1 = (l'+1,0,\tfrac12T)$ and $\omega_{1+r} = (l'+1+r,\tfrac14T,\tfrac12T)$ for all $1 \leq r < l-l'$.
To construct a wrap sequence $Q$, we order the jobs and/or job pieces in $K^-$ by class, beginning with class $e$, and insert a suitable setup before the jobs of each class. Remark that $Q$ starts with $s_e$ followed by an arbitrary order of $\set{j^{[1]}|j \in C_e} \cap K^-$. Finally we wrap $Q$ into $\omega$ using $\textsc{Wrap}(Q,\omega)$.

See \Cref{fig:preemptions:bottom_of_large_machines} for an example solution at the bottom of the large machines. The split item class setup $e$ is colored red. Note that we even got two setups $s_6 = s_e$ since there was a big job in $K^+ \cap \set{j^{[1]}|j\in C_e}$. Also notice that setup $s_{14}$ was a critical item on machine $6$ and therefore it was moved below the next gap.

\noindent\textbf{Case \ref{preemptive:algorithm:chp-:b}:} $F \geq \sum_{i \in I_{\chp}^*} (s_i + P(C_i))$. Then there is enough time to schedule the jobs $J(I_{\chp}^*)$ entirely outside the large machines. Taking the previous case as a more complex model for this one, split $J(I_{\chp}^- \setminus I_{\chp}^*)$ into two well-defined wrap sequences $Q_1$, $Q_2$ such that $\load(Q_1) = F - \sum_{i \in I_{\chp}^*} (s_i + P(C_i))$ and there is at most one class $e \in I_{\chp}^- \setminus I_{\chp}^*$ with jobs (or job pieces) in \emph{both} sequences. Such a splitting can be obtained by a simple greedy approach. Then $J(I_{\exp}^+ \cup I_{\exp}^- \cup I_{\chp}^*)$ and the job pieces of $Q_1$ lead to a nice instance for the residual $m - l$ machines while the job pieces of $Q_2$ can be named $K$, be split into $K = K^+ \cupdot K^-$ and be handled just like before.

\subsection{Analysis}

We study case \ref{preemptive:algorithm:chp-:a} ($F < \sum_{i \in I_{\chp}^*} (s_i + P(C_i))$) only since the opposite case is much easier. We want to show the following theorem.

\begin{theorem}
	\label{preemptive:decision}
	Let $I$ be an instance and let $T$ be a makespan. Let $\alpha_i' = \floor*{\tfrac{P(C_i)}{T-s_i}}$ for all $i \in I_{\exp}^+$ and $x_{\cks}$ be the optimal solution to the knapsack problem of step \ref{preemptive:algorithm:chp-} and let
	\[L_{\op{pmtn}} = P(J) + \sum_{i \in I_{\exp}^+} \alpha_i' s_i + \sum_{i \in [c] \setminus I_{\exp}^+} s_i + \sum_{\substack{i \in I_{\chp}^*\\(x_{\cks})_i = 0}} s_i\]
	and $m' = |I_{\exp}^0| + \sum_{i \in I_{\exp}^+} \alpha_i' + \ceil{\tfrac12|I_{\exp}^-|}$. Then the following hold.
	\begin{enumerate}[label=(\roman*)]
		\item If $mT < L_{\op{pmtn}}$ or $m < m'$, then it is true that $T < \OPT_{\op{pmtn}}(I)$.
		\label{preemptive:decision:T_check_false}
		\item Otherwise a feasible schedule with makespan at most $\tfrac32T$ can be computed in time $\oh(n)$.
		\label{preemptive:decision:T_check_true}
	\end{enumerate}
\end{theorem}
\begin{proof}
	\ref{preemptive:decision:T_check_false}.
	We show that $T \geq \OPT_{\op{pmtn}}$ implies $mT \geq L_{\op{pmtn}}$ and $m \geq m'$. Let $T \geq \OPT_{\op{pmtn}}$. Then there is a feasible schedule $\sigma$ with makespan $T$. Let $\load(\sigma) = \sum_{i=1}^c (\lambda_i^{\sigma}s_i + P(C_i))$.		
	Since $F < \sum_{i \in I_{\chp}^*} (s_i + P(C_i^*))$, we know that we will need an extra setup $s_i$ for all unselected classes $i \in I_{\chp}^*$ holding $(x_{\cks})_i = 0$, due to \Cref{note:too_much_load_for_large_machines}. Together with \Cref{lemma:a_i-lower-bound} we get that
	\begin{align*}
	mT \geq \load(\sigma)
	&\geq P(J) + \sum_{i=1}^c \alpha_i s_i\\ &\geq P(J) + \sum_{i \in I_{\exp}^+} \alpha_i' s_i \, + \!\!\!\sum_{i \in [c] \setminus I_{\exp}^+} \!\!\!\!\!\!s_i \, + \!\!\!\sum_{\substack{i \in I_{\chp}^*\\(x_{\cks})_i = 0}} \!\!\!\!\!s_i = L_{\op{pmtn}}.
	\end{align*}
	Due to \Cref{lemma:a_i-lower-bound,lemma:machine_number_expensive_classes_multiplicities} we know that $m \geq \sum_{i \in I_{\exp}} \lambda_i^{\sigma} \geq \sum_{i \in I_{\exp}} \alpha_i$ and hence
	\[
	m \geq \!\sum_{i \in I_{\exp}^0}\! \alpha_i + \!\sum_{i \in I_{\exp}^+}\! \alpha_i + \!\sum_{i \in I_{\exp}^-}\! \alpha_i
	\geq |I_{\exp}^0| + \!\sum_{i \in I_{\exp}^+}\! \alpha_i' + \ceil{\tfrac12|I_{\exp}^-|}
	= m'.
	\]
	
	\ref{preemptive:decision:T_check_true}.
	Let $mT \geq L_{\op{pmtn}}$ and $m \geq m'$. Apparently there are enough machines $m \geq m' \geq |I_{\exp}^0| = l$ for step \ref{preemptive:algorithm:exp0}. It is important to see that this simple placement of one machine per class is legitimate only due to \Cref{preemptive:large_machines}. Now we study step \ref{preemptive:algorithm:split}.
	As mentioned in the description of the algorithm, $I^{(\op{new})}$ is a nice instance but we want to see that it is placed feasibly on the last $m-l$ machines. So we look on the split item class $e \in I_{\chp}^*$ again. We find that
	\begin{align*}
	\sum_{j \in C_e} t_j^{[2]}
	&= \sum_{j \in C_e \setminus C_e^*} \!\!\!\! (x_{\cks})_e t_j + \sum_{j \in C_e^*} \left((x_{\cks})_e t_j^{(1)} + t_j^{(2)}\right)\\
	&= \sum_{j \in C_e^*} t_j^{(2)} + (x_{\cks})_e \left(\sum_{j \in C_e \setminus C_e^*} t_j + \sum_{j \in C_e^*} t_j^{(1)}\right)\\
	&= L_e^* + (x_{\cks})_e \left(\sum_{j \in C_e} \! t_j - \!\!\sum_{j \in C_e^*} \! t_j^{(2)}\right) \quad\,\,\, \text{//  \eqref{obligatory_load_of_class_outside_large_machines}, } t_j^{(1)} = t_j - t_j^{(2)}\\
	&= L_e^* + (x_{\cks})_e w_e \hspace{6em} \text{// } w_e = P(C_e) - \sum_{j \in C_e^*} t_j^{(2)}
	\end{align*}
	and this means that the jobs $\set{j^{[2]}|j \in C_e}$ do expand the obligatory load $L_e^*$ outside the large machines of $e$ by exactly $(x_{\cks})_e w_e$, as mentioned before. Turning back to instance $I^{\op{new}}$, we name the cheap load $L_{\chp}^{(\op{new})}$ and find that
	\begin{equation}\label{L-nice-new}
	L^{(\op{new})}_{\op{nice}} = \sum_{i \in I_{\exp}^+} (\alpha_i's_i + P(C_i)) + \sum_{i\in I_{\exp}^- \cup I_{\chp}^+} (s_i + P(C_i)) + L_{\chp}^{(\op{new})}.
	\end{equation}
	We use the continuous knapsack characteristic $\sum_{i \in I_{\chp}^*} (x_{\cks})_i w_i + (x_{\cks})_e w_e = Y = F - L^*$ as well as $w_i = P(C_i) - L_i^*$ and hence $P(C_i) = L_i^* + w_i$ to show the following equality.
	\begin{align*}
	L_{\chp}^{(\op{new})}
	&= \!\!\!\sum_{\substack{i \in I_{\chp}^*\\(x_{\cks})_i = 1}} \!\! (s_i + P(C_i)) + \!\!\!\!\!\!\! \sum_{\substack{i \in I_{\chp}^* \setminus \set{e}\\(x_{\cks})_i = 0}} \!\!\! (s_i + P(\set{j^{(2)}|j \in C_i^*}))\\\vspace*{-3em}
	&\hspace{14em} + s_e + P(\set{j^{[2]}|j \in C_e})\\
	&= \sum_{\substack{i \in I_{\chp}^*\\(x_{\cks})_i = 1}} \!\!\!\!\! (s_i + L_i^* + w_i) + \!\!\!\!\!\!\! \sum_{\substack{i \in I_{\chp}^* \setminus \set{e}\\(x_{\cks})_i = 0}} \!\!\!\!\!\!\! (s_i + L_i^*) + s_e + L_e^* + (x_{\cks})_e w_e\\
	&= \sum_{i \in I_{\chp}^*} (s_i + L_i^*) + \sum_{\substack{i \in I_{\chp}^*\\(x_{\cks})_i = 1}} w_i + (x_{\cks})_e w_e\\
	&= L^* + \sum_{i \in I_{\chp}^*} (x_{\cks})_i w_i + (x_{\cks})_e w_e
	\,\,=\,\, L^* + F - L^*
	\,\,=\,\, F
	\end{align*}
	So with \Cref{L-nice-new,obligatory_load_outside_large_machines} it follows directly that $m^{(\op{new})}T = (m-l)T = L^{(\op{new})}_{\op{nice}}$. Also we have
	\[
	m^{(\op{new})} = m - l \geq m' - l = l + \sum_{i \in I_{\exp}^+} \alpha_i' + \ceil{\tfrac12|I_{\exp}^-|} - l = m_{\op{nice}}
	\]
	so \Cref{preemptive:simple}\ref{preemptive:simple:decision:T_check_true} is satisfied and hence $I^{\op{new}}$ is scheduled feasibly on the last $m-l$ machines with a makespan of at most $\tfrac32T$.
	It remains to show that $K$ can be placed at the bottom of the large machines. As already stated in the description, this is ensured by the optimality of the continuous knapsack solution. One may formally confirm that \[l\cdot \tfrac14T \geq \tilde{F} \geq P(K) + \sum_{\substack{i \in I_{\chp}^*\\(x_{\cks})_i = 0}} s_i + \sum_{i \in I_{\chp}^- \setminus I_{\chp}^*} s_i.\]
	
	Apparently each job in $K^+$ has a load of at least $\tfrac14T$ and is placed on exactly one large machine $u$, which holds $T-\load(u) < \tfrac14T$. According to this, the wrap template $\omega$ suffices to wrap the residual jobs $K^-$. Even big jobs of the split item class (or its setups) are no problem, since big jobs fill up large machines more than possible (in a $T$-feasible schedule). So the only setup to worry about is the setup $s_e$ wrapped into $\omega$. One can see that it is not part of the above inequality, since there is no time reserved for it. Fortunately, the time period $S(\omega)$ provided by $\omega$ is large enough, though. This is true, since
	\begin{align*}
	S(\omega) = \tfrac12T + (l-l'-1)\tfrac14T &= (l-l'+1)\tfrac14T\\
	&\geq (l-|K^+|)\tfrac14T + s_e = |K^-|\tfrac14T + s_e.
	\end{align*}
	\noindent Remark that jobs do never run in parallel, according to sufficient gap heights. However, it remains to obtain the running time. The total sum of the lengths of all used wrap templates \emph{and} wrap sequences is in $\oh(n)$. So they are wrapped in a total time of $\oh(n)$. Also we need at time of $\oh(n)$ to compute the knapsack instance. To solve it, we have linear time again. Overall the running time is $\oh(n)$.
\end{proof}

	\subsection{Class Jumping}
\label{improvedsearch}

Here we use our idea of \emph{Class Jumping} (cf. \Cref{class-jumping-splittable}) to show the following theorem.

\restatepreemptiverunningtimeimprovement

The idea for the splittable case can be applied to the preemptive one with a small modification as follows. We need to replace step \ref{preemptive:algorithm:simple:exp+} of \Vref{preemptive:algorithm:simple} such that the jumps of $I_{\exp}^+$ depend less on the setup time $s_i$. As in the algorithm for the splittable case, we define a gap of size $\tfrac12T$ above each setup. If a last machine got a total load of at least $T$ the machines are filled well. If a last machine got load less than $T$ its job load will be at most $T-s_i$. It turns out that the machines $u$ before hold $\tfrac32T - \load(u) = \tfrac32T - (s_i + \tfrac12T) = T - s_i$. So we simply move the job load of the last machine to the top of the second last machine (and remove the setup on the last machine). See \Cref{fig:preemptions:runtime_improvement} for an example. To define the associated machine number $\gamma_i$ for all classes $i \in I_{\exp}^+$ we set
\def\goodClassA{i_a}
\def\goodClassB{i_b}
\def\mValue{\gamma}
\[
\beta_i' = \floor*{\frac{2P(C_i)}T} \text{ and } 
\mValue_i = \begin{cases}
\max\set{\beta_i',1} & P(C_i) - \beta_i' \cdot \tfrac12T \leq T - s_i\\
\beta_i & \text{otherwise}
\end{cases}.
\]
Remark that $\gamma_i \leq \beta_i$. One can see that class $i$ jumps right after the last machine got a job load of exactly $T-s_i$. In more detail we obtain that a makespan $T$ is a jump of class $i$ if $P_i = \mValue_i(T) \cdot \tfrac12T + (T - s_i)$ and this may be rearranged to $T = 2(s_i + P_i)/(\mValue_i(T)+2)$.

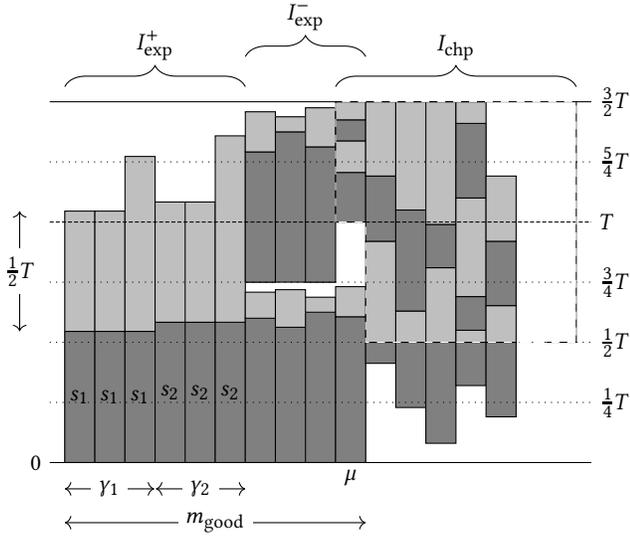
\begin{figure}[h]
	\centering
	\begin{tikzpicture}
  \usetikzlibrary{patterns}

\def\width{0.4}

\draw [-, decorate, decoration={brace,amplitude=8pt}] (0*\width,5) -- node[above=7pt] {$I_{\exp}^+$} (6*\width,5);
\draw [-, decorate, decoration={brace,amplitude=8pt}] (6*\width,5.4) -- node[above=7pt] {$I_{\exp}^-$} (10*\width,5.4);
\draw [-, decorate, decoration={brace,amplitude=8pt}] (9*\width,5) -- node[above=7pt] {$I_{\chp}$} (17*\width,5);

\draw [<->] (-0.6,3.2*6/11) -- (-0.6,3.2*6/11+1.6) node[fill=white, pos=.5] {$\tfrac12T$};

\DrawMachines{
{
	{6/11}/$s_1$/[fill=gray],
	{5.5/11}//
}/\width/[fill=lightgray],
{
	{6/11}/$s_1$/[fill=gray],
	{5.5/11}//
}/\width/[fill=lightgray],
{
	{6/11}/$s_1$/[fill=gray],
	{8/11}//
}/\width/[fill=lightgray],
{
	{7/12}/$s_2$/[fill=gray],
	{6/12}//
}/\width/[fill=lightgray],
{
	{7/12}/$s_2$/[fill=gray],
	{6/12}//
}/\width/[fill=lightgray],
{
	{7/12}/$s_2$/[fill=gray],
	{9.3/12}//
}/\width/[fill=lightgray],
{
	{7.2/12}//[fill=gray],
	{1.3/12}//,
    {0.5/12}//[opacity=0],
	{6.5/12}//[fill=gray],
	{2./12}//
}/\width/[fill=lightgray],
{
	{9/16}//[fill=gray],
	{2.5/16}//,
    {0.5/16}//[opacity=0],
	{10/16}//[fill=gray],
	{1/16}//
}/\width/[fill=lightgray],
{
	{10/16}//[fill=gray],
	{1/16}//,
    {1/16}//[opacity=0],
	{9/16}//[fill=gray],
	{2.6/16}//
}/\width/[fill=lightgray],
{
	{9.7/16}//[fill=gray],
	{2/16}//,
	{4.3/16}//[opacity=0],
    {3.3/16}//[fill=gray],
	{2.1/16}//,
	{1.4/16}//[fill=gray],
	{1.2/16}//
}/\width/[fill=lightgray],
{
	{6.6/16}//[opacity=0],
	{1.4/16}//[fill=gray],
	{3.36/8}//,
	{2.17/8}//[fill=gray],
	{2.47/8}//
}/\width/[fill=lightgray],
{
	{3.66/16}//[opacity=0],
	{2.17/8}//[fill=gray],
{0.13}//,
{0.42}//[fill=gray],
{0.45}//
}/\width/[fill=lightgray],
{
	{0.08}//[opacity=0],
{0.42}//[fill=gray],
	{0.31}//,
{0.18}//[fill=gray],
{0.51}//
}/\width/[fill=lightgray],
{
	{0.32}//[opacity=0],
	{0.18}//[fill=gray],
	{0.05}//,
{0.14}//[fill=gray],
{0.41}//,
{0.31}//[fill=gray],
{0.09}//
}/\width/[fill=lightgray],
{
	{0.19}//[opacity=0],
{0.31}//[fill=gray],
	{1.22/8}//,
	{2.14/8}//[fill=gray],
	{2.17/8}//
}/\width/[fill=lightgray]
}

\draw (-0.2,0) node [left] {$0$} -- (17*\width+0.2, 0);
\draw (-0.2,0.8) -- (17*\width+0.2, 0.8) [dotted] node [right] {$\tfrac14T$};
\draw (-0.2,1.6) -- (17*\width+0.2, 1.6) [dotted] node [right] {$\tfrac12T$};
\draw (-0.2,2.4) -- (17*\width+0.2, 2.4) [dotted] node [right] {$\tfrac34T$};
\draw (-0.2,3.2) -- (17*\width+0.2, 3.2) [dash pattern={on 1.5pt off 0.8pt}] node [right] {$T$};
\draw (-0.2,4) -- (17*\width+0.2, 4) [dotted] node [right] {$\tfrac54T$};
\draw (-0.2,4.8) -- (17*\width+0.2, 4.8) node [right] {$\tfrac32T$};

\draw (16*\width,1.6) -- (17*\width,1.6) -- (17*\width,4.8);
\draw [dashed,white] (9*\width,3.2) -- (9*\width,4.8) -- (17*\width,4.8) -- (17*\width,1.6) -- (10*\width,1.6) -- (10*\width,3.2) -- (9*\width,3.2);

\draw [<->] (0,-0.35) -- (3*\width-0.01,-0.35) node[fill=white, pos=.5] {$\gamma_1$};
\draw [<->] (0.01+3*\width,-0.35) -- (6*\width,-0.35) node[fill=white, pos=.5] {$\gamma_2$};
\draw [<->] (0,-0.8) -- (10*\width,-0.8) node[fill=white, pos=.5] {$m_{\op{good}}$};

\node (mu) at (9.5*\width,-0.5*\width) {$\mu$};

\end{tikzpicture}

	\caption{An example solution after using the modification of \Cref{preemptive:algorithm:simple} with $I_{\exp}^+ = \set{1,2}$}
	\label{fig:preemptions:runtime_improvement}
\end{figure}

\begin{algorithm}[h]
	\caption{Class Jumping for Preemptive Scheduling}
	\label{alg:preemptive:runtime_improvement}
	\begin{enumerate}[1.]
		\item Compute $P_i = P(C_i)$ for all $i \in [c]$ in time $\oh(n)$
		\label{alg:preemptive:runtime_improvement:compute_processing_times}
		\item Use consecutive binary search routines to find a right interval $V = (A_1,T_1]$ in time $\oh(n \log c)$ such that each $T \in (A_1, T_1]$ causes the same sets $I_{\exp}^+$, $I_{\exp}^0$, $I_{\exp}^-$ and $\set{i \in I_{\chp}^*| (x_{\cks})_i = 0}$
		\label{alg:preemptive:runtime_improvement:first_guess}
		\item Find a \emph{fastest jumping class} $\fastClass \in I_{\exp}^+$, i.e. $s_{\fastClass} + P_{\fastClass} \geq s_{\slowClass} + P_{\slowClass}$ for all $\slowClass \in I_{\exp}^+$, in time $\oh(c)$
		\label{alg:preemptive:runtime_improvement:find_fastest_jumping_class}
		\item Guess the right interval $W = \left(\tfrac{2(s_{\fastClass}+P_{\fastClass})}{\mValue_{\fastClass}(T_1)+k+3},\tfrac{2(s_{\fastClass} + P_{\fastClass})}{\mValue_{\fastClass}(T_1)+k+2}\right]$ for some integer $k < m$ such that $X := V \cap W \neq \emptyset$, in time $\oh(c\log m)$ and set $T_2 := \tfrac{2(s_{\fastClass}+P_{\fastClass})}{\mValue_{\fastClass}(T_1)+k+2}$
		\label{alg:preemptive:runtime_improvement:second_guess}
		\item Find and sort the $\oh(c)$ jumps in $X$ in time $\oh(c\log c)$
		\label{alg:preemptive:runtime_improvement:find_and_sort_jumps}
		\item Guess the right interval $Y = (T_{\op{fail}},T_{\op{ok}}] \subseteq W$ for two jumps $T_{\op{fail}},T_{\op{ok}}$ of two classes $i_a,i_b \in I_{\exp}^+$ which jump in $X$ such that no \emph{other} class jumps in $Y$, in time $\oh(c\log c)$
		\label{alg:preemptive:runtime_improvement:third_guess}
		\item Choose a suitable makespan in this interval in constant time
		\label{alg:preemptive:runtime_improvement:compute_opt}
	\end{enumerate}
\end{algorithm}

\para{Analysis of \Cref{alg:preemptive:runtime_improvement}.}
Again we focus on the crucial claim, that $X$ contains no more than $c$ jumps in total.

\begin{lemma}\label{preemptive:fastClass}
	If $T'$ is a jump of $\fastClass$, i.e. $T' = 2(s_{\fastClass}+P_{\fastClass})/(\mValue_{\fastClass}(T')+2)$, and $T''$ is a jump of a different class $\slowClass$, i.e. $T'' = 2(s_{\slowClass}+P_{\slowClass})/(\mValue_{\slowClass}(T'')+2)$, such that $T'' \leq T'$, then the next jump of class $\slowClass$ is smaller than the next jump of $\fastClass$, which may be written as
	\[
	\frac{2(s_{\slowClass}+P_{\slowClass})}{\mValue_{\slowClass}(T'')+3} \leq \frac{2(s_{\fastClass}+P_{\fastClass})}{\mValue_{\fastClass}(T')+3}.
	\]
\end{lemma}
\begin{proof}
	By simply rearranging the equations for $T'$ and $T''$ we find that
	\[\mValue_{\fastClass}(T') = \frac{2(s_{\fastClass}+P_{\fastClass})}{T'}-2
	\quad\, \text{and} \,\quad
	\mValue_{\slowClass}(T'') = \frac{2(s_{\slowClass}+P_{\slowClass})}{T''}-2
	\]
	and therefore we get
	\begin{align*}
		\frac{2(s_{\slowClass}+P_{\slowClass})}{\mValue_{\slowClass}(T'')+3}
		= \frac{2(s_{\slowClass}+P_{\slowClass})}{\frac{2(s_{\slowClass}+P_{\slowClass})}{T''}-2+3}
		&= \frac1{\frac1{T''} + \frac1{2(s_{\slowClass}+P_{\slowClass})}}\\
		&\leq \frac1{\frac1{T'} + \frac1{2(s_{\slowClass}+P_{\slowClass})}} \qquad \text{// } T'' \leq T'\\
		&\leq \frac1{\frac1{T'} + \frac1{2(s_{\fastClass}+P_{\fastClass})}}\\
		&= \frac{2(s_{\fastClass}+P_{\fastClass})}{\frac{2(s_{\fastClass}+P_{\fastClass})}{T'}-2+3}
		= \frac{2(s_{\fastClass}+P_{\fastClass})}{\mValue_{\fastClass}(T')+3}.
	\end{align*}
\end{proof}

\begin{proof}[Proof of \Cref{preemptive:running_time_improvement}]
	Due to \Cref{preemptive:fastClass} any class that jumps in $X$ jumps \emph{outside} of $X$ the next time. So for every class $i \in I_{\exp}^+$ there is at most one jump in $X$ and hence $X$ contains at most $|I_{\exp}^+| \leq c$ jumps. Apparently the sum of the running times of each step is $\oh(n\log(c+m))$ and the returned value $T \in \set{T_{\op{ok}},T_{\op{new}}}$ holds $T \leq \OPT_{\op{pmtn}}$ while being accepted by \Cref{preemptive:decision} \ref{preemptive:decision:T_check_true} such that \Cref{preemptive:algorithm} computes a feasible schedule with ratio $\tfrac32$ in time $\oh(n)$. So we get a total running time of $\oh(n\log(c+m)) \leq \oh(n \log n)$ since $c \leq n$ and $m < n$.
\end{proof}

\section{Conclusion}

There are several open questions. First of all: can we find any polynomial time approximation scheme for the preemptive scheduling problem? Remark that the splittable and preemptive case only differ in the (non-)parallelization of jobs. However, the preemptive problem turns out to be much harder to approximate. Especially because Jansen et al. \cite{DBLP:conf/innovations/JansenKMR19} have not found an (E)PTAS using n-folds, this remains as a very interesting open question. It might be an option to fix $m$ (cf. \Vref{table:known-results}) for first results. Also there may be constant bounds less than $\tfrac32$ with similar small running times as well.

Another remarkable point may be the investigation of unrelated/uniform machines; in fact, there is a known result of Correa et al. \cite{DBLP:conf/ipco/CorreaMMSSVV14} for the splittable case on unrelated machines.

Also we only discussed sequence-independent setups here. Considering sequence-\emph{dependent} setups the setup times are given as a matrix $S \in \naturals^{c\times c}$ of values $s_{(i_1,i_2)}$ which means that processing jobs of class $i_2$ on a machine currently set up for class $i_1$ will cost a setup time of $s_{(i_1,i_2)}$.
For example there is a very natural application to TSP for $m=1$ and $C_i = \set{j_i}$ with $t_{j_i}=0$ where the jobs/classes identify cities. Selecting setups dependent by the previous job as well as the next job, we have the classical TSP (path version).	

There may be similar approximation results by (re)using the ideas of this paper.

\begin{acks}
	We want to thank our reviewers for all of their profound reviews and many useful comments which helped us a lot to improve our paper.
\end{acks}


\bibliographystyle{ACM-Reference-Format}
\bibliography{../../.common/ref}

\appendix
	
	\crefalias{section}{appsec}
	\crefalias{subsection}{appsec}
	
	\vspace{4em}
	
	\section{More Preliminaries}
	
	Here we give the definition of \emph{Batch Wrapping} as well as some simple upper bounds for all three problem contexts.
	
	\subsection{Batch Wrapping}
	\label{principle:mcnaughton}
	
	Robert McNaughton solved \problem{P}{pmtn}{C_{\max}} in linear time \cite{McNaughton:1959:SDL:2780402.2780403}. McNaughton's \emph{wrap-around rule} simply schedules all jobs greedily from time $0$ to time $T = \max\set{t_{\max}, \tfrac1m\sum_{j \in J} t_j}$ splitting jobs whenever they cross the border $T$. Indeed, this is not applicable for our setup time problems. However, our idea of \emph{Batch Wrapping} can be understood as a generalization of McNaughton's wrap-around rule suitable for scheduling with setup times. To define it we need \emph{wrap templates} and \emph{wrap sequences} as follows.
	
	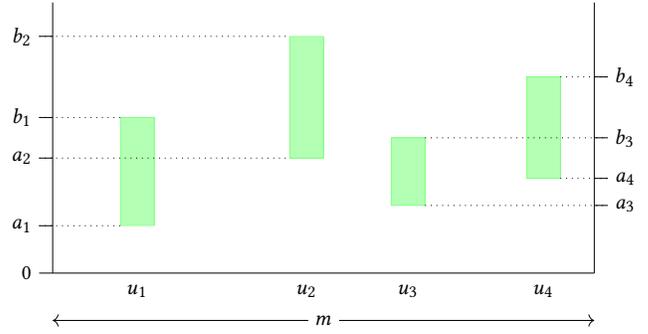
\begin{figure}
		\centering
		\scalebox{0.9}{\begin{tikzpicture}
  \usetikzlibrary{patterns}

\pgfmathsetmacro{\AONE}{0.7}
\pgfmathsetmacro{\BONE}{2.3}

\pgfmathsetmacro{\ATWO}{1.7}
\pgfmathsetmacro{\BTWO}{3.5}

\pgfmathsetmacro{\ATHREE}{1}
\pgfmathsetmacro{\BTHREE}{2}

\pgfmathsetmacro{\AFOUR}{1.4}
\pgfmathsetmacro{\BFOUR}{2.9}

  \draw (-0.2,0) node[left] {$0$} -- (8,0);
  \draw (0,0) -- (0,4);
  \draw (8,0) -- (8,4);

  \draw [color=green,fill=green,opacity=0.3] (1,\AONE) rectangle (1.5,\BONE);
  \draw [color=green,fill=green,opacity=0.3] (3.5,\ATWO) rectangle (4,\BTWO);
  \draw [color=green,fill=green,opacity=0.3] (5,\ATHREE) rectangle (5.5,\BTHREE);
  \draw [color=green,fill=green,opacity=0.3] (7,\AFOUR) rectangle (7.5,\BFOUR);

  \draw [dotted] (-0.2,\AONE) node[left] {$a_1$} -- (1,\AONE);
  \draw [dotted] (-0.2,\BONE) node[left] {$b_1$} -- (1,\BONE);
  \draw [dotted] (-0.2,\ATWO) node[left] {$a_2$} -- (3.5,\ATWO);
  \draw [dotted] (-0.2,\BTWO) node[left] {$b_2$} -- (3.5,\BTWO);
  \draw [dotted] (5.5,\ATHREE) -- (8.2,\ATHREE) node[right] {$a_3$};
  \draw [dotted] (5.5,\BTHREE) -- (8.2,\BTHREE) node[right] {$b_3$};
  \draw [dotted] (7.5,\AFOUR) -- (8.2,\AFOUR) node[right] {$a_4$};
  \draw [dotted] (7.5,\BFOUR) -- (8.2,\BFOUR) node[right] {$b_4$};

  \draw (-0.2,\AONE) -- (0,\AONE);
  \draw (-0.2,\BONE) -- (0,\BONE);
  \draw (-0.2,\ATWO) -- (0,\ATWO);
  \draw (-0.2,\BTWO) -- (0,\BTWO);
  \draw (8,\ATHREE) -- (8.2,\ATHREE);
  \draw (8,\BTHREE) -- (8.2,\BTHREE);
  \draw (8,\AFOUR) -- (8.2,\AFOUR);
  \draw (8,\BFOUR) -- (8.2,\BFOUR);

  \draw (1.25,-0.26) node {$u_1$};
  \draw (3.75,-0.26) node {$u_2$};
  \draw (5.25,-0.26) node {$u_3$};
  \draw (7.25,-0.26) node {$u_4$};

  \draw [<->] (0,-0.7) -- (8,-0.7) node[fill=white, pos=.5] {$m$};

\end{tikzpicture}
}
		\caption{An example of a wrap template $\omega$ with $|\omega| = 4$}
		\label{fig:pipe_example}
	\end{figure}
	
	\begin{definition}
		A \emph{wrap template} is a list $\omega = (\omega_1,\dotsc,\omega_{|\omega|}) \in ([m] \times \rationals \times \rationals)^*$ of triples $\omega_r = (u_r,a_r,b_r) \in [m]\times\rationals\times\rationals$ for $1 \leq r \leq |\omega|$ that hold the following properties:
		\[
		\text{(i)} \,\, u_r < u_{r+1} \text{ f.a. } 1 \leq r < |\omega|
		\quad
		\text{(ii)} \,\, 0 \leq a_r < b_r \text{ f.a. } 1 \leq r \leq |\omega|
		\]
		Let $S(\omega) := \sum_{r=1}^{|\omega|} (b_r-a_r)$ denote the provided period of time. A \emph{wrap sequence} is a sequence $Q = [s_{i_l},C_l']_{l \in [k]}$ where
		\[
		[s_{i_l},C_l']_{l \in [k]} = (s_{i_1}, j^1_1, \dotsc, j^1_{n_1}, s_{i_2}, j^2_1, \dotsc, j^2_{n_2}, \,\dotsc\,, s_{i_k}, j^k_1, \dotsc, j^k_{n_k})
		\]
		and $C_l' = \set{j^l_1, \dotsc, j^l_{n_l}}$ with $n_l = |C_l'|$ is a set of jobs and/or job pieces of a class $i_l \in [c]$ for $1 \leq l \leq k$. Let $\load(Q) := \sum_{l=1}^k (s_{i_l} + P(C_l'))$ denote the load of $Q$.
	\end{definition}
	
	These technical definitions need some intuition. Have a look at \Cref{fig:pipe_example} to see that a wrap template simply stores some free time gaps (colored green in \Cref{fig:pipe_example}) in a schedule where jobs may be placed. Remark that there can be at most one gap on each machine per definition. However, a wrap sequence is just a sequence of batches.
	
	We use wrap templates to schedule wrap sequences in the following manner. Let $X = \bigcup_{l=1}^k C_l'$ be a set of jobs and job pieces of $k$ different classes $i_1, \dotsc, i_k \in [c]$ where $C_l' = \set{j^l_1, \dotsc, j^l_{n_l}}$ is a set of jobs and job pieces of $C_{i_l}$ for all $1 \leq l \leq k$. Furthermore, let $Q = [s_{i_l},C_l']_{l \in [k]}$ be a wrap sequence. Hence, $Q$ has a length of $|Q| = \sum_{l=1}^k (1+n_l) = k + \sum_{l=1}^k n_l$. Now let $\omega = ((u_1,a_1,b_1),\dotsc,(u_{|\omega|},a_{|\omega|},b_{|\omega|}))$ be a wrap template. We want to schedule $Q$ in McNaughton's wrap-around style using the gaps $[a_r,b_r]$ for $1 \leq l \leq |\omega|$. The critical point is when an item $q$ hits the border $b_r$. If $q$ is a setup, the solution is simple. In this case we simply move $q$ below the next gap to be sure that the following jobs (or job pieces) get scheduled feasibly. If the critical item $q$ is a job (piece) of a class $i$, we split $q$ at time $b_r$ into two new jobs (just like McNaughton's wrap-around rule) to place one job piece at the end of the current gap and the other job piece at the beginning of the next gap. As before, we add a setup $s_i$ below the next gap to guarantee feasibility. We refer to the described algorithm as \textsc{Wrap}. Note that the critical job piece can be large such that it needs multiple gaps (and splits) to be placed. An algorithm for this splitting is given as \Cref{alg:split}, namely $\textsc{Split}$. $\textsc{Split}$ gets the current gap number $r$ as well as the point in time $t$ where $q$ should be processed. It returns the gap number and the point in time where the next item should start processing. A simple comparison with McNaughton's wrap-around rule will prove the following lemma.
	
	\begin{algorithm}
		\caption{Split a critical job (piece) to the subsequent gaps and add necessary setups}
		\label{alg:split}
		\begin{algorithmic}
			\Procedure{Split}{$q,\omega,r,t$}
			\State Let $i$ be the class of job (piece) $q$
			\State $p \gets q$; $t' \gets t+t_q$
			\While{$t' > b_r$}
			\State Split $q$ into new job pieces $q_1, q_2$
			\State\quad with $t_{q_1} = t' - b_r$ and $t_{q_2} = b_r - t$
			\State Place job piece $q_2$ at time $t$ on machine $u_r$
			\State $r \gets r+1$;
			$q \gets q_1$;
			$t \gets a_r$;
			$t' \gets a_r + t_q$
			\State\Comment{turn to the next gap}
			\State Place setup $s_i$ at time $t-s_i$ on machine $u_r$
			\EndWhile
			\State\textbf{done}
			\State Place job piece $q$ at time $t$ on machine $u_r$	\State\Comment{$q$ ($=q_1$) fits in the gap $[a_r,b_r]$}
			\State \Return{$(r,t')$}
			\EndProcedure
		\end{algorithmic}
	\end{algorithm}
	
	
	\begin{lemma}\label{lemma:wrap_soundness}
		Let $Q$ be a wrap sequence containing a largest setup $s_{\max}^{(Q)}$ and $\omega$ be a wrap template with $\load(Q) \leq S(\omega)$. Then $\textsc{Wrap}$ will place the last job (piece) of $Q$ in a gap $\omega_r$ with $r \leq |\omega|$. If there was a free time of at least $s_{\max}^{(Q)}$ below each gap but the first, the load gets placed feasibly.\qed
	\end{lemma}

	\begin{lemma}\label{lemma:wrapsplit_running_time}
		If $\load(Q) \leq S(\omega)$ then $\textsc{Wrap}(Q,\omega)$ has a running time of $\oh(|Q|+|\omega|)$.
	\end{lemma}
	\begin{proof}
		Whenever a critical job (piece) is obtained, $\textsc{Wrap}$ runs $\textsc{Split}$. Apparently each turn of the while loop in $\textsc{Split}$ results in a switch to the next entry of $\omega$, i.e. $r \gets r+1$, such that the total number of loop turns over all $q \in Q$ is bounded by $|\omega|$ thanks to \Cref{lemma:wrap_soundness}. A turn of the for loop in $\textsc{Wrap}$ needs constant time except for the execution of $\textsc{Split}$ such that we get a total running time of $\oh(|Q|+|\omega|)$.
	\end{proof}
	
	\subsection{Simple Upper Bounds}
	\label{simple_upper_bounds}
	
	For all three problems there is a $2$-approximation with running time $\oh(n)$. The algorithm for the splittable case is a simple application of wrap templates. For the non-preemptive and preemptive case one can use a slightly modified greedy solution.
	\begin{lemma}\label{splittable:2-approximation}
		There is a $2$-approximation for the splittable case running in time $\oh(n)$.
	\end{lemma}
	\begin{proof}
		Let $I$ be an instance and let $T^{(1)}_{\min} := \max\set{\tfrac1m N, s_{\max}} \leq \OPT_{\op{split}}(I)$ where $N = \sum_{i=1}^c s_i + \sum_{j \in J} t_j$. We construct a wrap template $\omega = (\omega_1,\dots,\omega_m)$ of length $|\omega| = m$ by setting $\omega_r := (r,s_{\max},s_{\max}+\tfrac1m N)$ for all $r \in [m]$, to wrap a wrap sequence $Q = [s_i,C_i]_{i \in [c]}$ containing all classes/jobs. Apparently we have $S(\omega) = \sum_{r=1}^m \tfrac1m N = N = \load(Q)$ and obviously the time $s_{\max}$ below each gap is sufficient to place the possibly missing setups. Hence, thanks to \Cref{lemma:wrap_soundness} this wrapping builds a feasible schedule with a makespan of at most $s_{\max} + \tfrac1m N \leq 2T_{\min}^{(1)} \leq 2\OPT_{\op{split}}(I)$. The attentive reader recognizes that $\textsc{Wrap}$ runs in time $\oh(|Q|+|\omega|) = \oh(c+n+m) > \oh(n)$ (cf. \Cref{lemma:wrapsplit_running_time}) but actually a smarter implementation of $\textsc{Split}$ is able to overcome this issue. For details see the proof of \Vref{lemma:splittable:algorithm}.
	\end{proof}
	
	Due to the same lower bounds (cf. \Cref{preemptive_simple_lower_bound,non-preemptive_simple_lower_bound}) the non-preemptive and preemptive case can be approximated using the same algorithm, stated in the proof of the following \Cref{non-preemptive:2-approximation}.
	
	\begin{lemma}\label{non-preemptive:2-approximation}
		There is a $2$-approximation for both the non-preemptive and preemptive case running in time $\oh(n)$.
	\end{lemma}
	
	\begin{proof}
		Let $I$ be an instance and let $T_{\min}$ be a makespan with $T_{\min} = \max\set{\tfrac1m N, \max_{i \in [c]} (s_i + t_{\max}^{(i)})}$ where $t_{\max}^{(i)} = \max_{j \in C_i} t_j$ and $N = \sum_{i=1}^c s_i + \sum_{j \in J} t_j$. Consider the non-preemptive case and remark that $T_{\min} \le \OPT_{\op{nonp}}(I)$. First, we group the jobs by classes. Let $C_i = \set{j_1^i,\dotsc,j_{n_i}^i}$  with $n_i = |C_i|$ for all classes $i \in [c]$.
		Beginning on machine $1$, we add one setup for each class followed by all jobs of the class. Whenever the load of the current machine exceeds $T_{\min}$, we keep the last placed item and proceed to the next machine. So we add the items $s_1,j_1^1,\dotsc,j_{n_1}^1,s_2,j_1^2,\dotsc,j_{n_2}^2,\dotsc,s_c,j_1^c,\dotsc,j_{n_c}^c$ to the machines using a next-fit strategy with threshold $T_{\min}$ (see \Cref{fig1} on the left).
		\begin{figure}[h]
			\centering
			\scalebox{1.00}{\begin{tikzpicture}
\usetikzlibrary{patterns,arrows.meta}

\def\width{0.53}

\draw (-0.2,0) -- (5*\width,0) node[pos=-0.06] {$0$};

\draw [fill=red!30!white] (0,0) rectangle (\width,1) node[pos=.5] {$s_1$};
\draw [fill=red!30!white] (0,1) rectangle (\width,1.8);
\draw [fill=red!30!white] (0,1.8) rectangle (\width,2.3);
\draw [fill=red!30!white] (0,2.3) rectangle (\width,2.9);
\draw [preaction={fill,red!30!white}, pattern=north east lines] (0,2.9) rectangle (\width,3.7);
\draw [fill=red!30!white] (\width,0) rectangle (2*\width,0.4);
\draw [fill=red!30!white] (\width,0.4) rectangle (2*\width,0.7);
\draw [fill=green!30!white] (\width,0.7) rectangle (2*\width,1.4) node[pos=.5] {$s_2$};
\draw [fill=green!30!white] (\width,1.4) rectangle (2*\width,1.6);
\draw [fill=green!30!white] (\width,1.6) rectangle (2*\width,1.7);
\draw [fill=green!30!white] (\width,1.7) rectangle (2*\width,2.1);
\draw [fill=green!30!white] (\width,2.1) rectangle (2*\width,2.7);
\draw [fill=blue!30!white] (\width,2.7) rectangle (2*\width,3.3) node[pos=.5] {$s_3$};
\draw [preaction={fill,blue!30!white}, pattern=north east lines] (\width,3.3) rectangle (2*\width,3.9);
\draw [fill=blue!30!white] (2*\width,0) rectangle (3*\width,0.3);
\draw [fill=yellow!30!white] (2*\width,0.3) rectangle (3*\width,2.3) node[pos=.5] {$s_4$};
\draw [fill=yellow!30!white] (2*\width,2.3) rectangle (3*\width,2.6);
\draw [preaction={fill,yellow!30!white}, pattern=north east lines] (2*\width,2.6) rectangle (3*\width,3.8);
\draw [fill=yellow!30!white] (3*\width,0) rectangle (4*\width,0.6);
\draw [fill=yellow!30!white] (3*\width,0.6) rectangle (4*\width,2.1);
\draw [fill=yellow!30!white] (3*\width,2.1) rectangle (4*\width,2.2);
\draw [preaction={fill,orange!30!white}, pattern=north east lines] (3*\width,2.2) rectangle (4*\width,3.6) node[pos=.5] {$s_5$};
\draw [fill=orange!30!white] (4*\width,0) rectangle (5*\width,0.2);
\draw [fill=orange!30!white] (4*\width,0.2) rectangle (5*\width,0.3);
\draw [fill=orange!30!white] (4*\width,0.3) rectangle (5*\width,0.7);
\draw [fill=orange!30!white] (4*\width,0.7) rectangle (5*\width,1);
\draw [dashed] (-0.2,3.5) -- (5.2*\width,3.5) node[right] {$T_{\min}$};
\draw (-0.2,7) -- (5.2*\width,7) node[right] {$2T_{\min}$};
\draw [<->] (0,-0.5) -- (5*\width,-0.5) node[fill=white, pos=.5] {$m$};

\draw[-{Latex[scale=2.0]}] (6*\width-0.2,2) -- (7*\width+0.2,2);

\draw (7.8*\width,0) -- (13*\width,0) node[pos=-0.06] {$0$};

\draw [fill=red!30!white] (8*\width,0) rectangle (9*\width,1) node[pos=.5] {$s_1$};
\draw [fill=red!30!white] (8*\width,1) rectangle (9*\width,1.8);
\draw [fill=red!30!white] (8*\width,1.8) rectangle (9*\width,2.3);
\draw [fill=red!30!white] (8*\width,2.3) rectangle (9*\width,2.9);

\draw [fill=red!30!white] (9*\width,0) rectangle (10*\width,1) node[pos=.5] {$s_1$};
\draw [preaction={fill,red!30!white}, pattern=north east lines] (9*\width,1) rectangle (10*\width,1.8);
\draw [fill=red!30!white] (9*\width,1.8) rectangle (10*\width,2.2);
\draw [fill=red!30!white] (9*\width,2.2) rectangle (10*\width,2.5);
\draw [fill=green!30!white] (9*\width,2.5) rectangle (10*\width,3.2) node[pos=.5] {$s_2$};
\draw [fill=green!30!white] (9*\width,3.2) rectangle (10*\width,3.4);
\draw [fill=green!30!white] (9*\width,3.4) rectangle (10*\width,3.5);
\draw [fill=green!30!white] (9*\width,3.5) rectangle (10*\width,3.9);
\draw [fill=green!30!white] (9*\width,3.9) rectangle (10*\width,4.5);
\draw [fill=blue!30!white] (9*\width,4.5) rectangle (10*\width,5.1) node[pos=.5] {$s_3$};
\draw [fill=blue!30!white] (10*\width,0) rectangle (11*\width,0.6) node[pos=.5] {$s_3$};
\draw [preaction={fill,blue!30!white}, pattern=north east lines] (10*\width,0.6) rectangle (11*\width,1.2);
\draw [fill=blue!30!white] (10*\width,1.2) rectangle (11*\width,1.5);
\draw [fill=yellow!30!white] (10*\width,1.5) rectangle (11*\width,3.5) node[pos=.5] {$s_4$};
\draw [fill=yellow!30!white] (10*\width,3.5) rectangle (11*\width,3.8);
\draw [fill=yellow!30!white] (11*\width,0) rectangle (12*\width,2) node[pos=.5] {$s_4$};
\draw [preaction={fill,yellow!30!white}, pattern=north east lines] (11*\width,2) rectangle (12*\width,3.2);
\draw [fill=yellow!30!white] (11*\width,3.2) rectangle (12*\width,3.8);
\draw [fill=yellow!30!white] (11*\width,3.8) rectangle (12*\width,5.3);
\draw [fill=yellow!30!white] (11*\width,5.3) rectangle (12*\width,5.4);
\draw [preaction={fill,orange!30!white}, pattern=north east lines] (12*\width,0) rectangle (13*\width,1.4) node[pos=.5] {$s_5$};
\draw [fill=orange!30!white] (12*\width,1.4) rectangle (13*\width,1.6);
\draw [fill=orange!30!white] (12*\width,1.6) rectangle (13*\width,1.7);
\draw [fill=orange!30!white] (12*\width,1.7) rectangle (13*\width,2.1);
\draw [fill=orange!30!white] (12*\width,2.1) rectangle (13*\width,2.4);
\draw [dashed] (7.8*\width,3.5) -- (13.2*\width,3.5) node[right] {$T_{\min}$};
\draw (7.8*\width,7) -- (13.2*\width,7) node[right] {$2T_{\min}$};
\draw [<->] (8*\width,-0.5) -- (13*\width,-0.5) node[fill=white, pos=.5] {$m$};

\end{tikzpicture}
}
			\caption{Example for a next-fit schedule with $m = c = 5$}
			\label{fig1}
		\end{figure}
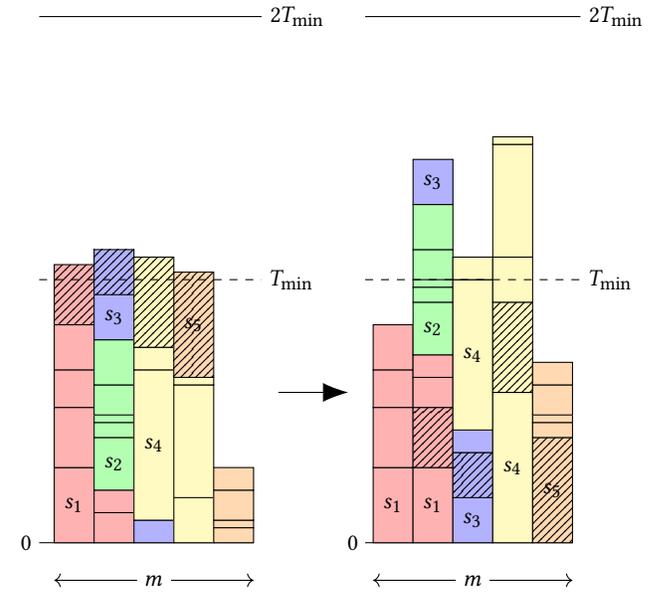
		The idea of the next step is to move the items (both jobs and setups) that cross the border $T_{\min}$ to the beginning of the next machine. For each moved item that was a job $j \in C_k$ we place an additional setup $s_k$ right before $j$. All other load is shifted up as much as required to place the moved items (see \Cref{fig1} on the right).
		In a last step one can remove unnecessary setups, i.e. setups which are scheduled last on a machine. In \Cref{fig1} this reduces the load of machine $2$ by the setup time $s_3$. Remark that the $T$-crossing items of the first step are hatched.
		
		\para{Analysis.}
		First consider the first step. The load placed is exactly \[\sum_{i=1}^c (s_i + P(C_i)) = \sum_{i=1}^c s_i + \sum_{j \in J} t_j = N = m\cdot\tfrac1m N \leq m \cdot T_{\min}\]
		and this states that the load of the last machine is at most $T_{\min}$. Apparently the makespan of the resulting schedule is at most $T_{\min} + \max(s_{\max}, t_{\max})$.
		Now turn to the second step and consider a machine $u < m$. Passing away the item that exceeds $T_{\min}$, the load of $u$ reduces to at most $T_{\min}$. Finally it increases to at most $T_{\min} + \max_{i \in [c]}(s_i + t_{\max}^{(i)}) \le T_{\min} + T_{\min} = 2T_{\min}$ since $u$ potentially receives an item $q$ from machine $u - 1$ as well as an initial setup if $q$ is a job. As mentioned above, the last machine $u = m$ already has a load of at most $T_{\min}$. So after the reassignment its load holds the bound of $2T_{\min}$ as well.
		Hence, in total the schedule has a makespan of at most $2T_{\min} \leq 2\OPT$ and apparently both steps do run in time $\oh(n)$ such that the described algorithm runs in time $\oh(n)$.
		
		Every solution to the non-preemptive case is a solution to the preemptive case and we obtained the same lower bounds for the preemptive case as for the non-preemptive one, i.e. $T_{\min} \le \OPT_{\op{pmtn}}(I)$. So the approximation can be used in the preemptive case as well.
	\end{proof}

	\section{Preemptive Scheduling}
	\label{preemptivescheduling}
	
	
	\subsection{The Soundness of Large Machines}
	
	The proof of \Cref{preemptive:decision} implicitly uses the fact that the use of our large machines is reasonable. To be convinced we prove the following lemma.
	\begin{lemma}\label{preemptive:large_machines}
		For every feasible schedule with makespan $T$ there is a feasible schedule $\sigma$ with a makespan of at most $\tfrac32T$ such that each class $i \in I_{\exp}^0$ is placed on exactly one machine $u_i$ which holds $s_i + P(C_i) \leq \load_\sigma(u_i) \leq T$.
	\end{lemma}
	
	\noindent We start with a more simple property.
	\begin{lemma}\label{preemptive:large_machines_modification}
		Let $\sigma$ be a feasible schedule with makespan $T$. Then there is a feasible schedule $\sigma'$ with a makespan of at most $\tfrac32T$ that holds the following properties.
		\begin{enumerate}
			\item If a class $i \in I_{\exp}^0$ is scheduled on exactly one machine $u_i$ in $\sigma$ (i.e. $\lambda^{\sigma}_i = 1$) then it is scheduled on exactly one machine $u_i'$ in $\sigma'$ such that setup $s_i$ starts processing at time $\tfrac12T$ and there is no more load above $C_i$ while $\load'(u_i') = \load(u_i) \leq T$.
			\item On all other machines of $\sigma'$ no job (piece) starts processing before time $\tfrac12T$.
		\end{enumerate}
	\end{lemma}
	\begin{proof}
		Let $\sigma$ be a feasible schedule with makespan $T$. Let $\load(\sigma) = \sum_{i=1}^c (\lambda^{\sigma}_i s_i + P(C_i))$. We do a simple machine modification. Consider a machine $u_i$ in $\sigma$ that schedules a class $i \in I_{\exp}^0$ holding $\lambda^{\sigma}_i = 1$. We reorder machine $u_i$ as follows.
		
		\begin{figure}[h]
			\centering
			\begin{tikzpicture}
  \usetikzlibrary{patterns,arrows.meta}

\def\diff{3}
\def\width{1.0}

\draw (-1,0) node [left] {$0$} -- (6, 0);
\draw (-1,0.8) -- (6, 0.8) [dotted] node [right] {$\tfrac14T$};
\draw (-1,1.6) -- (6, 1.6) [dotted] node [right] {$\tfrac12T$};
\draw (-1,2.4) -- (6, 2.4) [dotted] node [right] {$\tfrac34T$};
\draw (-1,3.2) -- (6, 3.2) [dash pattern={on 1.5pt off 0.8pt}] node [right] {$T$};
\draw (-1,4) -- (6, 4) [dotted] node [right] {$\tfrac54T$};
\draw (-1,4.8) -- (6, 4.8) node [right] {$\tfrac32T$};

\DrawMachines{
{
    {3.2/32}/{$A_i$}/,
	{17/32}/{$s_i$}/[fill=gray],
	{8/32}/{$C_i$}/,
    {3.8/32}/{$B_i$}/
}/\width/[fill=lightgray],
{}/\diff/,
{
    {3.2/32}/{$A_i$}/,
    {9/32}//[opacity=0],
    {3.8/32}/{$B_i$}/,
	{17/32}/{$s_i$}/[fill=gray],
	{8/32}/{$C_i$}/
}/\width/[fill=lightgray]
}

\fill[white] (2,-0.2) rectangle (3,5);

\node (mu) at (0.5,-0.4*\width) {$u_i$};

\draw[-{Latex[scale=2.0]}] (2,2) -- (3,2);

\end{tikzpicture}

			\caption{Modification of a large machine $u_i$}
			\label{modify_large_machine}
		\end{figure}
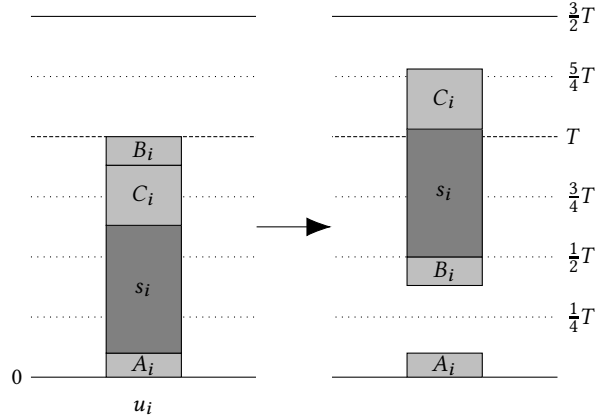
		
		We refer to the load below setup time $s_i$ as set $A_i$ and to the load above the last job of $C_i$ as set $B_i$. So $A_i$ and $B_i$ hold setup times, jobs and job pieces of classes $i' \neq i$.
		Now move up the setup time $s_i$ as well as all jobs of $C_i$ such that $s_i$ starts at time $\tfrac12T$ and the jobs of $C_i$ are scheduled consecutively right behind it. Also we move down each item of $B_i$ by exactly $\tfrac12T$ such that $A_i$ and $B_i$ are scheduled until time $\tfrac12T$. Since $s_i + P(C_i) > \tfrac34T$, we get $\load(A_i) + \load(B_i) < \tfrac14T$ and thus it follows $\load(A_i) < \tfrac14T$ as well as $\load(B_i) < \tfrac14T$. Hence the items of $A_i$ and $B_i$ do not intersect in time after this modification. Also notice that for different classes $i \neq i'$ there is no forbidden parallelization for preempted jobs of $A_i \cup A_{i'} \cup B_i \cup B_{i'}$ because relatively to each other they are scheduled in time just as before.	On all other machines of $\sigma$ we move up every item by exactly $\tfrac12T$ such that nothing is scheduled before line $\tfrac12T$. Apparently we get a feasible schedule with a makespan of at most $\tfrac32T$. 
	\end{proof}
	
	\begin{proof}[Proof of \Cref{preemptive:large_machines}]
		Let $\sigma$ be a feasible schedule with makespan $T$ and $\load(\sigma) = \sum_{i=1}^c (\lambda^{\sigma}_i s_i + P(C_i))$. First let us add some notation. We set $I_{\exp}^{0,1}(\sigma) \subseteq I_{\exp}^0$ as the set of classes $i \in I_{\exp}^0$ with $\lambda^{\sigma}_i = 1$ whereas $I_{\exp}^{0,2}(\sigma) = I_{\exp}^0 \setminus I_{\exp}^{0,1}(\sigma)$ denotes the set of classes $i \in I_{\exp}^0$ holding $\lambda^{\sigma}_i \geq 2$ such that $I_{\exp}^0 = I_{\exp}^{0,1}(\sigma) \cupdot I_{\exp}^{0,2}(\sigma)$.
		So $I_{\exp}^{0,1}(\sigma)$ is the set of classes already placed like intended. Nevertheless we need to modify their placement according to the feasibility of other classes.
		So for all $i \in I_{\exp}^{0,1}(\sigma)$ there is exactly one machine $u_i$ that schedules all jobs of $C_i$ in schedule $\sigma$. We modify them using \Cref{preemptive:large_machines_modification}. Apparently these machines may already schedule jobs or job pieces of $I_{\chp}$. To identify the residual load of $I_{\chp}$ we do the following. Using the notation of \Cref{preemptive:large_machines_modification}, let $t_j^{(1)}$ be the sum of the processing times of all job pieces in $\bigcup_{i \in I_{\exp}^{0,1}(\sigma)} X_i$ of a job $j \in J(I_{\chp})$ where $X_i = A_i \cup B_i$ for all $i \in I_{\exp}^{0,1}(\sigma)$. Remember that $X_i$ does not contain jobs or job pieces of class $i$. We create a new job piece $j^{(2)}$ for all jobs $j \in J(I_{\chp})$ with processing time $t_j^{(2)} := t_j - t_j^{(1)}$. So the residual jobs and job pieces of $I_{\chp}$ are $C_i' := \set{j^{(2)}| j \in C_i, t_j^{(2)} > 0}$ for all $i \in I_{\chp}$. Let $I_{\chp}' = \set{i \in I_{\chp}| 1 \leq |C_i'|}$ and since $\sigma$ is feasible, we obtain
		\begin{equation}\label{eq:appendixA:residual_machines_1}
		\begin{split}
		(m - |I_{\exp}^{0,1}(\sigma)|)T &\geq \load(\sigma) - \!\!\!\sum_{i I_{\exp}^{0,1}(\sigma)}\!\!\! \load(u_i)\\
		&= \!\!\!\!\!\!\!\!\!\sum_{i \in (I_{\exp} \setminus I_{\exp}^{0,1}(\sigma))} \!\!\!\!\!\!\!\!\! (\lambda_i s_i + P(C_i)) + \!\!\sum_{i \in I_{\chp}'} \! (s_i + P(C_i')).
		\end{split}
		\end{equation}	
		Each one of the $\lambda_i \geq 2$ machines used to schedule the jobs of a class $i \in I_{\exp}^{0,2}(\sigma)$ in $\sigma$ has a total load of different classes $i' \neq i$ of at most $\tfrac12T$ since $s_i > \tfrac12T$. We aim to schedule them on a single machine such that $\lambda_i' = 1$ if $\lambda_r'$ is the number of setup times used to schedule class $r$ in schedule $\sigma'$.	In fact, we can do this without scheduling load of different classes on the selected (single) machines. We extend the schedule as follows.
		Each class $i \in I_{\exp}^{0,2}$ is placed on a single machine with an initial setup time $s_i$ followed by $C_i$ with no load of other classes underneath or above.
		At first glance this seems rather wasteful because in schedule $\sigma$ there may be other load on machines scheduling $I_{\chp}^{0,2}(\sigma)$ in general. With a closer look we can convince us that its reasonable though. The idea is the following. A class $i \in I_{\exp}^{0,2}(\sigma)$ was placed in schedule $\sigma$ with a load of $L_i = \lambda^{\sigma}_i s_i + P(C_i) \geq 2s_i + P(C_i)$ on $\lambda^{\sigma}_i \geq 2$ machines whereas we sum up to a load of $L_i' = s_i + P(C_i)$ on only one machine in schedule $\sigma'$ now. Hence, we get at least one $s_i > \tfrac12T$ of processing time on a \emph{different} and so far unused machine since there cannot be two setup times of expensive classes on one machine. So we waste a time of $T - (s_i + P(C_i)) < \tfrac14T$ to schedule class $i$ while gaining at least $s_i > \tfrac12T$ of processing time.
		
		\begin{figure}[h]
			\centering
			\scalebox{0.93}{\begin{tikzpicture}
  \usetikzlibrary{patterns}

\def\width{0.3}

\draw [-, decorate, decoration={brace,amplitude=8pt}] (0,6) -- node[above=7pt] {$I_{\exp}^0$} (10*\width,6);
\draw [-, decorate, decoration={brace,amplitude=8pt}] (0,5) -- node[above=7pt] {$I_{\exp}^{0,1}(\sigma)$} (5*\width,5);
\draw [-, decorate, decoration={brace,amplitude=8pt}] (5*\width,5) -- node[above=7pt] {$I_{\exp}^{0,2}(\sigma)$} (10*\width,5);
\draw [-, decorate, decoration={brace,amplitude=8pt}] (10*\width,5) -- node[above=7pt] {$I_{\exp}^+$} (17*\width,5);
\draw [-, decorate, decoration={brace,amplitude=8pt}] (17*\width,5.4) -- node[above=7pt] {$I_{\exp}^-$} (21*\width,5.4);
\draw [-, decorate, decoration={brace,amplitude=8pt}] (20*\width,5) -- node[above=7pt] {$I_{\chp}'$} (26*\width,5);

\DrawMachines{
{
    {3.2/32}//,
    {9/32}//[opacity=0],
    {3.8/32}//,
	{17/32}//[fill=gray],
	{8/32}//
}/\width/[fill=lightgray],
{
    {4/32}//,
    {12/32}//[opacity=0],
    {0/32}//,
	{23/32}//[fill=gray],
	{5/32}//
}/\width/[fill=lightgray],
{
    {0/32}//,
    {10/32}//[opacity=0],
    {6/32}//,
	{25/32}//[fill=gray],
	{1/32}//
}/\width/[fill=lightgray],
{
    {0/32}//,
    {16/32}//[opacity=0],
    {0/32}//,
	{20/32}//[fill=gray],
	{10/32}//
}/\width/[fill=lightgray],
{
    {2/32}//,
    {13/32}//[opacity=0],
    {1/32}//,
	{18/32}//[fill=gray],
	{8/32}//
}/\width/[fill=lightgray],
{
	{17/32}//[fill=gray],
	{10/32}//
}/\width/[fill=lightgray],
{
	{19/32}//[fill=gray],
	{7/32}//
}/\width/[fill=lightgray],
{
	{25/32}//[fill=gray],
	{5/32}//
}/\width/[fill=lightgray],
{
	{18/32}//[fill=gray],
	{10/32}//
}/\width/[fill=lightgray],
{
	{21/32}//[fill=gray],
	{4/32}//
}/\width/[fill=lightgray],
{
	{6/11}/$s_1$/[fill=gray],
	{5/11}//
}/\width/[fill=lightgray],
{
	{6/11}/$s_1$/[fill=gray],
	{5/11}//
}/\width/[fill=lightgray],
{
	{6/11}/$s_1$/[fill=gray],
	{9/11}//
}/\width/[fill=lightgray],
{
	{7/12}/$s_2$/[fill=gray],
	{5/12}//
}/\width/[fill=lightgray],
{
	{7/12}/$s_2$/[fill=gray],
	{5/12}//
}/\width/[fill=lightgray],
{
	{7/12}/$s_2$/[fill=gray],
	{5/12}//
}/\width/[fill=lightgray],
{
	{7/12}/$s_2$/[fill=gray],
	{6.3/12}//
}/\width/[fill=lightgray],
{
	{7.2/12}//[fill=gray],
	{1.3/12}//,
    {0.5/12}//[opacity=0],
	{6.5/12}//[fill=gray],
	{2./12}//
}/\width/[fill=lightgray],
{
	{9/16}//[fill=gray],
	{2.5/16}//,
    {0.5/16}//[opacity=0],
	{10/16}//[fill=gray],
	{1/16}//
}/\width/[fill=lightgray],
{
	{10/16}//[fill=gray],
	{1/16}//,
    {1/16}//[opacity=0],
	{9/16}//[fill=gray],
	{2.6/16}//
}/\width/[fill=lightgray],
{
	{9.7/16}//[fill=gray],
	{2/16}//,
	{4.3/16}//[opacity=0],
	{3.3/16}//[fill=gray],
	{2.1/16}//,
	{1.4/16}//[fill=gray],
	{1.2/16}//
}/\width/[fill=lightgray],
{
	{6.6/16}//[opacity=0],
	{1.4/16}//[fill=gray],
	{3.36/8}//,
	{2.17/8}//[fill=gray],
	{2.47/8}//
}/\width/[fill=lightgray],
{
	{3.66/16}//[opacity=0],
	{2.17/8}//[fill=gray],
	{0.13}//,
	{0.42}//[fill=gray],
	{0.45}//
}/\width/[fill=lightgray],
{
	{0.08}//[opacity=0],
	{0.42}//[fill=gray],
	{0.31}//,
	{0.18}//[fill=gray],
	{0.51}//
}/\width/[fill=lightgray],
{
	{0.32}//[opacity=0],
	{0.18}//[fill=gray],
	{0.05}//,
	{0.14}//[fill=gray],
	{0.41}//,
	{0.31}//[fill=gray],
	{0.09}//
}/\width/[fill=lightgray],
{
	{0.19}//[opacity=0],
	{0.31}//[fill=gray],
	{1.22/8}//,
	{2.14/8}//[fill=gray],
	{2.17/8}//
}/\width/[fill=lightgray]
}

\draw [-, decorate, decoration={brace,amplitude=7pt}] (5*\width,-0.1) -- node[below=4pt] {$I_{\chp}$} (0,-0.1);

\draw (-0.2,0) node [left] {$0$} -- (26*\width+0.2, 0);
\draw (-0.2,0.8) -- (26*\width+0.2, 0.8) [dotted] node [right] {$\tfrac14T$};
\draw (-0.2,1.6) -- (26*\width+0.2, 1.6) [dotted] node [right] {$\tfrac12T$};
\draw (-0.2,2.4) -- (26*\width+0.2, 2.4) [dotted] node [right] {$\tfrac34T$};
\draw (-0.2,3.2) -- (26*\width+0.2, 3.2) [dash pattern={on 1.5pt off 0.8pt}] node [right] {$T$};
\draw (-0.2,4) -- (26*\width+0.2, 4) [dotted] node [right] {$\tfrac54T$};
\draw (-0.2,4.8) -- (26*\width+0.2, 4.8) node [right] {$\tfrac32T$};

\draw [<->] (0,-0.8) -- (10*\width-0.01,-0.8) node[fill=white, pos=.5] {$|I_{\exp}^0|$};
\draw [<->] (10*\width+0.01,-0.8) -- (26*\width,-0.8) node[fill=white, pos=.5] {$m^\#$};
\draw [<->] (0,-1.2) -- (26*\width,-1.2) node[fill=white, pos=.5] {$m$};

\node (mu) at (20.5*\width,-0.5*\width) {$\mu$};

\draw (10*\width,0) -- (10*\width,4.8);
\draw (26*\width,0) -- (26*\width,4.8);

\draw [dashed,white] (10*\width,0) -- (10*\width,4.8) -- (26*\width,4.8) -- (26*\width,0) -- (10*\width,0);

\draw (18*\width,2.4) node {};

\end{tikzpicture}
}
			\caption{An example solution after using \Cref{preemptive:large_machines} with $I_{\exp}^+ = \set{1,2}$}
			\label{fig:preemptive:pumping_lemma}
		\end{figure}
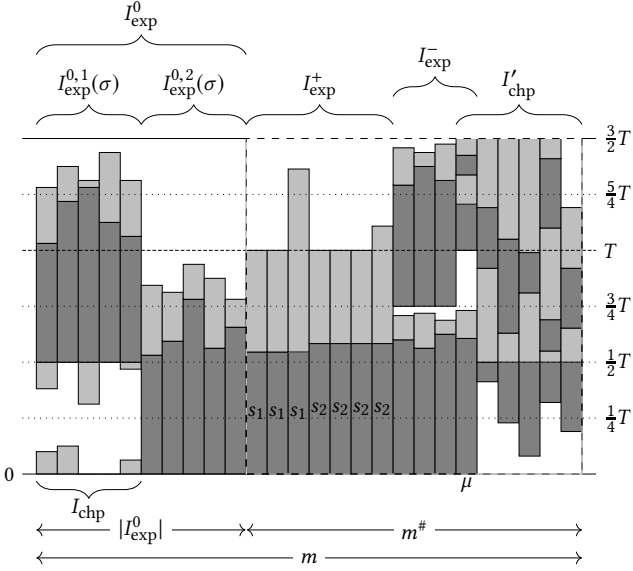
		
		To look at this issue in more detail we call $R$ the processing time of the residual load in $\sigma$ and we find
		$
		R = \sum_{i \in I_{\exp}^+ \cup I_{\exp}^-} (\lambda^{\sigma}_i s_i + P(C_i)) + \sum_{i \in I_{\chp}'} (s_i + P(C_i')).
		$
		Applying $\lambda^{\sigma}_i \geq 2$ for all $i \in I_{\exp}^{0,2}(\sigma)$ we can use
		\[
		\sum_{i \in I_{\exp}^{0,2}(\sigma)} \!\!\!\!\! (\lambda^{\sigma}_i s_i + P(C_i))
		\geq \!\!\!\!\!\!\sum_{i \in I_{\exp}^{0,2}(\sigma)} \!\!\!\!\! (2s_i + P(C_i))
		\geq \!\!\!\!\!\! \sum_{i \in I_{\exp}^{0,2}(\sigma)} \!\!\!\!\! (T + P(C_i))
		\geq |I_{\exp}^{0,2}(\sigma)|T
		\]
		and $I_{\exp} \setminus I_{\exp}^{0,1}(\sigma) = I_{\exp}^{0,2}(\sigma) \cupdot I_{\exp}^+ \cupdot I_{\exp}^-$ to see that
		\begin{align*}
			&(m - |I_{\exp}^{0,1}(\sigma)|)T\\ &\geq \!\sum_{i \in (I_{\exp} \setminus I_{\exp}^{0,1}(\sigma))} \!\!\!\!\!\! (\lambda^{\sigma}_i s_i + P(C_i)) + \sum_{i \in I_{\chp}'} (s_i + P(C_i')) \qquad\quad \quad\,\,\,\, \text{// \eqref{eq:appendixA:residual_machines_1}}\\
			&= \!\!\! \sum_{i \in I_{\exp}^{0,2}(\sigma)} \!\!\! (\lambda^{\sigma}_i s_i + P(C_i)) + \!\!\!\!\!\! \sum_{i \in I_{\exp}^+ \cup I_{\exp}^-} \!\!\!\!\!\! (\lambda^{\sigma}_i s_i + P(C_i)) + \sum_{i \in I_{\chp}'} (s_i + P(C_i'))\\
			&\geq |I_{\exp}^{0,2}(\sigma)|T + R
		\end{align*}
		and thus it follows
		$
		(m - |I_{\exp}^0|)T = (m - |I_{\exp}^{0,1}(\sigma)|) T - |I_{\exp}^{0,2}(\sigma)|T\geq R
		$.
		Hence, the residual $m-|I_{\exp}^0|$ machines provide a processing time of at least $R$. So we can build a residual instance $I^\#$ for the residual $m^\# := m - |I_{\exp}^0|$ machines to place $J^\# := J(I_{\exp}^+ \cup I_{\exp}^-) \cup \bigcup_{i \in I_{\chp}'} C_i'$ with $C_i^\# := C_i$ for all expensive classes $i \in I_{\exp}^+ \cup I_{\exp}^-$ as well as $C_i^\# := C_i'$ for all cheap classes $i \in I_{\chp}'$. Apparently $I^\#$ is a nice instance and it actually holds the requirements of \Cref{preemptive:simple} \ref{preemptive:simple:decision:T_check_true}. In more detail we obtain that
		\begin{align*}
			m^\#\cdot T &= (m - |I_{\exp}^0|)T\\
			&\ge R
			= \sum_{i \in (I_{\exp}^+ \cup I_{\exp}^-)} (\lambda^{\sigma}_i s_i + P(C_i)) + \sum_{i \in I_{\chp}'} (s_i + P(C_i'))\\
			&\quad\!\quad\ge P(J^\#) + \sum_{i \in I_{\exp}^+} \! \alpha_i' s_i + \!\!\!\!\!\! \sum_{i \in I_{\exp}^- \cup I_{\chp}'} \!\!\!\!\!\!\! s_i \quad\quad\, \text{// \Cref{lemma:a_i-lower-bound}}, \alpha_i \ge \alpha_i'
		\end{align*}
		and with \Cref{lemma:machine_number_expensive_classes_multiplicities,lemma:a_i-lower-bound} and $\alpha_i \ge \alpha_i', \alpha_i \ge 1$ we get that
		\begin{align*}
		m^\# = m - |I_{\exp}^0|
		&\ge \sum_{i \in I_{\exp}} \alpha_i - |I_{\exp}^0|\\
		&\ge \sum_{i \in I_{\exp}^+} \alpha_i' + |I_{\exp}^-|
		\geq \sum_{i \in I_{\exp}^+} \alpha_i' + \ceil*{\tfrac12|I_{\exp}^-|}.
		\end{align*}
		So \Cref{preemptive:simple} leads us to use \Cref{preemptive:algorithm:simple} to complete our schedule feasibly with a makespan of at most $\tfrac32T$. \Cref{fig:preemptive:pumping_lemma} illustrates the use of \Cref{preemptive:algorithm:simple} with dashed lines around the area of the $m^\#$ last machines.
	\end{proof}

	
	\section{Splittable Scheduling}
	\label{section:splittable}
	
	For this section let $T_{\min} := \max\set{\tfrac1m N, s_{\max}}$ where $N = \sum_{i=1}^c s_i + \sum_{j \in J} t_j$.
	Let $I$ be an instance and let $T \geq T_{\min}$ be a makespan.
	We describe the algorithm in two steps.

	\para{Step 1.} First we place all jobs of expensive classes. We define a wrap template $\omega^{(i)}$ of length $|\omega^{(i)}| = \ceil{2P(C_i)/T} = \beta_i$ for each class $i \in I_{\exp}$ as follows. Let $\omega^{(i)}_1 = (u_i,0,s_i+\tfrac12T)$ and $\omega^{(i)}_{1+r} = (u_i+r,s_i,s_i+\tfrac12T)$ for $1 \leq r < \beta_i$. Here the first machines $u_i$ have to be chosen distinct to all machines of the other wrap templates. We convert the expensive classes $i \in I_{\exp}$ into simple wrap sequences $Q^{(i)} = [s_i,C_i]$ consisting of an initial setup $s_i$ followed by an arbitrary order of all jobs in $C_i$. For all $i \in I_{\exp}$ we use $\textsc{Wrap}(Q^{(i)},\omega^{(i)})$ to wrap $Q^{(i)}$ into $\omega^{(i)} = (\omega_1^{(i)},\dots,\omega_{\beta_i}^{(i)})$. Remark that $\textsc{Wrap}$ places a setup $s_i$ at time $0$ on each machine $u_i+l$ where $1 \leq l < \beta_i$.
	See \Cref{fig:alg_splittable_exp} for an example.

	\para{Step 2.}
	The next and last step is to place the jobs of cheap classes. Let $\bar{u}_i$ be the \emph{last} machine used to wrap a sequence $Q^{(i)}$ in the previous step, i.e. $\bar{u}_i = u_i + \beta_i - 1$ for all $i \in I_{\exp}$. The idea is to use the free time left on the machines $\bar{u}_i$ while reserving a time of exactly $\tfrac12T$ for a cheap setup. Once these machines are filled, we turn to unused machines. In more detail we define one wrap template $\omega$ and one wrap sequence $Q$ to place all jobs $J(I_{\chp})$ as follows.
	Let $m_{\exp} = \sum_{i \in I_{\exp}} \beta_i$ be the number of machines used in step (1) and let $i_1,\dots,i_p \in I_{\exp}$ be all $p$ classes $i \in I_{\exp}$ that hold $\load(\bar{u}_i) < T$. We define $\omega_l = (\bar{u}_{i_l},\load(\bar{u}_{i_l}) + \tfrac12T,\tfrac32T)$ for all $1 \leq l \leq p$. To fill the residual (and empty) $k = m - m_{\exp}$ machines $r_1, \dots, r_k \in [m]$ we set $\omega_{p+l} = (r_l,\tfrac12T,\tfrac32T)$ for all $1 \leq l \leq k$.	
	The wrap sequence $Q = [s_i,C_i]_{i \in I_{\chp}}$ simply consists of all jobs of $J(I_{\chp}) = \bigcup_{i \in I_{\chp}} C_i$ with an initial setup $s_i$ before all the jobs of $C_i$ for a cheap class $i \in I_{\chp}$. As predicted, we wrap $Q$ into $\omega = (\omega_1,\dots,\omega_p,\omega_{p+1},\dots,\omega_{p+m-m_{\exp}})$ using $\textsc{Wrap}(Q, \omega)$.
	See \Cref{fig:alg_splittable_exp_chp} for an example.
	
	\subsection{Analysis}
	\label{section:splittable:analysis}
	
	We want to show the following lemma.
	
	\restatelemmasplittablealgorithm
	\begin{proof}
		\ref{lemma:splittable:algorithm:T_check_false}.
		We show that $T \geq \OPT_{\op{split}}$ implies $mT \geq L_{\op{split}}$ and $m \geq m_{\exp}$. Let $T \geq \OPT_{\op{split}}$. Then there is a feasible schedule $\sigma$ with makespan $T$ and $\load(\sigma) = \sum_{i=1}^c (\lambda^{(\sigma)}_i s_i +P(C_i))$. Due to \Cref{lemma:a_i-lower-bound}, we have
		\begin{align*}
		mT \geq \load(\sigma) &= \sum_{i=1}^c (\lambda_i^{\sigma} s_i + P(C_i))\\
		&\geq P(J) + \sum_{i=1}^c \alpha_i s_i \geq P(J) + \!\!\sum_{i \in I_{\chp}}\!\! s_i + \!\!\sum_{i \in I_{\exp}}\!\! \beta_i s_i = L_{\op{split}}.
		\end{align*}
		Also $m \geq m_{\exp}$ is a direct consequence of \Cref{lemma:a_i-lower-bound,lemma:machine_number_expensive_classes_multiplicities}.
		
		\ref{lemma:splittable:algorithm:T_check_true}.
		Let $mT \geq L_{\op{split}}$ and $m \geq m_{\exp}$. Note that the number of machines used in step (1) is $\sum_{i \in I_{\exp}} \beta_i = m_{\exp} \leq m$ and hence we have enough machines but we have to check for all $i \in I_{\exp}$ that the wrap template $\omega^{(i)}$ in step (1) is suitable to wrap $Q^{(i)} = [s_i,C_i]$ into it, i.e. $S(\omega^{(i)}) \geq \load(Q^{(i)})$. This is true since
		\[
		S(\omega^{(i)})
		= s_i + \beta_i \cdot \tfrac12T
		= s_i + \ceil*{\frac{P(C_i)}{\tfrac12T}} \cdot \tfrac12T
		\geq s_i + P(C_i)
		= \load(Q^{(i)}).
		\]
		Each wrap template $\omega^{(i)}$ is filled with exactly one class $i$ and reserves a time $s_i$ below each gap. So there is enough space to place a setup $s_i$ below all gaps of $\omega^{(i)}$.
		It remains to show that the wrap template $\omega$ in step (2) is suitable to wrap $Q$ into it. This needs a bit more effort. Apparently for all $i \in I_{\exp}$ the load $\load(\bar{u}_i)$ of the last machine $\bar{u}_i = u_i + \beta_i - 1$ holds
		\[
		\beta_i s_i + P(C_i) = (\beta_i-1)(s_i+\tfrac12T) + \load(\bar{u}_i) \geq (\beta_i - 1)T + \load(\bar{u}_i)
		\]
		since $s_i \geq \tfrac12T$. Hence, if the last machine is filled to at least $T$, i.e. $\load(\bar{u}_i) \geq T$, we obtain that
		$\beta_i s_i + P(C_i) \geq \beta_i T$
		and otherwise it follows that
		$\beta_i s_i + P(C_i) + T - \load(\bar{u}_i) \geq \beta_i T$.
		These two inequalities imply that 
		\begin{align*}
			L'
			&:= \sum_{\substack{i \in I_{\exp}\\\load(\bar{u}_i) < T}} (T-\load(\bar{u}_i)) + \sum_{i \in I_{\exp}} (\beta_i s_i + P(C_i))\\
			&= \sum_{\substack{i \in I_{\exp}\\\load(\bar{u}_i) < T}} (\beta_i s_i + P(C_i) + T-\load(\bar{u}_i)) + \sum_{\substack{i \in I_{\exp}\\\load(\bar{u}_i) \geq T}} (\beta_i s_i + P(C_i))\\
			&\geq \sum_{\substack{i \in I_{\exp}\\\load(\bar{u}_i) < T}} \beta_i T + \sum_{\substack{i \in I_{\exp}\\\load(\bar{u}_i) \geq T}} \beta_i T = m_{\exp}T
		\end{align*}
		and we use this inequality to show that $\omega$ is suitable to wrap $Q$ since
		\begin{align*}
			S(\omega)
			&= \sum_{\substack{i \in I_{\exp}\\\load(\bar{u}_i) < T}} (T-\load(\bar{u}_i)) + (m - m_{\exp})T\\
			&\geq \sum_{\substack{i \in I_{\exp}\\\load(\bar{u}_i) < T}} (T-\load(\bar{u}_i)) + L_{\op{split}} - m_{\exp}T && \text{// } mT \geq L_{\op{split}}\\
			&= L' + \sum_{i \in I_{\chp}} (s_i + P(C_i)) - m_{\exp}T \,\,\,\geq \,\,\,\load(Q). && \text{// } L' \geq m_{\exp}T
		\end{align*}
		One can easily confirm that the reserved processing time of $\tfrac12T$ below all used gaps is sufficient. In detail, we only place jobs of cheap classes in step (2). So the call of $\textsc{Wrap}(Q,\omega)$ needs a time of at most $\tfrac12T$ to place cheap setups below the gaps.
		Hence, the computed schedule is feasible and this proves \Cref{lemma:splittable:algorithm}.
		
		Nevertheless we still need to analyze the running time.
		Apparently the running time directly depends on the running time of $\textsc{WrapSplit}$ as follows. For step (1) we get a running time of $\oh(1) + \sum_{i \in I_{\exp}} \oh(|Q^{(i)}|+|\omega^{(i)}|) = \sum_{i \in I_{\exp}} \oh(|C_i|+\beta_i) = \oh(|J(I_{\exp})|+\sum_{i \in I_{\exp}}\beta_i) \leq \oh(n+m)$ due to \Cref{lemma:wrapsplit_running_time} since $|Q_i| = 1+|C_i|$ and $|\omega_i| = \beta_i$ for all $i \in I_{\exp}$.		
		To optimize the running time we find another implementation\footnote{A similar idea was already mentioned by Jansen et al. in \cite{DBLP:conf/innovations/JansenKMR19}} of $\textsc{Split}$ (and $\textsc{Wrap}$) for our use case. As mentioned before we allow that a schedule consists of machine configurations with given multiplicities.
		In fact, there is a more efficient implementation of $\textsc{Split}$ for ranges of wrap sequences where all gaps start and end at equal times, i.e. $a_{r_1} = a_{r_2}$ and $b_{r_1} = b_{r_2}$. Apparently $\textsc{Split}$ will place at most three different gap types (or gap configurations) for each job (piece) in such ranges of \emph{parallel} gaps. To see that let $0 \leq a < b$ describe the gaps and consider a job (piece) $j$ which we start to place at time $t \in [a,b)$. If $j$ is split by $\textsc{Split}$ at most once, we obviously have at most two used gaps; hence, we have at most two different gap configurations. If $j$ is split at least two times $j$ is split into a first piece with processing time $b-t$ followed by $\mu_j := \floor{(t_j-(b-t))/(b-a)}$ gaps filled with processing time $b-a$ and a last gap starting with a piece of time $t_j - (b-t) - \mu_j(b-a)$. These define at most three different gap configurations. Since the multiplicity $\mu_j$ of the in between gaps can be computed in constant time, we can compute these three gap configurations and its multiplicities in constant time.
		So we get a running time of $\oh(n+c) = \oh(n)$ for step (1).
		For step (2) we apply this technique only for the $m-m_{\exp}$ last gaps $\omega_{p+1},\dots,\omega_{p-m-m_{\exp}}$ which are parallel in our sense. Remark that $p \leq c$ to see that the running time is $\oh(c+|Q|) \leq \oh(n)$.
		Hence, we get a total running time of $\oh(n)$.
	\end{proof}

	\section{Non-Preemptive Scheduling}
	\label{nonpreemptivescheduling}
	
	Doing \emph{non-preemptive} scheduling we do \emph{not} allow jobs to be preempted. Even an optimal schedule needs to place at least one setup to schedule a job on a machine, so remember \Cref{non-preemptive_simple_lower_bound} which says $\OPT_{\op{nonp}} \geq \max_{i \in [c]} (s_i + t_{\max}^{(i)})$ where $t_{\max}^{(i)} = \max_{j \in C_i} t_j$.
	
	\noindent Analogous to preemptive scheduling we assume $m < n$. For this section let $T_{\min} := \max\set{\tfrac1m N, \max_{i\in[c]}(s_i+t_{\max}^{(i)})}$ where $N = \sum_{i=1}^c s_i + \sum_{j \in J} t_j$ and $t_{\max}^{(i)} = \max_{j \in C_i} t_j$.
	
	\restatenonpreemptiverunningtime
	
	Let $I$ be an instance and $T \geq T_{\min}$ be a makspan. For later purposes we split the jobs into big and small ones. In more detail, let $J_+ = \set{j \in J|t_j > \tfrac12T}$ and $J_- = \set{j \in J| t_j \leq \tfrac12T}$. In the following we will look at three subsets of $J$. They are $J_+$, $J(I_{\exp}) = \bigcup_{i \in I_{\exp}} C_i$ as well as $K := \bigcup_{i \in I_{\chp}}\set{j \in C_i \cap J_-| s_i + t_j > \tfrac12T}$ and one can easily see that they are in pairs disjoint. Let $L = J_+ \cupdot J(I_{\exp}) \cupdot K$.
	
	\begin{note}\label{non-preemptive:L-description}
		It is true that $L = \bigcup_{i \in [c]} \set{j \in C_i| s_i + t_j > \tfrac12T}$.\qed
	\end{note}
	
	\noindent We find the following minimum number of machines for each class.
	For all $i \in [c]$ let
	\[m_i = \begin{cases}
	\ceil*{\frac{P(C_i)}{T - s_i}} = \alpha_i & : i \in I_{\exp}\\
	|C_i \cap J_+| + \ceil*{\frac{P(C_i \cap K)}{T - s_i}} & : i \in I_{\chp}
	\end{cases}.\]

	\begin{note}
		\label{non-preemptive:differentmachinesforL}
		Different jobs in $L$ of different classes have to be scheduled on different machines. Furthermore, every job in $J_+ \subseteq L$ needs an own machine.
	\end{note}
	\begin{proof}
		Assume that two jobs $j_1, j_2 \in L$ of different classes $i_1, i_2 \in [c]$ are scheduled feasibly on one machine $u$, i.e. $\load(u) \leq T$. Due to \Cref{non-preemptive:L-description}, we have $s_{i_1} + t_{j_1} > \tfrac12T$ as well as $s_{i_2} + t_{j_2} > \tfrac12T$. To schedule $j_1$ and $j_2$ on $u$ it needs at least one setup time for both of them since $i_1 \neq i_2$. So we get a total load of $\load(u) \geq (s_{i_1} + t_{j_1}) + (s_{i_2} + t_{j_2}) > \tfrac12T + \tfrac12T = T$, a contradiction.
		Now assume that $j_1, j_2 \in J_+$ are jobs of a common class $i \in [c]$ and they are scheduled feasibly on one machine $u$, i.e. $\load(u) \leq T$. Since both jobs are in the same class $i$ we need only one setup time $s_i$, but the total load is $\load(u) \geq s_i + t_{j_1} + t_{j_2} > s_i + T > T$ since $t_{j_1} > \tfrac12T$ and $t_{j_2} > \tfrac12T$. Again, this is a contradiction.
	\end{proof}
	
	\begin{lemma}\label{non-preemptive:minimal_number_of_machines}
		Let $\sigma$ be a feasible schedule with makespan $T$. Then $\sigma$ needs at least $m_i$ different machines to schedule a class $i \in [c]$ and in total $\sigma$ needs at least $\sum_{i=1}^c m_i$ different machines.\qed
	\end{lemma}
	\noindent\Cref{non-preemptive:minimal_number_of_machines} is a consequence of \Cref{lemma:a_i-lower-bound,lemma:machine_number_expensive_classes_multiplicities,non-preemptive:differentmachinesforL}.
	\begin{algorithm}
		\caption{A $\tfrac32$-dual Approximation for Non-Preemptive Scheduling}
		\label{alg:32approx}
		\begin{enumerate}[1.]
			\item Schedule all jobs of $L = \bigcup_{i \in [c]} \set{j \in C_i| s_i + t_j > \tfrac12T}$ on $m_i$ machines for each class $i$
			\label{alg:32approx:L}
			\item Schedule as many jobs as possible of $J \setminus L = \bigcup_{i \in [c]} \set{j \in C_i| s_i + t_j \leq \tfrac12T}$ on used machines without adding new setup times
			\label{alg:32approx:fill}
			\item Take one new setup time for each remaining class and place the remaining jobs greedily
			\label{alg:32approx:greedy}
			\item Make the schedule non-preemptive and add setup times as needed
			\label{alg:32approx:modify}
		\end{enumerate}
	\end{algorithm}
	We look at \Cref{alg:32approx} in more detail.
	
	\Sref{alg:32approx:L} We schedule the jobs of $L$. For every class $i \in [c]$ do the following. If $i$ is expensive, we place all jobs of $C_i$ preemptively (until $T$) with one initial setup time $s_i$ at the beginning of each of the required machines. In detail, we use a wrap template $\omega^{(i)} = (\omega^{(i)}_1,\dots,\omega^{(i)}_{\alpha_i})$ of length $|\omega{(i)}| = \alpha_i = \ceil{P(C_i)/(T-s_i)}$ with $\omega^{(i)}_1 = (u_i,0,T)$ and $\omega^{(i)}_{1+r} = (u_i+r,s_i,T)$ for a first machine $u_i$ such that all used machines are distinct and $1 \leq r < \alpha_i$. We use this wrap template to schedule a simple wrap sequence $Q^{(i)} = [s_i,C_i]$ with $\textsc{Wrap}(Q^{(i)},\omega^{(i)})$.
	If $i$ is cheap, we place all the jobs $j \in C_i \cap J_+$ with an initial setup time $s_i$ on a single unused machine $v^{(i)}_k$, i.e. the load of such a machine will be $\tfrac12T < s_i + t_j \leq T$. After that we place all jobs of $C_i \cap K$ preemptively (until $T$) on unused machines with one initial setup time $s_i$ at the beginning of each of the required machines. As before, we use a simple wrap template $\omega^{(i)} = (\omega^{(i)}_1,\dots,\omega^{(i)}_{\alpha_i})$ of length $|\omega^{(i)}| = \alpha_i$ with $\omega^{(i)}_1 = (u_i,0,T)$ and $\omega^{(i)}_{1+r} = (u_i+r,s_i,T)$ for a first machine $u_i$ such that all used machines are distinct and $1 \leq r < \alpha_i$. We use $\omega^{(i)}$ to wrap a wrap sequence $Q^{(i)} = [s_i,C_i \cap K]$ with $\textsc{Wrap}(Q^{(i)},\omega^{(i)})$.
	For all classes $i \in [c]$ let $\bar{u}_i = u_i + m_i - 1$ be the last machine used to wrap the sequence $Q^{(i)}$.
	For an example schedule after this step see \Cref{fig:non-preemptive:L}. The dashed lines indicate the wrap templates and the wrap sequences are filled green (dark if preempted).
	
	\begin{figure}[h]
		\centering
		\scalebox{0.94}{\begin{tikzpicture}

\def\diff{0.45}
\def\width{0.5}

\draw (-0.2,0) node [left] {$0$} -- (4*\diff+12*\width+0.2, 0);
\draw (-0.2,3.2) -- (4*\diff+12*\width+0.2, 3.2) node [right] {$T$};

\draw [-, decorate, decoration={brace,amplitude=8pt}] (0,4) -- node[above=7pt] {$I_{\exp}$} (2*\diff+5*\width,4);

\draw [-, decorate, decoration={brace,amplitude=8pt}] (2*\diff+5*\width,4) -- node[above=7pt] {$I_{\chp}$} (4*\diff+12*\width,4);

\draw [-, decorate, decoration={brace,amplitude=8pt}] (\diff,3.3) -- node[above=7pt] {$C_1$} (\diff+5*\width,3.3);

\draw [-, decorate, decoration={brace,amplitude=8pt}] (3*\diff+5*\width,3.3) -- node[above=7pt] {$C_2 \cap J_+$} (3*\diff+8*\width,3.3);

\draw [-, decorate, decoration={brace,amplitude=8pt}] (3*\diff+8*\width,3.3) -- node[above=7pt] {$C_2 \cap K$} (3*\diff+12*\width,3.3);

\DrawMachines{
{
	1/{\hbox to 0.7em{.\hss.\hss.}}/
}/\diff/,
{
	{1.8/3.2}/{$s_1$}/,
	{0.2/3.2}//,
	{0.3/3.2}//,
	{0.15/3.2}//,
	{0.4/3.2}//,
	{0.05/3.2}//,
	{0.3/3.2}//[fill=green!30!gray]
}/\width/[fill=green!30!white],
{
	{1.8/3.2}/{$s_1$}/[fill=lightgray],
	{0.1/3.2}//[fill=green!30!gray],
	{0.6/3.2}//,
	{0.3/3.2}//,
	{0.15/3.2}//,
	{0.25/3.2}//
}/\width/[fill=green!30!white],
{
	{1.8/3.2}/{$s_1$}/[fill=lightgray],
	{0.3/3.2}//,
	{0.15/3.2}//,
	{0.35/3.2}//,
	{0.4/3.2}//,
	{0.05/3.2}//,
	{0.1/3.2}//,
	{0.05/3.2}//[fill=green!30!gray]
}/\width/[fill=green!30!white],
{
	{1.8/3.2}/{$s_1$}/[fill=lightgray],
	{0.6/3.2}//[fill=green!30!gray],
	{0.1/3.2}//,
	{0.05/3.2}//,
	{0.35/3.2}//,
	{0.2/3.2}//,
	{0.1/3.2}//[fill=green!30!gray]
}/\width/[fill=green!30!white],
{
	{1.8/3.2}/{$s_1$}/[fill=lightgray],
	{0.05/3.2}//[fill=green!30!gray],
	{0.1/3.2}//,
	{0.25/3.2}//,
	{0.07/3.2}//,
	{0.17/3.2}//
}/\width/[fill=green!30!white],
{
	1/{\hbox to 0.7em{.\hss.\hss.}}/
}/\diff/,
{
	1/{\hbox to 0.7em{.\hss.\hss.}}/
}/\diff/,
{
	{1/3.2}/{$s_2$}/,
	{1.9/3.2}//
}/\width/[fill=lightgray],
{
	{1/3.2}/{$s_2$}/,
	{1.65/3.2}//
}/\width/[fill=lightgray],
{
	{1/3.2}/{$s_2$}/,
	{1.8/3.2}//
}/\width/[fill=lightgray],
{
	{1/3.2}/{$s_2$}/,
	{0.65/3.2}//,
	{0.7/3.2}//,
	{0.75/3.2}//,
	{0.1/3.2}//[fill=green!30!gray]
}/\width/[fill=green!30!white],
{
	{1/3.2}/{$s_2$}/[fill=lightgray],
	{0.8/3.2}//[fill=green!30!gray],
	{1.0/3.2}//,
	{0.4/3.2}//[fill=green!30!gray]
}/\width/[fill=green!30!white],
{
	{1/3.2}/{$s_2$}/[fill=lightgray],
	{0.5/3.2}//[fill=green!30!gray],
	{0.8/3.2}//,
	{0.7/3.2}//,
	{0.2/3.2}//[fill=green!30!gray]
}/\width/[fill=green!30!white],
{
	{1/3.2}/{$s_2$}/[fill=lightgray],
	{0.7/3.2}//[fill=green!30!gray],
	{0.9/3.2}//
}/\width/[fill=green!30!white],
{
	1/{\hbox to 0.7em{.\hss.\hss.}}/
}/\diff/
}

\node (mu) at (\diff+0.5*\width,-0.35) {$u_1$};
\node (mu) at (\diff+4.5*\width,-0.35) {$\bar{u}_1$};

\node (mu) at (3*\diff+5.5*\width,-0.35) {$v^{(2)}_1$};
\node (mu) at (3*\diff+6.5*\width,-0.35) {$v^{(2)}_2$};
\node (mu) at (3*\diff+7.5*\width,-0.35) {$v^{(2)}_3$};
\node (mu) at (3*\diff+8.5*\width,-0.35) {$u_2$};
\node (mu) at (3*\diff+11.5*\width,-0.35) {$\bar{u}_2$};

\draw [<->] (\diff,-0.8) -- (\diff+5*\width,-0.8) node[fill=white, pos=.5] {$m_1$};

\draw [<->] (3*\diff+5*\width,-0.8) -- (3*\diff+12*\width,-0.8) node[fill=white, pos=.5] {$m_2$};

\draw [<->] (0,-1.1) -- (4*\diff+12*\width,-1.1) node[fill=white, pos=.5] {$m'$};

\draw [color=white,dashed]
(\diff,0) -- (\diff,3.2) -- (\diff+5*\width,3.2) -- (\diff+5*\width,1.8) -- (\diff+\width,1.8) -- (\diff+\width,0) -- (\diff,0);

\draw [color=white,dashed]
(3*\diff+8*\width,0) -- (3*\diff+8*\width,3.2) -- (3*\diff+12*\width,3.2) -- (3*\diff+12*\width,1) -- (3*\diff+9*\width,1) -- (3*\diff+9*\width,0) -- (3*\diff+8*\width,0);

\end{tikzpicture}
}
		\caption{An example situation after step \ref{alg:32approx:L} of \Cref{alg:32approx} with $1 \in I_{\exp}$ and $2 \in I_{\chp}$}
		\label{fig:non-preemptive:L}
	\end{figure}
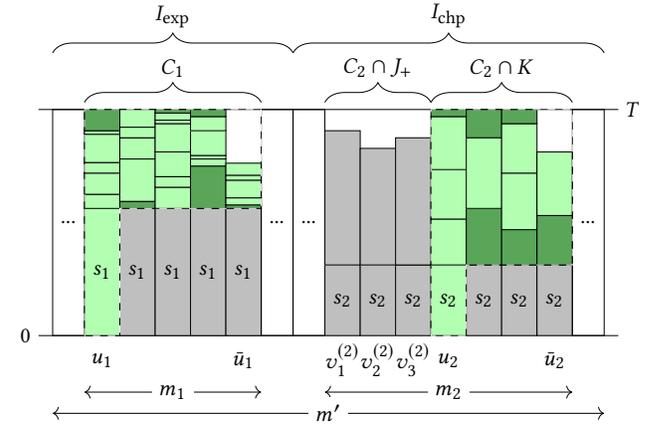
	
	\Sref{alg:32approx:fill} Now we will place as many jobs as possible of $J \setminus L$ without adding new machines or setup times. Note that there is at most one setup time on a used machine so far. Let $v^{(i)}_1,\dots,v^{(i)}_{k_i}$ be the machines used to schedule $C_i \cap J_+$ in step \ref{alg:32approx:L}, i.e. $k_i = |C_i \cap J_+|$, and let $v^{(i)}_{k_i+1} = \bar{u}_i$. For every cheap class $i \in I_{\chp}$ set $C_i' \leftarrow C_i \setminus L$ and start the following loop.	
	Let $j \in C_i'$ and find a used machine $u = v^{(i)}_k$ for $1 \leq k \leq k_i+1$ that has a load $L(u) < T$. If such a machine can not be found, the remaining jobs $C_i'$ will be placed in step \ref{alg:32approx:greedy}. If $L(u) + t_j \leq T$ place $j$ on top of machine $u$ and set $C_i' \leftarrow C_i' \setminus \set{j}$. Otherwise split $j$ into two new job pieces $j_1$, $j_2$ (of class $i$) such that $t_{j_1} = T - L(u)$ as well as $t_{j_2} = t_j - t_{j_1}$ and place $j_1$ on top of machine $u$ and set $C_i' \leftarrow (C_i' \setminus \set{j}) \cup \set{j_2}$. Furthermore, we save $j$ as the \emph{parent job} of the new job pieces $j_1$ and $j_2$, i.e. we set $\parent(j_1) \gets j$ as well as $\parent(j_2) \gets j$. Process $j_2$ next in the loop.
	See \Cref{fig:non-preemptive:first_fillup} for an example of this step. The new placed jobs of the cheap class $2 \in I_{\chp}$ are colored blue (and dark if preempted). Remark that jobs may be split more than once; in fact, one can see that there is a job of class $2$ that is split onto the machines $v^{(2)}_2$, $v^{(2)}_3$ and $v^{(2)}_4 = \bar{u}_2$.
	
	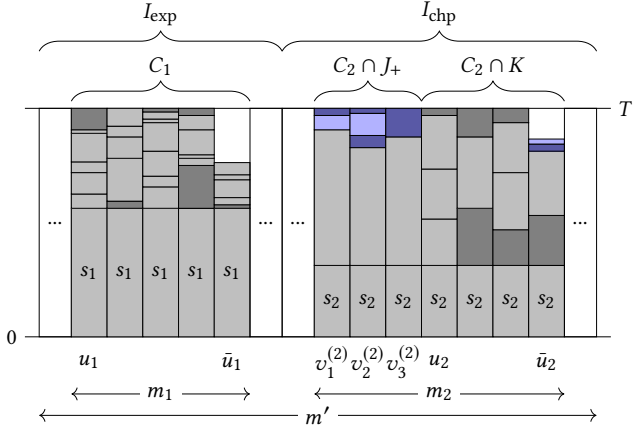
\begin{figure}[h]
		\centering
		\scalebox{0.95}{\begin{tikzpicture}

\def\diff{0.45}
\def\width{0.5}

\draw (-0.2,0) node [left] {$0$} -- (4*\diff+12*\width+0.2, 0);
\draw (-0.2,3.2) -- (4*\diff+12*\width+0.2, 3.2) node [right] {$T$};

\draw [-, decorate, decoration={brace,amplitude=8pt}] (0,4) -- node[above=7pt] {$I_{\exp}$} (2*\diff+5*\width,4);

\draw [-, decorate, decoration={brace,amplitude=8pt}] (2*\diff+5*\width,4) -- node[above=7pt] {$I_{\chp}$} (4*\diff+12*\width,4);

\draw [-, decorate, decoration={brace,amplitude=8pt}] (\diff,3.3) -- node[above=7pt] {$C_1$} (\diff+5*\width,3.3);

\draw [-, decorate, decoration={brace,amplitude=8pt}] (3*\diff+5*\width,3.3) -- node[above=7pt] {$C_2 \cap J_+$} (3*\diff+8*\width,3.3);

\draw [-, decorate, decoration={brace,amplitude=8pt}] (3*\diff+8*\width,3.3) -- node[above=7pt] {$C_2 \cap K$} (3*\diff+12*\width,3.3);

\DrawMachines{
	{
		1/{\hbox to 0.7em{.\hss.\hss.}}/
	}/\diff/,
	{
		{1.8/3.2}/{$s_1$}/,
		{0.2/3.2}//,
		{0.3/3.2}//,
		{0.15/3.2}//,
		{0.4/3.2}//,
		{0.05/3.2}//,
		{0.3/3.2}//[fill=gray]
	}/\width/[fill=lightgray],
	{
		{1.8/3.2}/{$s_1$}/,
		{0.1/3.2}//[fill=gray],
		{0.6/3.2}//,
		{0.3/3.2}//,
		{0.15/3.2}//,
		{0.25/3.2}//
	}/\width/[fill=lightgray],
	{
		{1.8/3.2}/{$s_1$}/,
		{0.3/3.2}//,
		{0.15/3.2}//,
		{0.35/3.2}//,
		{0.4/3.2}//,
		{0.05/3.2}//,
		{0.1/3.2}//,
		{0.05/3.2}//[fill=gray]
	}/\width/[fill=lightgray],
	{
		{1.8/3.2}/{$s_1$}/,
		{0.6/3.2}//[fill=gray],
		{0.1/3.2}//,
		{0.05/3.2}//,
		{0.35/3.2}//,
		{0.2/3.2}//,
		{0.1/3.2}//[fill=gray]
	}/\width/[fill=lightgray],
	{
		{1.8/3.2}/{$s_1$}/,
		{0.05/3.2}//[fill=gray],
		{0.1/3.2}//,
		{0.25/3.2}//,
		{0.07/3.2}//,
		{0.17/3.2}//
	}/\width/[fill=lightgray],
	{
		1/{\hbox to 0.7em{.\hss.\hss.}}/
	}/\diff/,
	{
		1/{\hbox to 0.7em{.\hss.\hss.}}/
	}/\diff/,
	{
		{1/3.2}/{$s_2$}/,
		{1.9/3.2}//,
		{0.2/3.2}//[fill=blue!30!white],
		{0.1/3.2}//[fill=blue!30!gray]
	}/\width/[fill=lightgray],
	{
		{1/3.2}/{$s_2$}/,
		{1.65/3.2}//,
		{0.17/3.2}//[fill=blue!30!gray],
		{0.31/3.2}//[fill=blue!30!white],
		{0.07/3.2}//[fill=blue!30!gray]
	}/\width/[fill=lightgray],
	{
		{1/3.2}/{$s_2$}/,
		{1.8/3.2}//,
		{0.4/3.2}//[fill=blue!30!gray]
	}/\width/[fill=lightgray],
	{
		{1/3.2}/{$s_2$}/,
		{0.65/3.2}//,
		{0.7/3.2}//,
		{0.75/3.2}//,
		{0.1/3.2}//[fill=gray]
	}/\width/[fill=lightgray],
	{
		{1/3.2}/{$s_2$}/,
		{0.8/3.2}//[fill=gray],
		{1.0/3.2}//,
		{0.4/3.2}//[fill=gray]
	}/\width/[fill=lightgray],
	{
		{1/3.2}/{$s_2$}/,
		{0.5/3.2}//[fill=gray],
		{0.8/3.2}//,
		{0.7/3.2}//,
		{0.2/3.2}//[fill=gray]
	}/\width/[fill=lightgray],
	{
		{1/3.2}/{$s_2$}/,
		{0.7/3.2}//[fill=gray],
		{0.9/3.2}//,
		{0.1/3.2}//[fill=blue!30!gray],
		{0.07/3.2}//[fill=blue!30!white]
	}/\width/[fill=lightgray],
	{
		1/{\hbox to 0.7em{.\hss.\hss.}}/
	}/\diff/
}

\node (mu) at (\diff+0.5*\width,-0.35) {$u_1$};
\node (mu) at (\diff+4.5*\width,-0.35) {$\bar{u}_1$};

\node (mu) at (3*\diff+5.5*\width,-0.35) {$v^{(2)}_1$};
\node (mu) at (3*\diff+6.5*\width,-0.35) {$v^{(2)}_2$};
\node (mu) at (3*\diff+7.5*\width,-0.35) {$v^{(2)}_3$};
\node (mu) at (3*\diff+8.5*\width,-0.35) {$u_2$};
\node (mu) at (3*\diff+11.5*\width,-0.35) {$\bar{u}_2$};

\draw [<->] (\diff,-0.8) -- (\diff+5*\width,-0.8) node[fill=white, pos=.5] {$m_1$};

\draw [<->] (3*\diff+5*\width,-0.8) -- (3*\diff+12*\width,-0.8) node[fill=white, pos=.5] {$m_2$};

\draw [<->] (0,-1.1) -- (4*\diff+12*\width,-1.1) node[fill=white, pos=.5] {$m'$};



\end{tikzpicture}
}
		\caption{The situation after step \ref{alg:32approx:fill} of \Cref{alg:32approx} with $1 \in I_{\exp}$ and $2 \in I_{\chp}$}
		\label{fig:non-preemptive:first_fillup}
	\end{figure}
	
	\Sref{alg:32approx:greedy} Now we cannot schedule a job of $C_i'$ for any $i \in I_{\chp}$ without paying a new setup time $s_i$. However, we can discard classes $i$ without residual load, i.e. $P(C_i') = 0$. So we build a wrap sequence $Q = [s_i,C_i']_{i:P(C_i')>0}$ that only contains classes with non-empty residual load.
	Instead of wrapping $Q$ using a wrap template, we greedily fill up the used machines with a load \emph{less} than $T$ until an item crosses the border $T$. We do not split these critical items but just keep them as they are (\emph{non-preempted}) and turn to the next machine. Once all used machines are filled to at least $T$, we fill up the unused machines in just the same manner.
	In \Cref{fig:non-preemptive:second_fillup} one can see an example situation after this step where the items of $Q$ are colored red (dark if $T$-crossing).
	
	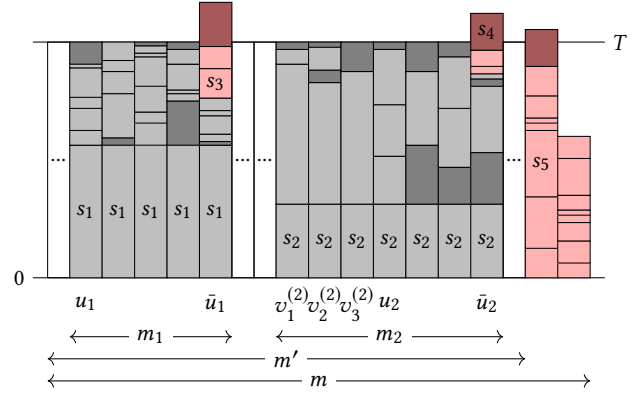
\begin{figure}[h]
		\centering
		\scalebox{0.98}{\begin{tikzpicture}

\def\diff{0.3}
\def\width{0.44}

\draw (-0.2,0) node [left] {$0$} -- (4*\diff+14*\width+0.2, 0);
\draw (-0.2,3.2) -- (4*\diff+14*\width+0.2, 3.2) node [right] {$T$};

\DrawMachines{
{
	1/{\hbox to 0.7em{.\hss.\hss.}}/
}/\diff/,
{
	{1.8/3.2}/{$s_1$}/,
	{0.2/3.2}//,
	{0.3/3.2}//,
	{0.15/3.2}//,
	{0.4/3.2}//,
	{0.05/3.2}//,
	{0.3/3.2}//[fill=gray]
}/\width/[fill=lightgray],
{
	{1.8/3.2}/{$s_1$}/,
	{0.1/3.2}//[fill=gray],
	{0.6/3.2}//,
	{0.3/3.2}//,
	{0.15/3.2}//,
	{0.25/3.2}//
}/\width/[fill=lightgray],
{
	{1.8/3.2}/{$s_1$}/,
	{0.3/3.2}//,
	{0.15/3.2}//,
	{0.35/3.2}//,
	{0.4/3.2}//,
	{0.05/3.2}//,
	{0.1/3.2}//,
	{0.05/3.2}//[fill=gray]
}/\width/[fill=lightgray],
{
	{1.8/3.2}/{$s_1$}/,
	{0.6/3.2}//[fill=gray],
	{0.1/3.2}//,
	{0.05/3.2}//,
	{0.35/3.2}//,
	{0.2/3.2}//,
	{0.1/3.2}//[fill=gray]
}/\width/[fill=lightgray],
{
	{1.8/3.2}/{$s_1$}/,
	{0.05/3.2}//[fill=gray],
	{0.1/3.2}//,
	{0.25/3.2}//,
	{0.07/3.2}//,
	{0.17/3.2}//,
    {0.4/3.2}/$s_3$/[fill=red!30!white],
    {0.3/3.2}//[fill=red!30!white],
    {0.6/3.2}//[fill=red!30!gray]
}/\width/[fill=lightgray],
{
	1/{\hbox to 0.7em{.\hss.\hss.}}/
}/\diff/,
{
	1/{\hbox to 0.7em{.\hss.\hss.}}/
}/\diff/,
{
	{1/3.2}/{$s_2$}/,
	{1.9/3.2}//,
    {0.2/3.2}//,
    {0.1/3.2}//[fill=gray]
}/\width/[fill=lightgray],
{
	{1/3.2}/{$s_2$}/,
	{1.65/3.2}//,
    {0.17/3.2}//[fill=gray],
    {0.31/3.2}//,
    {0.07/3.2}//[fill=gray]
}/\width/[fill=lightgray],
{
	{1/3.2}/{$s_2$}/,
	{1.8/3.2}//,
    {0.4/3.2}//[fill=gray]
}/\width/[fill=lightgray],
{
	{1/3.2}/{$s_2$}/,
	{0.65/3.2}//,
	{0.7/3.2}//,
	{0.75/3.2}//,
	{0.1/3.2}//[fill=gray]
}/\width/[fill=lightgray],
{
	{1/3.2}/{$s_2$}/,
	{0.8/3.2}//[fill=gray],
	{1.0/3.2}//,
	{0.4/3.2}//[fill=gray]
}/\width/[fill=lightgray],
{
	{1/3.2}/{$s_2$}/,
	{0.5/3.2}//[fill=gray],
	{0.8/3.2}//,
	{0.7/3.2}//,
	{0.2/3.2}//[fill=gray]
}/\width/[fill=lightgray],
{
	{1/3.2}/{$s_2$}/,
	{0.7/3.2}//[fill=gray],
	{0.9/3.2}//,
    {0.1/3.2}//[fill=gray],
    {0.07/3.2}//,
    {0.1/3.2}//[fill=red!30!white],
    {0.22/3.2}//[fill=red!30!white],
    {0.5/3.2}/$s_4$/[fill=red!30!gray]
}/\width/[fill=lightgray],
{
	1/{\hbox to 0.7em{.\hss.\hss.}}/
}/\diff/,
{
	{0.4/3.2}//,
	{0.7/3.2}//,
	{0.9/3.2}/$s_5$/,
    {0.1/3.2}//,
    {0.07/3.2}//,
    {0.3/3.2}//,
    {0.4/3.2}//,
    {0.5/3.2}//[fill=red!30!gray]
}/\width/[fill=red!30!white],
{
	{0.2/3.2}//,
	{0.3/3.2}//,
	{0.25/3.2}//,
    {0.1/3.2}//,
    {0.07/3.2}//,
    {0.2/3.2}//,
    {0.5/3.2}//,
    {0.3/3.2}//
}/\width/[fill=red!30!white]
}

\node (mu) at (\diff+0.5*\width,-0.35) {$u_1$};
\node (mu) at (\diff+4.5*\width,-0.35) {$\bar{u}_1$};

\node (mu) at (3*\diff+5.5*\width,-0.35) {$v^{(2)}_1$};
\node (mu) at (3*\diff+6.5*\width,-0.35) {$v^{(2)}_2$};
\node (mu) at (3*\diff+7.5*\width,-0.35) {$v^{(2)}_3$};
\node (mu) at (3*\diff+8.5*\width,-0.35) {$u_2$};
\node (mu) at (3*\diff+11.5*\width,-0.35) {$\bar{u}_2$};

\draw [<->] (\diff,-0.8) -- (\diff+5*\width,-0.8) node[fill=white, pos=.5] {$m_1$};

\draw [<->] (3*\diff+5*\width,-0.8) -- (3*\diff+12*\width,-0.8) node[fill=white, pos=.5] {$m_2$};

\draw [<->] (0,-1.1) -- (4*\diff+12*\width,-1.1) node[fill=white, pos=.5] {$m'$};

\draw [<->] (0,-1.4) -- (4*\diff+14*\width,-1.4) node[fill=white, pos=.5] {$m$};



\end{tikzpicture}
}
		\caption{The situation after step \ref{alg:32approx:greedy} of \Cref{alg:32approx} with $1 \in I_{\exp}$ and $\set{2,3,4,5} \subseteq I_{\chp}$}
		\label{fig:non-preemptive:second_fillup}
	\end{figure}
	
	\Sref{alg:32approx:modify} The former solution is not feasible yet. That is due to a number of preemptively scheduled jobs on the one hand and the lack of some setup times on the other hand.
	The first step to obtain a non-preemptive solution is to consider each last job $j$ on a machine. If $j$ was scheduled integral, we keep it that way. If on the other hand $j$ is the first part of a split of step \ref{alg:32approx:L} or step \ref{alg:32approx:fill}, we remove $j$ from the machine and schedule the \emph{parent job} $\parent(j)$ instead. Also, we remove all other split pieces $j'$ with $\parent(j') = \parent(j)$ from the schedule and shift down the above jobs by $t_{j'}$. 
	Note that \emph{all} jobs are placed \emph{non-preemptively} now.
	The second step is to look upon the items scheduled in \ref{alg:32approx:greedy} \emph{in the order they were placed}. Every item $q$ that exceeds $T$ in the current schedule (and therefore is last on its machine) is moved to the machine of item $q'$ that was placed next. More precisely $q'$ and all jobs above $q'$ are shifted up by $s_i + t_q$ if $q$ is a job of class $i$ or by $q=s_i$ if $q$ is a setup. Accordingly $s_i$ followed by $q$ is placed at the free place below $q'$ if $q$ is a job of class $i$ or $q = s_i$ is placed at the free place below $q'$ if $q$ is a setup.
	In the analysis we will see that this builds a feasible schedule with makespan at most $\tfrac32T$.
	Have a look at \Cref{fig:non-preemptive:modify} to see an example result of \Cref{alg:32approx}. All previously preempted or $T$-crossing items are colored dark.
	
	\begin{figure}[h]
		\centering
		\scalebox{0.98}{\begin{tikzpicture}

\def\diff{0.3}
\def\width{0.44}

\draw (-0.2,0) node [left] {$0$} -- (4*\diff+14*\width+0.2, 0);
\draw (-0.2,3.2) -- (4*\diff+14*\width+0.2, 3.2) node [right] {$T$};

\DrawMachines{
{
	1/{\hbox to 0.7em{.\hss.\hss.}}/
}/\diff/,
{
	{1.8/3.2}/{$s_1$}/,
	{0.2/3.2}//,
	{0.3/3.2}//,
	{0.15/3.2}//,
	{0.4/3.2}//,
	{0.05/3.2}//,
	{0.4/3.2}//[fill=gray]
}/\width/[fill=lightgray],
{
	{1.8/3.2}/{$s_1$}/,
	{0.6/3.2}//,
	{0.3/3.2}//,
	{0.15/3.2}//,
	{0.25/3.2}//
}/\width/[fill=lightgray],
{
	{1.8/3.2}/{$s_1$}/,
	{0.3/3.2}//,
	{0.15/3.2}//,
	{0.35/3.2}//,
	{0.4/3.2}//,
	{0.05/3.2}//,
	{0.1/3.2}//,
	{0.65/3.2}//[fill=gray]
}/\width/[fill=lightgray],
{
	{1.8/3.2}/{$s_1$}/,
	{0.1/3.2}//,
	{0.05/3.2}//,
	{0.35/3.2}//,
	{0.2/3.2}//,
	{0.15/3.2}//[fill=gray]
}/\width/[fill=lightgray],
{
	{1.8/3.2}/{$s_1$}/,
	{0.1/3.2}//,
	{0.25/3.2}//,
	{0.07/3.2}//,
	{0.17/3.2}//,
    {0.4/3.2}/$s_3$/,
    {0.3/3.2}//
}/\width/[fill=lightgray],
{
	1/{\hbox to 0.7em{.\hss.\hss.}}/
}/\diff/,
{
	1/{\hbox to 0.7em{.\hss.\hss.}}/
}/\diff/,
{
	{1/3.2}/{$s_2$}/,
	{1.9/3.2}//,
    {0.2/3.2}//,
    {0.27/3.2}//[fill=gray]
}/\width/[fill=lightgray],
{
	{1/3.2}/{$s_2$}/,
	{1.65/3.2}//,
    {0.31/3.2}//,
    {0.57/3.2}//[fill=gray]
}/\width/[fill=lightgray],
{
	{1/3.2}/{$s_2$}/,
	{1.8/3.2}//
}/\width/[fill=lightgray],
{
	{1/3.2}/{$s_2$}/,
	{0.65/3.2}//,
	{0.7/3.2}//,
	{0.75/3.2}//,
	{0.9/3.2}//[fill=gray]
}/\width/[fill=lightgray],
{
	{1/3.2}/{$s_2$}/,
	{1.0/3.2}//,
	{0.9/3.2}//[fill=gray]
}/\width/[fill=lightgray],
{
	{1/3.2}/{$s_2$}/,
	{0.8/3.2}//,
	{0.7/3.2}//,
	{0.9/3.2}//[fill=gray]
}/\width/[fill=lightgray],
{
	{1/3.2}/{$s_2$}/,
	{0.9/3.2}//,
    {0.07/3.2}//,
    {0.4/3.2}/$s_3$/,
    {0.6/3.2}//[fill=gray],
    {0.1/3.2}//,
    {0.22/3.2}//
}/\width/[fill=lightgray],
{
	1/{\hbox to 0.7em{.\hss.\hss.}}/
}/\diff/,
{
	{0.5/3.2}/$s_4$/[fill=gray],
    {0.4/3.2}//,
	{0.7/3.2}//,
	{0.9/3.2}/$s_5$/,
    {0.1/3.2}//,
    {0.07/3.2}//,
    {0.3/3.2}//,
    {0.4/3.2}//
}/\width/[fill=lightgray],
{
	{0.9/3.2}/$s_5$/,
    {0.5/3.2}//[fill=gray],
	{0.2/3.2}//,
	{0.3/3.2}//,
	{0.25/3.2}//,
    {0.1/3.2}//,
    {0.07/3.2}//,
    {0.2/3.2}//,
    {0.5/3.2}//,
    {0.3/3.2}//
}/\width/[fill=lightgray]
}

\node (mu) at (\diff+0.5*\width,-0.35) {$u_1$};
\node (mu) at (\diff+4.5*\width,-0.35) {$\bar{u}_1$};

\node (mu) at (3*\diff+5.5*\width,-0.35) {$v^{(2)}_1$};
\node (mu) at (3*\diff+6.5*\width,-0.35) {$v^{(2)}_2$};
\node (mu) at (3*\diff+7.5*\width,-0.35) {$v^{(2)}_3$};
\node (mu) at (3*\diff+8.5*\width,-0.35) {$u_2$};
\node (mu) at (3*\diff+11.5*\width,-0.35) {$\bar{u}_2$};

\draw [<->] (\diff,-0.8) -- (\diff+5*\width,-0.8) node[fill=white, pos=.5] {$m_1$};

\draw [<->] (3*\diff+5*\width,-0.8) -- (3*\diff+12*\width,-0.8) node[fill=white, pos=.5] {$m_2$};

\draw [<->] (0,-1.1) -- (4*\diff+12*\width,-1.1) node[fill=white, pos=.5] {$m'$};

\draw [<->] (0,-1.4) -- (4*\diff+14*\width,-1.4) node[fill=white, pos=.5] {$m$};



\end{tikzpicture}
}
		\caption{The situation after step \ref{alg:32approx:modify} of \Cref{alg:32approx} with $1 \in I_{\exp}$ and $\set{2,3,4,5} \subseteq I_{\chp}$}
		\label{fig:non-preemptive:modify}
	\end{figure}
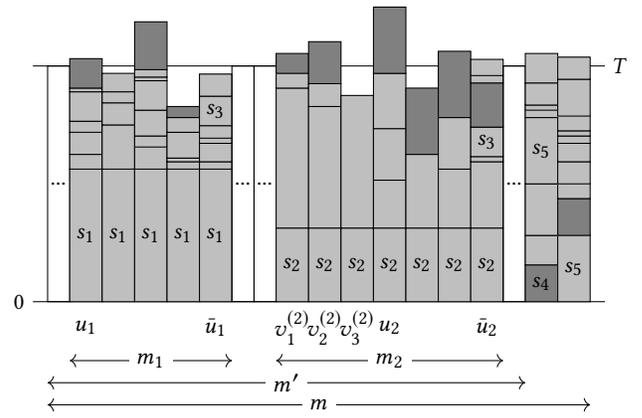
	
	\subsection{Analysis}
	
	We want to show the following theorem.
	
	\restatenonpreemptivealgorithmdecision
	
	\noindent We start with some preliminary work and obtain the following important notes.
	
	\begin{note}\label{non-preemptive:x_i}
		The remaining processing time for class $i$ after step \ref{alg:32approx:fill} is $x_i$, i.e. $P(C_i') = x_i$ for all $i \in [c]$ with $x_i \geq 0$. Furthermore, $x_i < 0$ implies that there is a time of $|x_i|$ left to schedule arbitrary jobs.
	\end{note}
	\begin{proof}
		We consider the situation right after step \ref{alg:32approx:L}.
		First we want to know the time $F_i$ that is left to schedule jobs of a class $i \in [c]$ without adding new setups. Each machine of class $i$ got a time of $T - s_i$ to schedule the jobs of $C_i \cap L$. Since there are $m_i$ of such machines we obtain $F_i = m_i(T - s_i) - P(C_i \cap L) \geq 0$.
		The remaining jobs of class $i$ are $C_i \setminus L$ and that gives us a total residual processing time of $P(C_i \setminus L)$. Let $C_i' \subseteq C_i \setminus L$ be the residual jobs \emph{after} step \ref{alg:32approx:fill}. Since $C_i = (C_i \setminus L) \cupdot (C_i \cap L)$, we obtain
		\begin{align*}
		P(C_i') = P(C_i \setminus L) - F_i
		&= P(C_i \setminus L) + P(C_i \cap L) - m_i(T - s_i)\\
		&= P(C_i) - m_i(T - s_i)
		\,\, = \, x_i
		\end{align*}
		if $x_i \geq 0$. So if $x_i < 0$, we have $P(C_i \setminus L) < F_i$ and that means there is a time of $F_i - P(C_i \setminus L) = |P(C_i \setminus L) - F_i| = |x_i|$ left to schedule any jobs.	
	\end{proof}
	
	\begin{note}\label{non-preemptive:additional_setups}
		A $T$-feasible schedule needs at least $m_i+1$ setups to place a class $i$ with $x_i > 0$.
	\end{note}
	\begin{proof}
		By its definition we know that $x_i > 0$ means $P(C_i) > m_i(T-s_i)$. So the obligatory $m_i$ machines (and setups) do not provide enough time to schedule all jobs of class $i$. Hence at least one additional setup must be placed.
	\end{proof}
	
	\begin{proof}[Proof of \Cref{non-preemptive:algorithm-decision}]
		\ref{lemma:non-preemptive:algorithm:T_check_false}. We show that $T \geq \OPT_{\op{nonp}}(I)$ implies that $mT \geq L_{\op{nonp}}$ and $m \geq m'$. So let $T \geq \OPT_{\op{nonp}}(I)$. Then there is a feasible schedule $\sigma$ with makespan $T$. Due to \Cref{non-preemptive:differentmachinesforL,non-preemptive:additional_setups,non-preemptive:minimal_number_of_machines} we get
		\begin{align*}
		mT \geq \load(\sigma) &\geq P(J) + \sum_{i:x_i \leq 0} m_i s_i + \sum_{i:x_i > 0} (m_i+1)s_i\\
		&= P(J) + \sum_{i=1}^c m_i s_i + \sum_{i:x_i > 0} s_i \,\, = \, L_{\op{nonp}}
		\end{align*}
		and also \Cref{non-preemptive:minimal_number_of_machines} proves that $m \geq \sum_{i=1}^c m_i = m'$.
		
		\ref{lemma:non-preemptive:algorithm:T_check_true}. Let $mT \geq L_{\op{nonp}}$ and $m \geq m'$. Note that step \ref{alg:32approx:L} uses $\sum_{i=1}^c m_i = m'$ machines and since $m \geq m'$, there are enough machines. Furthermore, one can easily confirm that the wrap templates $\omega^{(i)}$ suffice to schedule the wrap sequences $Q^{(i)}$, i.e. $S(\omega^{(i)}) \geq \load(Q^{(i)})$, but we still need to show that there is enough time to fill up with step \ref{alg:32approx:fill} and \ref{alg:32approx:greedy}.
		Instead of analyzing these steps separately we can use the values $x_i$ to find a much more intuitive formalization for both of them. Apparently in general the steps fill up the $m'$ obligatory machines to at least time $T$. In the worst case they fill them up to \emph{exactly} time $T$ since the residual load of $Q$, which is to be placed on the residual (and so far unused) $m-m'$ machines, gets maximized then.
		Due to \Cref{non-preemptive:x_i} and \Cref{non-preemptive:additional_setups}, this (worst case) residual load $R$ can be written as
		$R = \sum_{i: x_i < 0} x_i + \sum_{i: x_i > 0} (s_i + x_i)$ and we show that $R \leq (m - m')T$ as follows.
		\begin{align*}
			R &= \sum_{i = 1}^c x_i + \sum_{i:x_i > 0} s_i\\
			&= P(J)-\left(\sum_{i=1}^c m_i\right)T + \sum_{i=1}^c m_i s_i + \sum_{i: x_i > 0} s_i && \text{// Def. } x_i\\
			&= L_{\op{nonp}} - m'T\\
			&\leq (m - m')T && \text{// } mT \geq L_{\op{nonp}}
		\end{align*}
		So the $m - m'$ residual machines do provide enough time to schedule $R$. Hence, all load can be placed and it remains to show that step \ref{alg:32approx:modify} is correct.		
		Apparently step \ref{alg:32approx:L} and \ref{alg:32approx:fill} fill up machines to at most $T$. Step \ref{alg:32approx:greedy} fills machines to at most $\tfrac32T$. Now consider the situation right after step \ref{alg:32approx:greedy} and remark that the parent jobs $j$ of all preempted jobs of a class $i$ hold $t_j \leq s_i + t_j \leq \tfrac12T$ since $j \in J \setminus L$. The first modification of step \ref{alg:32approx:modify} is to replace preempted jobs (which are last on a machine with load at most $T$) with their integral parent jobs while removing all other child pieces. It is easy to see that the makespan can raise up to at most $T+\tfrac12T = \tfrac32T$. Also no jobs are preempted anymore since for each job piece $j$ there was a job piece $j'$ with $\parent(j') = \parent(j)$ such that $j'$ was last on a machine. The second and last modification repairs the lack of setups. Passing the $T$-crossing items to the next machine $u_+$ below the next job will give extra load of at most $\tfrac12T$ to machine $u_+$ (either passing a setup or a job with an additional setup). For $u_+$ there are two cases. If $u_+$ is not the last used machine, then $u_+$ passes away its last item too such that its load will be at most $\tfrac32T$ after all. If $u_+$ is the last used machine, it has a load of at most $T$ (otherwise this is a contradiction to $R \leq (m-m')T$) so it will end up with a load of at most $\tfrac32T$. The order of $Q$ guarantees that this movement/addition of setups will remove any lacks of setups such that the resulting schedule is feasible with a makespan of at most $\tfrac32T$.
		
		However, it remains to obtain the running time. The inclined reader will obtain that the primitive way of shifting up items on the considered next machines may require non-linear time, but this can actually be avoided as an implementation detail, with additional running time no more than $\oh(n)$. All other steps can be confirmed to run in linear time in a straightforward way.
	\end{proof}

\end{document}